%% file: main.tex
\documentclass{vldb}
\usepackage{graphicx}
\usepackage{balance}
\usepackage{color}
\usepackage[linesnumbered,ruled,vlined]{algorithm2e}
\usepackage{algpseudocode}
\usepackage{multirow}
\usepackage{times}
\usepackage{bm}
\usepackage{url}
\usepackage{hyperref}
\usepackage[show]{boxnotes}

\vldbTitle{Efficient Algorithms for Densest Subgraph Discovery}
\vldbAuthors{Yixiang Fang, Kaiqiang Yu, Reynold Cheng, Laks V.S. Lakshmanan, Xuemin Lin}
\vldbDOI{https://doi.org/10.14778/3342263.3342645}
\vldbVolume{12}
\vldbNumber{11}
\vldbYear{2019}

\newcommand{\squishlist}{
 \begin{list}{$\bullet$}
  {  \setlength{\itemsep}{0pt}
     \setlength{\parsep}{3pt}
     \setlength{\topsep}{3pt}
     \setlength{\partopsep}{0pt}
     \setlength{\leftmargin}{2em}
     \setlength{\labelwidth}{1.5em}
     \setlength{\labelsep}{0.5em}
} }

\newcommand{\squishend}{
  \end{list}
}
\newcommand{\eat}[1]{}

\newtheorem{theorem}{Theorem}
\newtheorem{lemma}{Lemma}

\newtheorem{example}{Example}

\newtheorem{definition}{Definition}
\newtheorem{problem}{Problem}

\begin{document}

\title{Efficient Algorithms for Densest Subgraph Discovery}

\numberofauthors{1}

\author{
\alignauthor
Yixiang Fang$^{\heartsuit\dag\star}$, Kaiqiang Yu$^{\ddag}$, Reynold Cheng$^{\ddag}$, Laks V.S. Lakshmanan$^{\S}$, Xuemin Lin$^{\dag\star}$\\
       \affaddr{$^{\heartsuit}$Guangzhou University, China, $^{\dag}$The University of New South Wales, Australia, $^{\star}$Zhejiang Lab, China, \\
       $^{\ddag}$The University of Hong Kong, China, $^{\S}$The University of British Columbia, Canada}\\
       \email{$^{\dag}$\{yixiang.fang@,lxue@cse.\}unsw.edu.au,$^{\ddag}$\{ky, ckcheng\}@cs.hku.hku,$^{\S}$laks@cs.ubc.ca}
}

\date{19 July 2019}

\maketitle

\newcommand{\tabincell}[2]{\begin{tabular}{@{}#1@{}}#2\end{tabular}}

\input{abstract}
\input{intro}\vspace{-0.1in}
\input{related}
\input{problem}
\input{basic}

\input{kcore}

\input{advanced}

\input{pds}

\input{experiments}

\input{conclusion}

\section*{Acknowledgments}
Reynold Cheng was supported by the Research Grants Council of Hong Kong (RGC Projects HKU 17229116, 106150091, and 17205115) and the University of Hong Kong (Projects 104004572, 102009508, and 104004129), and the Innovation and Technology Commission of Hong Kong (ITF project MRP/029/18).
Lakshmanan's research was supported in part by a discovery grant and a discovery accelerator supplement grant from NSERC (Canada).
Xuemin Lin was supported by 2019DH0ZX01, 2018YFB1003504, NSFC61232006, DP180103096, and DP170101628.
We would like to thank Dr. Charalampos E. Tsourakakis for bringing his KDD'15 paper to our attention.

\clearpage
\bibliographystyle{abbrv}
\bibliography{mds}

\clearpage

\input{appendix}

\end{document}

%% file: abstract.tex
\begin{abstract}

{\it Densest subgraph discovery} (DSD) is a fundamental problem in graph mining. It has been studied for decades, and is widely used in various areas, including network science, biological analysis, and graph databases. Given a graph $G$, DSD aims to find a subgraph $D$ of $G$ with the highest density (e.g., the number of edges over the number of vertices in $D$).  Because DSD is difficult to solve, we propose a new solution paradigm in this paper.  Our main observation is that the densest subgraph can be accurately found through a $k$-core (a kind of dense subgraph of $G$), with theoretical guarantees. Based on this intuition, we develop efficient exact and approximation solutions for DSD. Moreover, our solutions are able to find the densest subgraphs for a wide range of graph density definitions, including clique-based- and general pattern-based density. We have performed extensive experimental evaluation on both real and synthetic datasets. Our results show that our algorithms are up to four orders of magnitude faster than existing approaches.

\end{abstract}

%% file: intro.tex
\section{Introduction}
\label{sec:into}

Given a graph $G$ with $n$ vertices and $m$ edges, the {\it densest subgraph discovery} (DSD) is the problem of discovering a ``dense'' subgraph from $G$~\cite{chen2012dense,tsourakakis2013denser,fratkin2006motifcut,cohen2003reachability}.  For example, the densest subgraph of Figure~\ref{fig:intro}(a) is $S_1$, because its {\it edge-density}, or the average number of edges over the number of vertices in $S_1$, is the highest among all possible subgraphs of $G$.  The DSD problem is fundamental to graph mining \cite{tutorial}, and is widely used in network science, biological analysis, graph databases, and system optimization. In network science, for instance, the densest subgroups discovered can be used to find ``cohesive groups'' in social networks, for purposes of community detection \cite{chen2012dense,tsourakakis2013denser}. In biology, as another example, bioinformatics researchers have studied the use of DSD in identifying regulatory motifs in genomic DNA \cite{fratkin2006motifcut} and gene annotation graphs \cite{saha2010dense}.
In graph databases, the DSD is a building block for many graph algorithms, such as creating elegant index structures for reachability and distance queries \cite{cohen2003reachability,jin20093} and supporting graph visualization \cite{YZhang12,zhao2012large}.  In system optimization, DSD has been used in social piggybacking \cite{gionis2013piggybacking,tutorial}, which can be used to improve the throughput of social networking systems (e.g., Facebook).

At present, two variants of DSD have been proposed.  The first problem is to find the subgraph with the highest edge-density in $G$. In Figure~\ref{fig:intro}(a), for example, $S_1$ has the highest edge-density of 11/7 among all possible subgraphs of $G$.  Recently, researchers have studied DSD by defining density based on $h$-clique, which is a complete graph of $h$ vertices, with $h\geq2$. Figure~\ref{fig:intro}(b) shows a 3-clique (or ``triangle'') and a 4-clique. The goal of DSD is then to find the subgraph of $G$ that has the highest {\it $h$-clique-density} \cite{tsourakakis2015k,mitzenmacher2015scalable}, or the average number of $h$-cliques that a vertex participates in. In Figure~\ref{fig:intro}(a), subgraph $S_2$ has the highest ``3-clique-density'', in terms of number of triangles.
The DSD problem, based on clique-density, can be used for detecting larger near-cliques \cite{tsourakakis2015k,mitzenmacher2015scalable} (which can be used for communication network analysis and automatic test pattern generation \cite{clique}). The triangle-based densest subgraphs are useful for finding research groups in the DBLP network and clusters in senators' network on US bill voting \cite{tsourakakis2015k}, and discovering compact dense subgraphs from networks \cite{samusevich2016local}.
Note that an edge is a 2-clique, so edge-density is the 2-clique-density.

\begin{figure}
\hspace*{-.4cm}
\centering
\begin{tabular}{c c c}
  \begin{minipage}{4.3cm}
	\includegraphics[width=4.60cm]{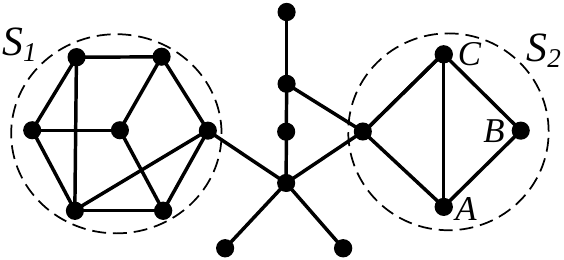}
  \end{minipage}
  & &
  \begin{minipage}{2.1cm}
	\includegraphics[width=2.1cm]{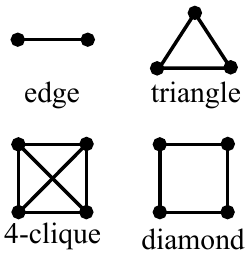}
  \end{minipage}
  \\
  (a) An example graph
  & &
  (b) Cliques and pattern
\end{tabular}
\caption{Illustrating the densest subgraphs.}
\label{fig:intro}
\end{figure}

Our main goal is to solve the DSD problem with respect to edge- and clique- densities.
This problem is technically challenging~\cite{goldberg1984finding,tsourakakis2015k,charikar2000greedy,zhao2012large}.
Existing DSD solutions, which often involve solving the maximum flow problem, are computationally expensive. For example, given a graph $G$ with $n$ vertices and $m$ edges, a well-known algorithm based on edge-density~\cite{goldberg1984finding} may incur a time complexity of ${\mathcal O}((mn + m^3)\log n)$, and is thus impractical for very large graphs. The $h$-clique-based DSD problem is even more complex~\cite{tsourakakis2015k,mitzenmacher2015scalable}. Moreover, our experiments show that existing DSD solutions cannot handle large graphs very well, and there is considerable room for developing faster solutions.

In this paper, our goal is to develop efficient algorithms for finding the subgraph with the highest edge- and $h$-clique-density. We leverage the $k$-core~\cite{md1983}, or the largest subgraph of graph $G$, where each vertex has at least $k$ neighbors.  We show that the densest subgraph (in terms of edge-density) is located in some $k$-cores, which are often much smaller than the entire graph $G$. For example, in Figure~\ref{fig:intro}(a), the subgraph $S_1$ is the 3-core, which is also the densest subgraph of $G$, w.r.t. edge-density. To solve DSD w.r.t. $h$-clique-density, we extend the $k$-core to the $k$-clique-core, or ($k$, $\Psi$)-core, which incorporates an $h$-clique $\Psi$ into the $k$-core definition.
Based on the cores, we develop efficient exact and approximation algorithms for finding the subgraphs with the highest edge-density and $h$-clique-density. Notably, this ``core-based solution'' achieves the same approximation ratio as the current state-of-the-art.

It is non-trivial to use ($k$, $\Psi$)-core to solve the DSD problem. Here we give an outline of this process. We denote by $k$ the {\it core number}. We first derive the lower and upper bounds on the $h$-clique-density for each ($k$, $\Psi$)-core. Based on these tight bounds, we can compute the upper and lower bounds of $\rho_{opt}$, which is the density of the densest subgraph, and further locate the densest subgraph w.r.t. an $h$-clique in some specific ($k$, $\Psi$)-cores. These ($k$, $\Psi$)-cores are often much smaller than the entire graph $G$, and thus we can directly compute the densest subgraph from these small cores, resulting in high efficiency.

Specifically, to compute the exact densest subgraph $D$, we first locate $D$ in a specific ($k$, $\Psi$)-core. Then, we build a flow network on this core, and find $D$ by solving the maximum flow problem using binary search. During the binary search, whenever we obtain a larger lower bound of $\rho_{opt}$, we can further locate $D$ in another core with higher core number and build an even smaller flow network to compute $D$. The binary search process stops when we have found $D$.
We further show that the  ($k_{\max}$, $\Psi$)-core, which is a ($k$, $\Psi$)-core with $k$ attaining the maximum value, is a good approximation to the densest subgraph, with theoretical guarantees. To find the ($k_{\max}$, $\Psi$)-core, a straightforward method is to perform core decomposition, which computes all the $(k,\Psi)$-cores in an incremental manner. This is costly and unnecessary because we only need the $(k_{\max}, \Psi)$-core, rather than all the $(k, \Psi)$-cores. We thus develop another efficient method that extracts the ($k_{\max}$, $\Psi$)-core without computing all the ($k$, $\Psi$)-cores. This solution finds the ($k_{\max}$, $\Psi$)-core from a set of small subgraphs induced by vertices with high degrees, and thus yields better performance.

In addition, we generalize the notion of density to allow arbitrary ``pattern graphs'' (e.g., the diamond pattern in Figure~\ref{fig:intro}(b)), and propose {\it pattern-density} to measure the average number of {\it patterns} in which a vertex participates. We further extend $k$-clique-core to $k$-pattern-core, and show that our solutions above can be smoothly adapted to finding the densest subgraph w.r.t. pattern-density.

We have performed extensive experiments to evaluate our approaches. On both real and synthetic graph datasets ranging from a few thousand to millions of vertices and edges, our new solutions show high efficiency. For example, our core-based exact algorithm, namely {\tt CoreExact}, is up to four orders of magnitude faster than the state-of-the-art exact DSD solution.  Our best approximation algorithm, called {\tt CoreApp}, is up to two orders of magnitude faster than the existing approximation solution. We further perform experiments to find pattern-based densest subgraphs and our results again confirm the superiority of our core-based approaches.

%{\bf Contributions.} In summary, our main contributions in this paper are as follows:
{\bf Contributions.} In summary, our main contributions are:
\squishlist

\item We present a new perspective on solving the DSD problem. Particularly, we propose the ($k$, $\Psi$)-core by incorporating an $h$-clique $\Psi$ where $h\geq2$ (Section~\ref{sec:kcore}). We further establish the lower and upper bounds of densities for ($k$, $\Psi$)-cores.

\item Based on the ($k$, $\Psi$)-cores, we develop fast exact and approximation DSD algorithms w.r.t. $h$-clique-density (Section~\ref{sec:advanced}).

\item We generalize $h$-clique-density to pattern-density and adapt our solutions to solving DSD w.r.t. pattern-density (Section~\ref{sec:pds}).

\item  We conduct extensive experiments on ten real datasets and three synthetic datasets to evaluate our algorithms. The results reveal that our proposed DSD algorithms are several orders of magnitude faster than existing ones (Section~\ref{sec:exp}).
\squishend

\textbf{Organization.} We review the related work in Section~\ref{sec:related}. The DSD problem is stated in Section~\ref{sec:problem}. In Sections 4-6 we present different DSD solutions. In Section~\ref{sec:pds}, we extend our algorithms for finding densest subgraphs for general patterns. We report experimental results in Section~\ref{sec:exp}, and conclude in Section~\ref{sec:conclusion}.
{\sl Due to space limitation, for some lemmas, we do not show the complete proof in this paper; instead, we give the proof sketch and show the complete proof in the technical report~\cite{fullVersion}.}

%% file: related.tex
\section{Related Work}
\label{sec:related}

The problem of dense subgraph computation has been extensively studied \cite{tutorial,Sahu:2017,chen2012dense,tsourakakis2013denser}. In the following, we review existing works that are highly related to our DSD problem.

{\bf Edge-based Densest Subgraph (EDS).} The edge-density of an undirected graph $G(V,E)$ is defined as $\frac{m}{n}$ with $n$=$|V|$ and $m$=$|E|$.  The EDS problem aims to find a subgraph such that its edge-density is the highest among all subgraphs.  This problem can be addressed by solving a parametric maximum-flow problem~\cite{goldberg1984finding,gallo1989fast}.  A typical variant of EDS is to impose a size restriction on the returned subgraph, i.e., finding a subgraph of  up to a given number of vertices whose density is the highest. This problem is NP-hard~\cite{asahiro2000greedily,asahiro2002complexity}.
Another version of EDS, called optimal quasi-clique~\cite{tsourakakis2013denser}, extracts a subgraph, which is more compact, with a smaller diameter than the EDS. Again, this variant is NP-hard~\cite{thesis2013}.
Qin et al.\ developed solutions for finding the top-$k$ locally densest subgraphs~\cite{qin2015locally}.
The EDS problem on evolving graphs is studied in~\cite{evolving2015}.
In \cite{tatti2015density,danisch2017large}, the edge-density-based graph decomposition is extensively studied.
Kannan and Vinay~\cite{directed1999} modeled the density on directed graphs, and then studied the DSD problem on directed graphs~\cite{charikar2000greedy}.

In general, exact EDS solutions work well for small graphs, but they perform poorly for large graphs.
Thus, researchers have developed approximation algorithms, in order to achieve higher efficiency. In~\cite{charikar2000greedy}, Charikar et al. proposed a greedy 0.5-approximation algorithm for solving the EDS problem.
Bahmani et al.~\cite{bahmani2012densest} devised a $1/(2+2\varepsilon$)-approximation algorithm under the streaming model, which takes ${\mathcal O}(m\frac{\log(n)}{\varepsilon})$ time.  The densest subgraph on directed graphs can also be computed by an approximation algorithm~\cite{khuller2009finding}.

Our solution is based on computing $k$-cores, which can then be used to find the EDS. Based on this intuition, we have developed exact and approximation algorithms, and show using extensive experiments that they are much faster than existing EDS solutions.

{\bf $h$-clique Densest Subgraph (CDS).} In \cite{tsourakakis2015k,mitzenmacher2015scalable}, Tsourakakis et al. modeled graph density based on $h$-cliques, and studied the $h$-clique densest subgraph (CDS) problem. It generalizes the EDS problem, which is a special case of CDS for $h$=2. They found that the 3-clique densest subgraphs (a 3-clique is a triangle) help identify cohesive researcher groups in a bibliographical network, as well as clusters of republicans in the network of US senators. Recently, a variant based on the 3-clique, called top-$k$ local triangle-densest subgraphs discovery, has been investigated \cite{samusevich2016local}.

There are four key differences between existing works \cite{tsourakakis2015k,mitzenmacher2015scalable} and our work. (1) Our algorithms, based on ($k$, $\Psi$)-cores where $\Psi$ is an $h$-clique, are substantially different from existing CDS solutions. (2) Whereas \cite{tsourakakis2015k,mitzenmacher2015scalable} can only handle $h$-cliques, our work supports any general pattern (e.g., 4-vertex subgraph \cite{jha2015path,wang2018vertex}).
(3) The approximation algorithm in \cite{mitzenmacher2015scalable} is a randomized algorithm which has a failure probability to obtain an approximation solution, while our core-based approximation algorithms are deterministic algorithms.
(4) Our empirical evaluation shows that our algorithms significantly outperform previous exact algorithms \cite{tsourakakis2015k,mitzenmacher2015scalable} and deterministic approximation algorithm \cite{tsourakakis2015k}.

{\bf Other Dense Subgraphs.} Recently, many other dense subgraph models \cite{fang2019survey}, such as $k$-core~\cite{kcore2003,lu2016h,peng2018efficient,Fang2016,fang2017effective,fang2017VLDBJ,fang2018spatial,fang2018effective,wang2018efficient,chen2018exploring}, $k$-truss~\cite{cohen2008trusses,Huang2014,wu2015robust,huang2016truss,huang2017attribute},
$k$-($r$, $s$) nucleus \cite{sariyuce2015finding,sariyuce2016fast,sariyuce2017nucleus,sariyuce2018local} (a generalization of $k$-core and $k$-truss),
$k$-clique \cite{cui2013online,hu2019discovering},
$k$-edge connected components \cite{hu2016querying,hu2017minimal}.
and $k$-plexes~\cite{seidman1978graph}, have also been explored.
However, these dense subgraphs are different from EDS and CDS, which attain the highest edge-density and clique-density.

%% file: problem.tex
\section{Problem Definition}
\label{sec:problem}

\textbf{Data model.}
In this paper, we consider an undirected, unweighted, and simple graph $G(V, E)$ with vertex set $V$ and edge set $E$, where
$n$=$|V|$ and $m$=$|E|$. The degree of a vertex $v$ in $G$, denoted by $deg_G(v)$, is the number of its neighbors, and we denote the maximum degree by $d$. Table~\ref{tab:notation} summarizes all the notations frequently used in this paper. Next, we first introduce two prominent notions of density that were employed in the DSD literature, namely edge-density and $h$-clique-density.

\begin{definition}[Edge-density~\cite{goldberg1984finding,gallo1989fast}]\label{def:edge-density}
Given a graph $G$ $(V,$ $E)$, its edge-density is $\tau(G)$= $\frac{{|E|}}{{|V|}}$.
\end{definition}

%Before introducing $h$-clique-density, we introduce some notions.

\begin{definition}[Clique instance]
\label{def:instance}
Given a graph $G(V,$ $E)$ and an integer $h$$\geq$$2$, we say a set of $h$ vertices, $S$$\in$$V$, is an $h$-clique instance, if each pair of vertices $u,v\in S$ is connected by an edge.
%$h$-clique $\Psi(V_\Psi, E_\Psi)$, a subgraph $S(V_S, E_S)\subseteq G$ is a clique instance of $\Psi$, if $S$ is isomorphic to $\Psi$.
\end{definition}

\begin{definition}[Clique-degree]
\label{def:patterndegree}
Given a graph $G(V,E)$ and an $h$-clique $\Psi$, the clique-degree of a vertex $v$ in $G$,
or $deg_{G}(v$, $\Psi)$, is the number of clique instances containing $v$.
\end{definition}

Note that for each of these instances, we do not consider permutations of vertices. For example, let $\Psi$ be the triangle (i.e., 3-clique). Then in Figure~\ref{fig:intro}(a), the subgraph $S_2$ contains two clique instances of $\Psi$, which share an edge. The clique-degrees of vertices $A$, $B$, and $C$ are 2, 1, and 2 respectively.

\begin{definition}[$h$-clique-density \cite{tsourakakis2015k}]
\label{def:cliquedensity}
Given a graph $G$ $(V, E)$ and an $h$-clique $\Psi(V_\Psi, E_\Psi)$ with $h$$\geq$2, the $h$-clique-density of $G$ w.r.t. $\Psi$ is
\begin{equation}
\small
\rho(G, \Psi)=\frac{\mu(G, \Psi)}{|V|},
\end{equation}
where $\mu(G, \Psi)$ is the number of clique instances of $\Psi$ in $G$.
\end{definition}

The densest subgraph of $G$ w.r.t. edge-density (resp., $h$-clique-density), i.e., \emph{EDS} \cite{goldberg1984finding} (resp., \emph{CDS} \cite{tsourakakis2015k,mitzenmacher2015scalable}), is the subgraph $D$= ($V_D$, $E_D$) of $G$ whose edge-density (resp., $h$-clique-density) is the highest. Clearly, if the $h$-clique is a single edge (i.e., $h$=2), the $h$-clique-density reduces to edge-density. For ease of exposition, in the following we simply focus on the $h$-clique-density with $h\geq 2$. We use the term CDS when we refer to the DSD problem using the edge-, or $h$-clique-based density. Where necessary, we make the distinction between EDS and CDS.

Now we formally introduce the problem studied in this paper.

\begin{problem}[CDS Problem \cite{tsourakakis2015k,mitzenmacher2015scalable}]
\label{prob:PDS}
Given a graph $G(V$ ,$E)$ and an $h$-clique $\Psi(V_\Psi$, $E_\Psi)$ ($h\geq 2$), return the subgraph $D$ of $G(V, E)$,
whose $h$-clique-density $\rho(D, \Psi)$ is the highest.
\end{problem}

We denote the $h$-clique-density of $D$ by $\rho_{opt}$, i.e., $\rho_{opt}$=$\rho(D, \Psi)$, where $D$ is the CDS.
For the graph $G$ of Figure~\ref{fig:intro}(a), if we let $\Psi$ be the single edge, we will return $S_1$ as the densest subgraph; if we let $\Psi$ be the 3-clique (i.e., triangle), then, $S_2$ is the subgraph with the highest 3-clique-density.

\renewcommand\arraystretch{1.12}
\begin{table}[]
  %\small
  \scriptsize
  \centering
  \caption {Notations and meanings.}\label{tab:notation}
  \begin{tabular}{c|l}
    \hline
         {\bf Notation}         & {\bf Meaning}\\
    \hline\hline
         $G(V, E)$              & a graph with vertex set $V$ and edge set $E$\\
    \hline
         $n$, $m$               & $n$=$|V|$, $m$=$|E|$\\
    \hline
         $deg_G(v)$             & (classical edge-based) degree of vertex $v$ in $G$\\
    \hline
         $d$                    & the maximum (classical edge-based) degree of $G$\\
    \hline
         $G[T]$                 & a subgraph of $G$ induced by vertex set $T$\\
    \hline
         $\Psi(V_\Psi, E_\Psi)$ & an $h$-clique (vertex set: $V_\Psi$, edge set $E_\Psi$)\\
    \hline
  %       $\psi$                 & a clique instance of $\Psi$ in the graph\\
  %  \hline
         $deg_G(v, \Psi)$       & clique-degree of vertex $v$ in $G$ w.r.t. $\Psi$\\
    \hline
         $\mu(S, \Psi)$         & number of clique instances of $\Psi$ in the graph $S$\\
    \hline
         $\rho(G, \Psi)$        & $h$-clique-density of graph $G$ w.r.t. an $h$-clique $\Psi$\\
    \hline
         $D(V_D, E_D)$          & the CDS whose $h$-clique-density is $\rho_{opt}$\\
%         $D(V_D, E_D)$          & $h$-clique-density-based densest subgraph (CDS)\\
%    \hline
%         $\rho_{opt}$           & $h$-clique-density of the CDS\\
    \hline
         ${\mathcal F}(V_{\mathcal F}, E_{\mathcal F})$          & a flow network with node set $V_{\mathcal F}$ and edge set $E_{\mathcal F}$\\
    \hline
%         $\mathcal A$, $\mathcal B$ & two sets of nodes () in of ${\mathcal F}$\\
%    \hline
  \end{tabular}
\end{table}

%Note that in the worst case, the number of pattern instances could be exponential in the number of vertices in $\Psi$, as any $|V_\Psi|$ vertices of $G$ may potentially form a pattern instance. In practice, a pattern is often small graph~\cite{nature2003,leskovec2006,lai2016scalable}, which enables us to develop efficient PDS algorithms.
%However, in practice the size of the pattern is often small, which motivates us to develop efficient algorithms for the PDS problem.

%% file: basic.tex
\section{Existing Approaches}
\label{sec:baseline}

%As mentioned in Section~\ref{sec:related}, there are a few exact DSD algorithms~\cite{goldberg1984finding,charikar2000greedy,tsourakakis2015k}. However, they are inefficient for large graphs, so some faster approximations have been developed~\cite{charikar2000greedy,tsourakakis2015k}.
In this section, we review existing algorithms for the EDS and CDS problems, and then discuss their limitations.

\input{basicExact}
\input{basicApp}
\input{basicDiscuss}

%% file: basicExact.tex
\subsection{The Exact Method}
\label{sec:basicExact}

Generally, the algorithms for finding exact EDS and CDS~\cite{goldberg1984finding,tsourakakis2015k,mitzenmacher2015scalable} follow the same framework by solving a maximum flow problem using binary search.
A flow network~\cite{flow} is a directed graph ${\mathcal F}(V_{\mathcal F}, E_{\mathcal F})$, where there is a source node\footnote{We use ``node'' to mean ``flow network node'' in this paper.} $s$, a sink node $t$, and some intermediate nodes; each edge has a capacity and the amount of flow on an edge cannot exceed the capacity of the edge.
The maximum flow of a flow network equals the capacity of its minimum st-cut, (${\mathcal S}$, ${\mathcal T}$), which partitions the node set $V_{\mathcal F}$ into two disjoint sets, $\mathcal S$ and $\mathcal T$, such that $s\in\mathcal S$ and $t\in\mathcal T$.

\begin{algorithm}{}
\small
\caption{The algorithm: {\tt Exact}.}
\label{alg:basicExact}
\KwIn{$G(V,E)$, $\Psi(V_\Psi,E_\Psi)$;}
\KwOut{The CDS $D(V_D, E_D)$;}
initialize $l \gets 0$, $u\gets\mathop {\max }\limits_{v \in V} deg_G(v, \Psi)$\;
initialize $\Lambda\gets$all the instances of ($h$--1)-clique in $G$, $D\gets\emptyset$\;
\While{$ u-l\geq \frac{1}{n(n-1)} $}{
    $\alpha\gets \frac{l+u}{2}$\;
    $V_{\mathcal F}\gets \{s\}\cup V\cup\Lambda\cup\{t\}$\tcp*{build a flow network}
    \For{$each$ $vertex$ $v\in V$}{
        add an edge $s$$\rightarrow$$v$ with capacity $deg_G(v, \Psi)$\;
        add an edge $v$$\rightarrow$$t$ with capacity $\alpha|V_\Psi|$\;

    }
    \For{$each$ ($h$--1)-$clique$ $\psi\in\Lambda$}{
        \For{$each$ $vertex$ $v\in\psi$}{
            add an edge $\psi$$\rightarrow$$v$ with capacity $+\infty$\;
        }
    }
    \For{$each$ ($h$--1)-$clique$ $\psi\in\Lambda$}{
        \For{$each$ $vertex$ $v\in V$}{
            \If {$\psi$ and $v$ form an $h$-clique}{
                add an edge $v$$\rightarrow$$\psi$ with capacity 1\;
            }
        }
    }
    find minimum st-cut ($\mathcal S$, $\mathcal T$) from the flow network $\mathcal F(V_{\mathcal F}, E_{\mathcal F})$\;
    \textbf{if} $\mathcal S$=$\{s\}$ \textbf{then} $u \gets \alpha$\;
    \textbf{else} \text{ } \text{ } \text{ } \text{ } \text{ }$l \gets \alpha$,  $D\gets$ the subgraph induced by $\mathcal S\backslash \{s\}$\;
}
\Return $D$\;
\end{algorithm}

We present the state-of-the-art algorithm from~\cite{mitzenmacher2015scalable} in Algorithm~\ref{alg:basicExact}, where the input is a graph $G$ and an $h$-clique $\Psi$.
First, it initializes lower and upper bounds of $\rho_{opt}$ and collects all the instances of ($h$--1)-clique (lines 1-2).
Then, it finds $D$ by using binary search (lines 3-18).
Specifically, in each binary search (lines 4-18), it tries to find a subgraph with density larger than a guessed value $\alpha$, by computing the minimum st-cut using Gusfield's algorithm~\cite{ahuja1994improved} in a flow network $\mathcal F(V_{\mathcal F}, E_{\mathcal F})$.
To build $\mathcal F(V_{\mathcal F}, E_{\mathcal F})$, it first creates a node set $V_{\mathcal F}$ (line 5), and then links its nodes by directed edges with different capacities (lines 6-15).
The binary search stops when the gap between the upper and lower bounds of $\alpha$ is less than $\frac{1}{n(n-1)}$.
We denote this algorithm by {\tt Exact}.

Note that if $\Psi$ is the single edge, the flow network $\mathcal F(V_{\mathcal F}$, $E_{\mathcal F})$ can be simplified such that \cite{goldberg1984finding}:
$V_{\mathcal F}$=$\{s\}\cup V\cup\{t\}$,
and for each vertex $v\in$$G$, there is a directed edge from $s$ to $v$ with capacity $m$ and a directed edge from $v$ to $t$ with capacity $m$+2$\alpha$-$deg_G(v)$;
for each edge ($v$, $u$)$\in$$G$, there is a directed edge from $u$ to $v$ with capacity $1$ and a directed edge from $v$ to $u$ with capacity $1$.

\begin{example}
\label{eg:exact-flow}
{\em
Let $\Psi$ be the triangle and $G$ be the graph in Figure~\ref{fig:exact-flow}(a).
The graph contains 4 edges (see Figure~\ref{fig:exact-flow}(b)).
By Algorithm~\ref{alg:basicExact}, we construct the flow network, as depicted in Figure~\ref{fig:exact-flow}(c), where the value on each edge denotes its capacity.}
\qed
\end{example}

\begin{figure}[]
	\centering
	\includegraphics[width=0.88\linewidth]{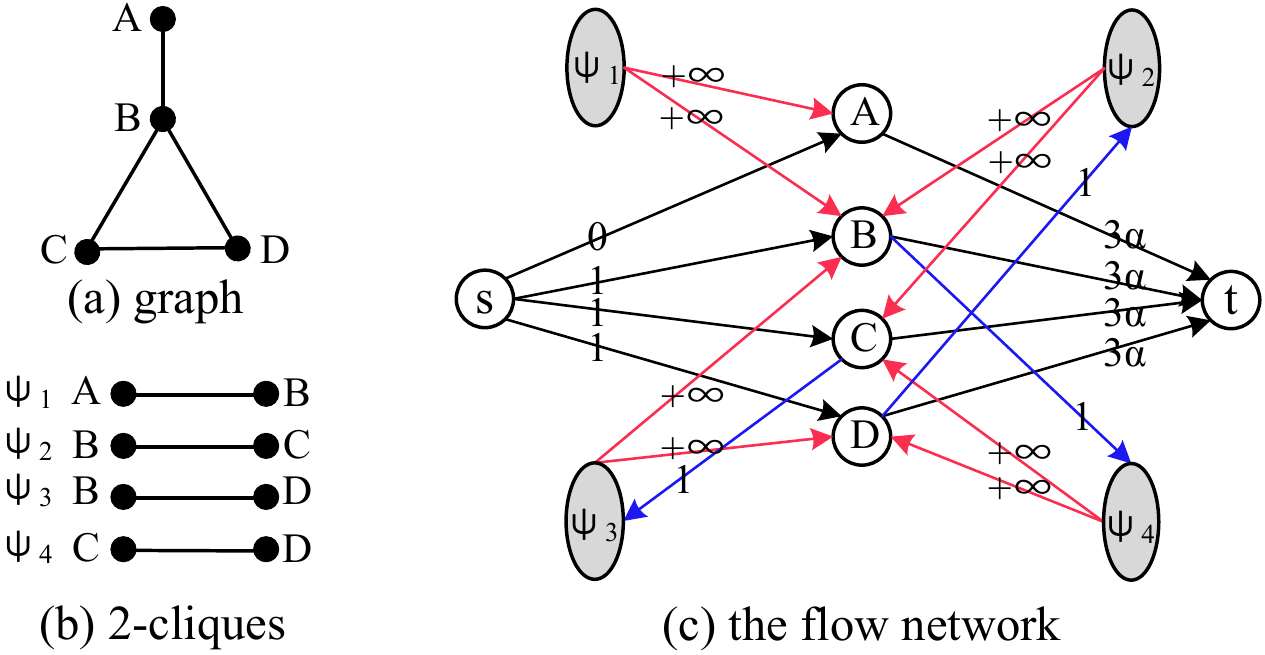}
    \vspace{-0.10in}
	\caption{Illustrating the flow network ($\Psi$ is a triangle).}
	\label{fig:exact-flow}
\end{figure}

\vspace{-0.1in}

\begin{lemma}
\label{lemma:exactTime}
Given a graph $G(V$,$E)$ and an $h$-clique $\Psi(V_\Psi$,$E_\Psi)$, {\tt Exact} takes
${\mathcal O}$$\left( {n\cdot{d-1 \choose h-1}}+ (n|\Lambda|+\min {{(n,|\Lambda|)}^3})\log n\right)$ time and $\mathcal{O}$ $(n+|\Lambda|)$ space, where $\Lambda$ is set of ($h-$1)-clique instances in $G$ \cite{tsourakakis2015k}.
\end{lemma}

\vspace{-0.05in}
\textsc{Proof sketch:} In the worst case, we will consider $n\cdot{d-1\choose h-1}$ $h$-clique instances, and each binary search of {\tt Exact} takes ${\mathcal O}(n\cdot|\Lambda |$+$\min {{(n,|\Lambda |)}^3})$ time. These are the main cost of {\tt Exact}.\qed

In practice, $h$ is often small and the number of clique instances, $|\Lambda|$, is often much larger than the number $n$ of vertices, so the second summand dominates the overall computational cost.

%% file: basicApp.tex
\subsection{The Approximation Method}
\label{sec:peel}

The approximation method of computing the EDS~\cite{charikar2000greedy} and CDS~\cite{tsourakakis2015k} follows the {\it peeling paradigm} and achieves an approximation ratio of $\frac{1}{|V_\Psi |}$. Here, the approximation ratio is the ratio of the $h$-clique-density of subgraph returned, over $\rho_{opt}$, which is at most 1.0.
Specifically, given a graph $G$ of $n$ vertices, it works in $n$ rounds.
In each round, it removes the vertex that participates in the minimum number of $h$-cliques, and recomputes the density of the residual graph. Finally, the subgraph of the largest $h$-clique-density is returned. Algorithm~\ref{alg:peelAlgo} outlines the steps. Because the algorithm removes vertices one by one, we call it {\tt PeelApp}.
%~\footnote{For Step 2 of Algorithm~\ref{alg:peelAlgo} (i.e., compute a vertex $v$'s pattern-degree), we enumerate all pattern instances containing $v$. For some common patterns (e.g., stars and cliques), this step can be performed faster (to be discussed in Section~\ref{sec:kcoreSpec}), and we use these better solutions in our experiments.}

\begin{algorithm}
\small
\caption{The algorithm: {\tt PeelApp}.}
\label{alg:peelAlgo}
\KwIn{$G(V, E)$, $\Psi(V_\Psi, E_\Psi)$;}
\KwOut{A subgraph $S^*$;}
initialize $S\gets G$, $S^*\gets\emptyset$\;
compute the clique-degree for each vertex of $G$\;
\While{$S\ne\emptyset$}{
   $v\gets$ the vertex with the minimum clique-degree in $S$\;
   $S\gets$ remove the vertex $v$ from $S$\;
   \textbf{if} $\rho(S, \Psi)\textgreater\rho(S^*, \Psi)$ \textbf{then} $S^*\gets S$\;
}
\Return $S^*$\;
\end{algorithm}

\begin{lemma}
\label{lemma:peelAppTime}
Given a graph $G$ and an $h$-clique $\Psi(V_\Psi,$$E_\Psi)$, then {\tt PeelApp} takes ${\mathcal O}\left( {n \cdot {{d-1} \choose {h-1}}}\right)$ time and ${\mathcal O}\left(m\right)$ space \cite{tsourakakis2015k}.
\end{lemma}

\vspace{-0.1in}
\textsc{Proof sketch:} The main time cost comes from enumerating clique instances, whose number is $n\cdot{d-1\choose h-1}$ in the worst case.\qed

%% file: basicDiscuss.tex
\subsection{Limitations of Existing Methods}
\label{sec:basicDiscuss}

From the above lemmas, we see that while {\tt PeelApp} is faster than {\tt Exact}, it also sacrifices some accuracy. For example, when $\Psi$ is an edge, {\tt Exact} finds the exact EDS in ${\mathcal O}((mn + m^3)\log n)$ time, while {\tt PeelApp} returns a subgraph with 0.5-approximation ratio in linear time, i.e., ${\mathcal O} (m)$.  Both solutions can be inefficient on larger graphs with more complex cliques.  We found that {\tt Exact} suffers from several problems: (1) the initial lower and upper bounds of $\alpha$ are not very tight; (2) the size of the flow network can be large when the graph is large and there are many clique instances of $\Psi$; and (3) the flow network $\mathcal F$ is always built on the entire graph $G$ in each iteration, while the CDS is often in a small subgraph of $G$. The {\tt PeelApp} algorithm also involves a lot of unnecessary computation: for the first few iterations, the graph contains many vertices with lower clique-degrees, which are unlikely to be in the CDS, but {\tt PeelApp} still computes the $h$-clique-density.  As shown in our experiments later, on a moderate-size graph ($n$$\approx$26K and $m$$\approx$100K), {\tt Exact} takes more than 5 days to find the densest subgraphs for 6-clique; on a million-scale graph ($n$$\approx$19M and $m$$\approx$298M), {\tt PeelApp} takes more than 2 days to find the CDS for 6-clique. Thus, there is room for improving their efficiency.

We next propose a core-based approach for locating a CDS, by quickly converging on smaller dense subgraphs that contain the CDS. To make our approach applicable for processing all the $h$-clique-density definitions ($h$$\geq$2), we lift the notion of $k$-cores to $k$-clique-cores and study how to exploit them in the DSD solution.

%% file: kcore.tex
\section{The Clique-Based Cores}
\label{sec:kcore}

We now study the $k$-clique-core, or ($k$, $\Psi$)-core, which is a generalization of the classical $k$-core~\cite{md1983,kcore2003} for an $h$-clique $\Psi$ (Section~\ref{sec:kcoreIntro}). As we will show, ($k$, $\Psi$)-cores are useful in locating the CDS in both exact and approximation algorithms. We then establish upper and lower bounds on the clique-density of ($k$, $\Psi$)-cores  (Section~\ref{sec:coredensity}), present efficient algorithms for decomposing ($k$, $\Psi$)-cores (Section~\ref{sec:decompose}), and give some discussions (Section~\ref{sec:kcoreDiscuss}) .

\input{kcoreIntro}

\input{coredensity}

\input{decompose}

\input{kcoreDiscuss}

%% file: kcoreIntro.tex
\subsection{$k$-core and ($k$, $\Psi$)-core}
\label{sec:kcoreIntro}

We first review the definition of $k$-core.

\begin{definition}[$k$-core~\cite{md1983,kcore2003}]
\label{def:kcore}
Given a graph $G$ and an integer $k$ ($k\geq 0$), the $k$-core,
denoted by $\mathcal H_k$, is the largest subgraph of $G$, such that $\forall v \in {\mathcal H}_k$, $deg_{{\mathcal H}_k}(v) \geq k$.
\end{definition}

We say that ${\mathcal H}_k$ has order $k$. The \emph{core number} of a vertex $v\in V$ is defined as the highest order of a $k$-core that contains $v$. In other words, a $k$-core is the largest subgraph induced by vertices whose core numbers are at least $k$.
A $k$-core has some interesting properties~\cite{kcore2003}:
(1) $k$-cores are ``nested'': given two nonnegative integers $i$ and $j$, if $i<j$, then ${\mathcal H}_j \subseteq {\mathcal H}_i$;
(2) a $k$-core may not be connected; and
(3) computing core numbers of all the vertices in a graph, known as $k$-core decomposition, can be done in linear time.

\begin{figure}[ht]
    \hspace*{-.2cm}
	\centering
    \begin{tabular}{c c}
        \begin{minipage}{3.45cm}
	       \includegraphics[width=3.45cm]{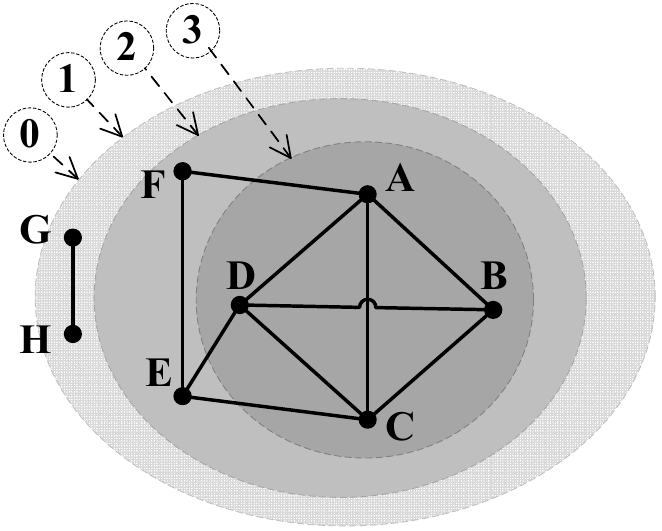}
        \end{minipage}
        &
        \begin{minipage}{3.45cm}
	       \includegraphics[width=3.45cm]{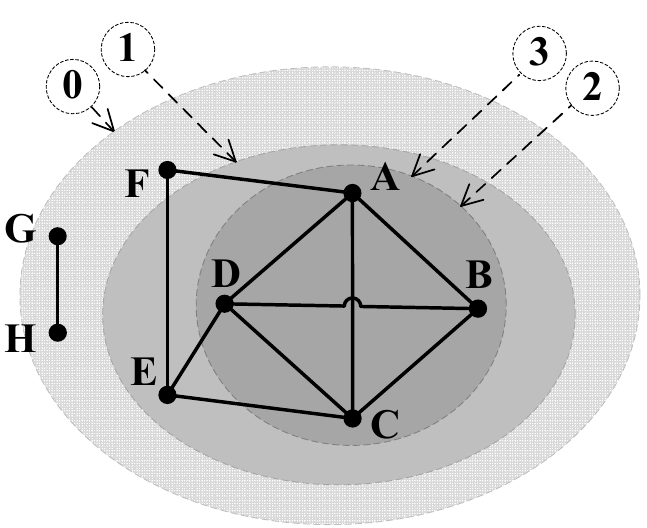}
        \end{minipage}
        \\
        (a) $k$-cores
        &
        (b) ($k$, $\Psi$)-cores
    \end{tabular}
    \vspace{-0.10in}
	\caption{$k$-core, and ($k$, $\Psi$)-core ($\Psi$ is a triangle). }
	\label{fig:kcoreGraph}
\end{figure}

\begin{example}
\label{eg:kcore}
{\em
Figure~\ref{fig:kcoreGraph}(a) depicts a graph of 8 vertices and its $k$-cores.
The number $k$ in each ellipse indicates the $k$-core contained in that ellipse. For instance, the subgraph induced by $\{A,B,C,$ $D\}$ is the 3-core, and the entire graph is both the 0-core and 1-core, which consist of two connected components. }
\qed
\end{example}

\begin{definition}[($k$, $\Psi$)-core]
\label{def:kpsicore}
Given a graph $G$, an integer $k$ ($k$$\geq$0), and an $h$-clique $\Psi$, the ($k$, $\Psi$)-core, denoted by ${\mathcal R}_k$,
is the largest subgraph of $G$ such that $\forall v \in {\mathcal R}_k$, $deg_{{\mathcal R}_k}(v, \Psi) \geq k$.
\end{definition}

Similar to $k$-cores, we say that ${\mathcal R}_k$ has order $k$. The \emph{clique-core number} of a vertex $v\in V$, $core_G(v, \Psi)$, is then the highest order of a ($k$, $\Psi$)-core containing $v$. We denote the maximum clique-core number by $k_{\max}$, where the underlying clique $\Psi$ is understood from the context.
Given a clique $\Psi$, a ($k$, $\Psi$)-core also has the following properties:
(1) ($k$, $\Psi$)-cores are ``nested'': given two nonnegative integers $i$ and $j$, if $i<j$, then ${\mathcal R}_j \subseteq {\mathcal R}_i$; (2) a ($k$, $\Psi$)-core may not be connected; and (3) $core_G(v,\Psi)$$\leq$$deg_G(v, \Psi)$.

\begin{example}
\label{eg:kphicore}
{\em
Let $\Psi$ be the triangle. Figure~\ref{fig:kcoreGraph}(b) shows all ($k$, $\Psi$)-cores of the graph.
The number $k$ in each circle indicates the ($k$, $\Psi$)-core contained in that ellipse.
For instance, the subgraph of $\{A,B,C,D\}$ is the (3, $\Psi$)-core as the 4-clique contains 4 triangle instances, and each vertex participates in 3 of them. Observe that $k$-cores and $(k,\Psi)$-cores are different between Figures~\ref{fig:kcoreGraph}(a) and ~\ref{fig:kcoreGraph}(b), for $k$=1, 2. Also, the entire graph is a $(0,\Psi)$-core.} \qed
\end{example}

%% file: coredensity.tex
\subsection{Density Bounds of ($k$, $\Psi$)-core}
\label{sec:coredensity}

The main result of this section is on the lower and upper bounds on the density of a ($k$, $\Psi$)-core.

\begin{theorem}
\label{thm:coreBounds}
Given a graph $G$ and an $h$-clique $\Psi(V_\Psi, E_\Psi)$, let ${\mathcal R}_k$ be a ($k$, $\Psi$)-core of $G$.
Then, the $h$-clique-density of ${\mathcal R}_k$ satisfies
\begin{equation}
\small
\frac{k}{{|{V_\Psi}|}} \le \rho ({\mathcal R}_k,\Psi ) \le {k_{\max}}.
\end{equation}
\end{theorem}

To prove this theorem, we develop the following lemmas.

\begin{lemma}
\label{lemma:connectivity}
Given a graph $G$ and an $h$-clique $\Psi$, the connected components of CDS $D$ have the same clique-density.
\end{lemma}

\vspace{-0.05in}
\textsc{Proof sketch.} The lemma can be proved by contradiction.
\qed

%\begin{proof}
%We first consider the case that $D$ has two connected components $C_1(V_{C_1}, E_{C_1})$ and $C_2(V_{C_2}, E_{C_2})$, and assume that their clique-densities are not the same.
%Then, we can easily conclude
%\begin{equation}
%{\rho _{opt}} = \frac{{\mu ({C_1},\Psi ) + \mu ({C_2},\Psi )}}{{|{V_{{C_1}}}| + |{V_{{C_2}}}|}}.
%\end{equation}
%
%Since $|V_{C_1}|$ and $|V_{C_2}|$ are non-negative, we can derive
%\begin{equation}
%\small{
%\min \left\{ {\frac{{\mu ({C_1},\Psi )}}{{|{V_{{C_1}}}|}},\frac{{\mu ({C_2},\Psi )}}{{|{V_{{C_2}}}|}}} \right\} \le {\rho _{opt}} \le \max \left\{ {\frac{{\mu ({C_1},\Psi )}}{{|{V_{{C_1}}}|}},\frac{{\mu ({C_2},\Psi )}}{{|{V_{{C_2}}}|}}} \right\}.}
%\end{equation}
%
%Notice that, $\rho ({C_1},\Psi ) = \frac{{\mu(C_1, \Psi)}}{{|{V_{{C_1}}}|}}$ and $\rho ({C_2},\Psi ) = \frac{{\mu(C_2, \Psi)}}{{|{V_{{C_2}}}|}}$. Therefore, if $\rho(C_1, \Psi)\neq \rho(C_2, \Psi)$, we can obtain a denser subgraph by removing the one with smaller density. This, however, contradicts the fact that $D$ is the CDS. Hence, $C_1$ and $C_2$ must be of the same density. The proof for the case where $D$ consists of more than two connected components is a straightforward extension, and is omitted here.
%\end{proof}

\begin{lemma}
\label{lemma:prune}
Given a graph $G(V,$ $E)$, an $h$-clique $\Psi(V_\Psi, E_\Psi)$, and the CDS $D(V_D, E_D)$, for any subset $U$ of $V_D$, removing $U$ from $D$ will result in the removal of at least $\rho_{opt}\times |U|$ clique instances from $D$.
\end{lemma}

\begin{proof}
We prove the lemma by contradiction. Assume that $D$ is the CDS and the removal of $U$ results in removing less than $\rho_{opt}\times |U|$ clique instances. Then, after removing $U$ from $V_D$, the clique-density of the residual graph
(denoted by $D\backslash U$) becomes:
\begin{equation}
\small
\rho (D\backslash U,\Psi ) = \frac{{\mu (D\backslash U,\Psi )}}{{|{V_D}| - |U|}} > \frac{{{\rho _{opt}}|{V_D}| - {\rho _{opt}}  |U|}}{{|{V_D}| - |U|}} = {\rho _{opt}}.
\end{equation}
However, this contradicts the assumption that $D$ is the CDS. Hence, the lemma holds.
\end{proof}

Based on the lemma above, we show an upper bound of $\rho_{opt}$.
\begin{lemma}
\label{lemma:upperbound}
Given a graph $G$, an $h$-clique $\Psi(V_\Psi, E_\Psi)$, and its maximum clique-core number $k_{\max}$, we have:
\begin{equation}
\rho_{opt} \le {k_{\max}}.
\end{equation}
\end{lemma}

\begin{proof}
We prove the lemma by contradiction. Suppose that we have $\rho_{opt} \textgreater {k_{\max}}$.
From Lemma \ref{lemma:prune}, we know that removing any vertex of $D$ will result in the removal of at least $\rho_{opt}$ clique instances, or more than $k_{\max}$ clique instances from $D$. In other words, each vertex of $D$ participates in at least $k_{\max}$+1 clique instances. This contradicts the fact that $k_{\max}$ is the maximum clique-core number. Hence, the value of $\rho_{opt}$ is at most ${k_{\max}}$.
\end{proof}

\noindent
{\bf Proof of \textsc{Theorem}~\ref{thm:coreBounds}}:
The upper bound follows by Lemma~\ref{lemma:upperbound}. Let us focus on the lower bound.
Let $r_k$ be the number of vertices in ${\mathcal R}_k$. By Definition~\ref{def:kpsicore}, since ${\mathcal R}_k$ is a ($k$, $\Psi$)-core, each vertex $v$ of ${\mathcal R}_k$ participates in at least $k$ clique instances. Meanwhile, each clique instance involves $|V_\Psi|$ vertices. As a result, there are at least $\frac{{k \times r_k}}{{|{V_\Psi}|}}$ clique instances in ${\mathcal R}_k$. Thus, we have $\rho ({\mathcal R}_k,\Psi ) \geq \frac{k}{{|{V_\Psi}|}}$. \qed

To further illustrate Theorem \ref{thm:coreBounds}, we give Example \ref{eg:bound}.

\begin{figure}[ht]
	\centering
	\includegraphics[width=0.78\linewidth]{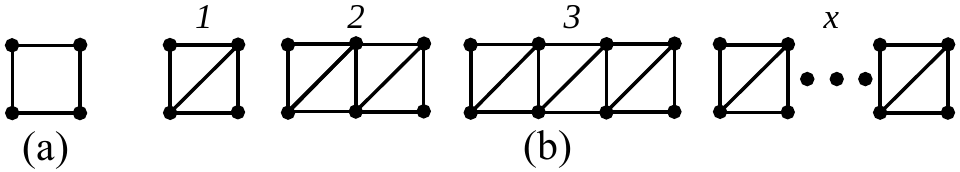}
    \vspace{-0.10in}
	\caption{Illustrating the lower and upper bounds. }
	\label{fig:bound}
\end{figure}

\begin{example}
\label{eg:bound}
{\em
Let $\Psi$ be an edge and consider the $k_{\max}$-core with $k_{\max}$=2. By Theorem \ref{thm:coreBounds}, the lower and upper bounds of the density of $k_{\max}$-core are 1 and 2 respectively. These bounds are attained by graphs in Figures \ref{fig:bound}(a) and \ref{fig:bound}(b) respectively. In Figure \ref{fig:bound}(a), the density of the $k_{\max}$-core is 4/4=1. In Figure \ref{fig:bound}(b), there is a list of graphs with $k_{\max}$=2, and the density values of $k_{\max}$-cores in the 1st, 2nd, $\cdots$, $x$-th graphs are $\frac{1+4}{2+2}$, $\frac{1+8}{2+4}$, $\cdots$, $\frac{1+4x}{2+2x}$, respectively. Clearly, when $x$$\rightarrow$$\infty$, the density converges to 2.
}
\qed
\end{example}

%% file: decompose.tex
\subsection{($k$, $\Psi$)-core Decomposition}
\label{sec:decompose}

Inspired by the $k$-core decomposition algorithm~\cite{kcore2003}, we develop an efficient ($k$, $\Psi$)-core decomposition algorithm for computing the clique-core number of each vertex. The algorithm exploits a key observation that, if we recursively remove vertices whose clique-degrees are less than a non-negative integer $k$, then the remaining graph, if non-empty, must be the ($k$, $\Psi$)-core.

Specifically, we first compute the clique-degree of each vertex, and sort vertices in increasing order of their clique-degrees. Then, we iteratively remove the vertex $v$ whose clique-degree is the smallest in each iteration, until the graph is empty. In each iteration, after removing $v$, we need to decrease the clique-degrees of vertices, which share clique instances with $v$, and re-sort the vertices. Notice that by using the bin-sort technique \cite{kcore2003}, sorting all the vertices takes linear time and re-sorting can also be done efficiently.

\begin{algorithm}
\small
\caption{$(k,\Psi)$-core decomposition.}
\label{alg:core}
\KwIn{$G(V, E)$, $\Psi(V_\Psi, E_\Psi)$;}
\KwOut{The clique-core number of each vertex;}
initialize $core[\text{ }]\gets$ an array with $n$ entries\;
\textbf{for} \emph{each vertex} $v\in V$ \textbf{do} compute its clique-degree $deg_G(v, \Psi)$\;
sort vertices of $V$ in increasing order of their clique-degrees\;
\While{$V$ is not empty}{
    $core[v]\gets deg_G(v, \Psi)$ where $v$ has the minimum clique-degree\;
    \For{each clique instance $\psi$ containing $v$}{
        \For{each vertex $u$ in $\psi$}{
            \If{$deg_G(u, \Psi)\textgreater deg_G(v, \Psi)$}{
                decrease $u$'s clique-degree;
            }
        }
    }
    update $G$ by removing $v$ and its incident edges\;
    resort the vertices in $V$\;
}
\Return the array $core[\text{ }]$;
\end{algorithm}

Algorithm~\ref{alg:core} presents the core decomposition algorithm.
First, we initialize an array $core[\text{ }]$ and compute the clique-degree of each vertex (lines 1-2).
Then, we sort all the vertices in increasing order (line 3).
Next, we recursively remove the vertex $v$ whose clique-degree is the smallest (lines 4-11).
In each iteration, we record $v$'s clique-core number (line 5), decrease the clique-degrees of vertices in $v$'s clique instances as removing $v$ causes the deletion of some clique instances (lines 6-9), update $G$, and resort vertices (lines 10-11).
Finally, we return $core[\text{ }]$ (line 12).

To compute the clique-degrees of all the vertices, we can first run an $h$-clique enumeration algorithm, and then compute the clique-degree of each vertex by listing all the $h$-cliques. During the core decomposition process, after removing a vertex $v$, we can first locate the subgraph induced by $v$ and its neighbors, then enumerate all the $h$-cliques in this subgraph, and finally decrease the clique-degrees of the vertices involved. In this paper, we use the state-of-the-art $h$-clique enumeration algorithm~\cite{Mauro1988}.

\begin{lemma}
\label{lemma:coreTime}
Given a graph $G$ and an $h$-clique $\Psi(V_\Psi, E_\Psi)$, the core decomposition algorithm above completes in ${\mathcal O}\left( {n \cdot {{d-1} \choose {h-1}}}\right)$ time and ${\mathcal O}\left(m\right)$ space.
\end{lemma}

\begin{proof}
For each vertex $v$, we need to compute the number of clique instances it involves, i.e., $deg_G(v,\Psi)$. In the worst case, any $h$--1 neighbors of $v$ can form an $h$-clique with $v$, so $deg_G(v,\Psi)$ is up to ${d-1 \choose h-1}$.
By using the bin-sort technique in \cite{kcore2003}, we can sort vertices of $V$ in linear time cost, and resorting after removing a vertex takes linear time cost to the clique-degree. In addition, computing $deg_G(v,\Psi)$ takes $\mathcal O(m)$ space as we can compute the clique instances sequentially. Hence, Lemma \ref{lemma:coreTime} holds.
\end{proof}

%% file: kcoreDiscuss.tex
\subsection{Extension and Discussion}
\label{sec:kcoreDiscuss}

The $k$-clique-core can be extended to $k$-pattern-core by incorporating a general pattern (e.g., star, loop, etc.). Let $\Psi$ be a pattern. Then, the ($k$, $\Psi$)-core is the largest subgraph of $G$, in which each vertex participates in at least $k$ instances of $\Psi$. The properties of $k$-clique-cores also hold for $k$-pattern-core. Besides, for any two patterns $\Psi$ and $\Psi'$, if $|V_\Psi|$=$|V_{\Psi'}|$ and $\Psi\subseteq\Psi'$, i.e., $\Psi$ is a subpattern of $\Psi'$, then the ($k$, $\Psi'$)-core is a subgraph of the ($k$, $\Psi$)-core. Algorithm \ref{alg:core} can also be extended for decomposing $k$-pattern-cores. We skip the details due to the space limitation.

Recently, Sariy{\"u}ce et al. studied the $k$-($r$, $s$) nucleus \cite{sariyuce2015finding,sariyuce2016fast,sariyuce2018local}, which is the maximal connected subgraph of the $r$-cliques where each $r$-clique is contained in at least $k$ $s$-cliques ($r$$\textless$$s$). When $\Psi$ is an $h$-clique, our ($k$, $\Psi$)-core can be considered as a special case of $k$-($r$, $s$) nucleus, i.e., $k$-(1, $h$) nucleus, in terms of clique-degree (or $\mathcal S$-degree in \cite{sariyuce2018local}). However, when $\Psi$ is a non-clique, ($k$, $\Psi$)-core is different with the $k$-($r$, $s$) nucleus, because in our ($k$, $\Psi$)-core, $\Psi$ can be an arbitrary pattern, such as clique, star, loop, etc., while $k$-($r$, $s$) nucleus is defined purely based on cliques. In other words, our ($k$, $\Psi$)-core can capture pattern-based dense subgraphs.
A second difference is that a $k$-($r$, $s$) nucleus requires that any two $r$-cliques $R$ and $R'$ are $\mathcal S$-connected: i.e., there exists a sequence of $r$-cliques $R$=$R_1$, $R_2$, $\cdots$, $R_l$=$R'$, such that $R_i$, $R_{i+1}$ are both contained by a specific $s$-clique ($i\in[1,l-1]$).
In addition, when $\Psi$ is an $h$-clique, the nucleus decomposition algorithm \cite{sariyuce2018local} can be applied to decomposing ($k$, $\Psi$)-cores. We will experimentally compare this method with ours in Section~\ref{sec:expKDS}.

%% file: advanced.tex
\section{Core-Based Approaches}
\label{sec:advanced}

Based on ($k$, $\Psi$)-cores, we develop efficient exact and approximation DSD algorithms. While our exact algorithm, {\tt CoreExact},  is significantly faster than the state-of-the-art algorithm ({\tt Exact}), we can speed it up further by trading accuracy: we develop an efficient approximation algorithm, namely {\tt CoreApp}, which has an approximation ratio of $\frac{1}{{\left| {{V_\Psi }} \right|}}$.

\input{advancedExact}
\input{coreAppAlgo}
\input{advancedDiscuss}

%% file: advancedExact.tex
\subsection{The Core-Based Exact Method}
\label{sec:advancedExact}

As shown in Lemma~\ref{lemma:exactTime}, the major limitation of Algorithm {\tt Exact} is its high computational cost.
To address this, in this section we exploit the $k$-clique-cores and propose the following three optimization techniques for boosting the efficiency. %Specifically, we propose the following four optimization techniques.

\noindent{\textbf{\underline{1 Tighter bounds on $\alpha$.}}}
In {\tt Exact}, the value of $\alpha$ is %supposed
within the range $[0, \mathop {\max }\limits_{v \in V} deg_G(v, \Psi)]$.
As discussed in Section~\ref{sec:coredensity}, by using the ($k$, $\Psi$)-cores, we can derive a tighter bound on $\alpha$.
Specifically, consider a ($k_{\max}$, $\Psi$)-core ${\mathcal R}_{k_{\max}}$. By Theorem~\ref{thm:coreBounds}, we can see that $\rho ({\mathcal R}_{k_{\max}},\Psi ) \ge \frac{k_{\max}}{{|{V_\Psi}|}}$, which implies that $\rho_{opt}\ge\frac{k_{\max}}{{|{V_\Psi}|}}$, so the lower bound of $\alpha$ is $\frac{k_{\max}}{{|{V_\Psi}|}}$.
On the other hand, by Lemma~\ref{lemma:upperbound}, we have $\rho_{opt}\le{k_{\max}}$ and thus the upper bound of $\alpha$ is $k_{\max}$. In practice, since $\frac{k_{\max}}{{|{V_\Psi}|}}$ is larger than 0 and $k_{\max}$ is smaller than the maximum clique-degree, the number of binary searches can be greatly reduced by using the tighter bounds.

\noindent{\textbf{\underline{2 Locating the CDS in a core.}}}
Recall that in each binary search of the algorithm {\tt Exact}, the flow network is reconstructed based on the entire graph $G$. This, however, is unnecessary, since the CDS is often in some ($k$, $\Psi$)-cores which could be much smaller than $G$.

\begin{lemma}
\label{lemma:mds}
Given a graph $G$ and an $h$-clique $\Psi$, the CDS is contained in the ($k$, $\Psi$)-core, where $k = \left\lceil {{\rho _{opt}}} \right\rceil$.
\end{lemma}

\begin{proof}
By Lemma~\ref{lemma:prune}, deleting any single vertex from CDS will result in the removal of $\left\lceil {{\rho _{opt}}} \right\rceil$ clique instances in CDS. In other words, each vertex of the CDS has participated in $\left\lceil {{\rho _{opt}}} \right\rceil$ clique instances. By the definition of ($k$, $\Psi$)-core, we conclude that the CDS is in the ($k$, $\Psi$)-core, where $k$=$\left\lceil {{\rho _{opt}}}\right\rceil$.
\end{proof}

As the value of $\rho _{opt}$ may not be known in advance, we can only locate it in the cores using the lower bounds of ${\rho_{opt}}$, by exploiting the nested property of cores.
For example, by Theorem~\ref{thm:coreBounds}, we have ${\rho _{opt}} \ge \frac{k_{\max}}{{|{V_\Psi}|}}$, which implies that the CDS must be in the ($k$, $\Psi$)-core, where $k$=$\left\lceil {\frac{k_{\max}}{{|{V_\Psi}|}}} \right\rceil$. Recall that in the core decomposition process, we delete vertices iteratively and obtain a residual subgraph after removing a vertex. In order to get a tighter lower bound on $\rho_{opt}$, we can compute the densities of these residual subgraphs.

$\bullet$ \emph{Pruning1:} The CDS is in the ($k'$, $\Psi$)-core, where $k'$=$\left\lceil {\rho'} \right\rceil$ and $\rho'$ is the highest $h$-clique-density of all residual graphs. The correctness directly follows Lemma \ref{lemma:mds}, since $\rho'\leq\rho_{opt}$.

Since the ($k'$, $\Psi$)-core may be disconnected and some connected components may be denser than others, we can further locate the CDS in a core with a larger core number, using \emph{Pruning2}.

$\bullet$ \emph{Pruning2:} For each connected component of the ($k'$, $\Psi$)-core, we compute its $h$-clique-density. Let $\rho''$ be the maximum $h$-clique-density of these connected components. If $\left\lceil {\rho ''} \right\rceil\textgreater k'$, we increase $k'$ to $k''$=$\left\lceil {\rho''} \right\rceil$ and the CDS is in the ($k''$, $\Psi$)-core. The correctness holds by Lemma \ref{lemma:mds}, since $\rho'\leq\rho''\leq\rho_{opt}$.

$\bullet$ \emph{Pruning3:} After locating the CDS in a connected component $C(V_C, E_C)$, we can change the stopping criterion of binary search to ``$u-l\textless\frac{1}{|V_C|(|V_C|-1)}$''.
Since $C(V_C$, $E_C)$ contains the CDS and the flow network is built using $C(V_C$, $E_C)$, the pruning is correct by following Algorithm~\ref{alg:basicExact}.

%\noindent{\textbf{\underline{3 Simplifying the flow network}.}}
%In {\tt Exact}, we create a node for each clique instance in the flow network ${\mathcal F}$.
%This, however, may be unnecessary, as some nodes can be pruned by Lemma~\ref{lemma:prunePattern}.
%As ${\mathcal F}$ is simplified, we can find the CDS more efficiently.
%
%\begin{lemma}
%\label{lemma:prunePattern}
%Given a graph $G$ and an $h$-clique $\Psi$, for any clique instance $\psi$, let $G'$ be the residual subgraph of $G$ after removing all the vertices of $\psi$ from $G$. If we have $\rho (G', \Psi) > \rho (G, \Psi)$, then we do not need to create a node for $\psi$ in the flow network.
%\end{lemma}
%
%\vspace{-0.1in}
%\textsc{Proof sketch.} Let $X$ be the set of vertices in $\psi$. We can first prove that at least one vertex of $\psi$ is not in the CDS, by contradiction. Next, we use this result to prove that after removing $\psi$ from the flow network $\mathcal F$, the CDS still can be computed correctly by conducting binary searching using $\mathcal F$.
%\qed

\noindent{\textbf{\underline{3 The flow network gradually becomes smaller}.}}
During the binary search, since the lower bound $l$ of $\alpha$ is gradually enlarged, we can locate the CDS in cores with larger clique-core numbers. As clique-core numbers increase, the sizes of cores become smaller, so the flow networks constructed become smaller gradually, and the cost of computing the minimum st-cut is greatly reduced.

\begin{algorithm}[]
\small
\caption{The algorithm: {\tt CoreExact}.}
\label{alg:advancedExact}
\KwIn{$G(V,E)$, $\Psi(V_\Psi,E_\Psi)$;}
\KwOut{The CDS $D(V_D, E_D)$;}
perform core decomposition using Algorithm~\ref{alg:core}\;
locate the ($k''$, $\Psi$)-core using pruning criteria\;
$\mathcal C\gets$all the connected components of ($k''$, $\Psi$)-core\;
initialize $D\gets\emptyset$, $U\gets\emptyset$, $l \leftarrow \rho''$, $u\gets k_{\max}$\;
\For{each connected component $C(V_C,E_C)\in\mathcal C$}{
    \textbf{if} $l \textgreater k''$ \textbf{then} $C(V_C,E_C)\gets C\cap(\left\lceil l \right\rceil$, $\Psi$)-core\;
    build a flow network $\mathcal F$($V_{\mathcal F}$, $E_{\mathcal F}$) by lines 5-15 of Algorithm \ref{alg:basicExact}\;
    find minimum st-cut ($\mathcal S$, $\mathcal T$) from $\mathcal F$($V_{\mathcal F}$, $E_{\mathcal F}$)\;
    \textbf{if} {$\mathcal S$=$\emptyset$} \textbf{then} continue\;
    \While{$ u-l\geq \frac{1}{|V_C|(|V_C|-1)} $}{
        $\alpha\gets \frac{l+u}{2}$\;
        build ${\mathcal F}$($V_{\mathcal F}$, $E_{\mathcal F}$) by lines 5-15 of Algorithm \ref{alg:basicExact}\;
        find minimum st-cut ($\mathcal S$, $\mathcal T$) from $\mathcal F$($V_{\mathcal F}$, $E_{\mathcal F}$)\;
        \uIf {$\mathcal S$=$\{s\}$ }{
            $u \gets \alpha$\;
        }\Else{
            \textbf{if} {$\alpha \textgreater \left\lceil {l} \right\rceil$} \textbf{then} remove some vertices from $C$\;
            $l \gets \alpha$\;
            %$U \gets {\mathcal S}-\{s\}$\;
            $U\gets\mathcal S\backslash \{s\}$\;
        }
    }
    \textbf{if} $\rho(G[U], \Psi)>\rho(D, \Psi)$ \textbf{then} $D\gets G[U]$\;
}
\Return $D$\;
\end{algorithm}

Combining the three optimization techniques above, we develop an advanced exact algorithm, called {\tt CoreExact}, as presented in Algorithm~\ref{alg:advancedExact}.
We first perform core decomposition and locate the CDS in the ($k''$, $\Psi$)-core (lines 1-2).
Then, we put the connected components of ($k''$, $\Psi$)-core into a set $\mathcal C$,
and initialize some variables including the lower and upper bounds of $\alpha$ (lines 3-4).
Next, in the loop (lines 5-20), we consider the connected components one by one.
Note that the lower bound $l$ is never decreased during the iterations.
If the current lower bound $l\textgreater k''$,
we replace $C$ by the core which has higher clique-core number and is contained by $C$ (line 6).
Intuitively, if $l$ is too large, then $C$ may not contain a subgraph with density $l$ and thus we skip it. Consequently, we build a flow network (line 7), and check whether $l$ is a feasible lower bound (line 8), i.e., whether there exists a subgraph with density at least $l$.

If $C$ cannot be skipped, we use binary search to find the CDS (lines 10-19). Each time we set a guess value of $\rho_{opt}$, namely $\alpha$, and check whether there is a subgraph with density of $\alpha$ or more.
Once we get a larger lower bound (line 16), we locate the CDS in the core with a larger clique-core number, so the network based on $C$ is even smaller. In other words, during the binary search, as the value of $\alpha$ approaches the true value of $\rho_{opt}$, the flow networks constructed become smaller. Finally, we get the CDS (line 21). We further illustrate {\tt CoreExact} by Example \ref{eg:CoreExact}.

\begin{figure}[ht]
	\centering
	\includegraphics[width=0.72\linewidth]{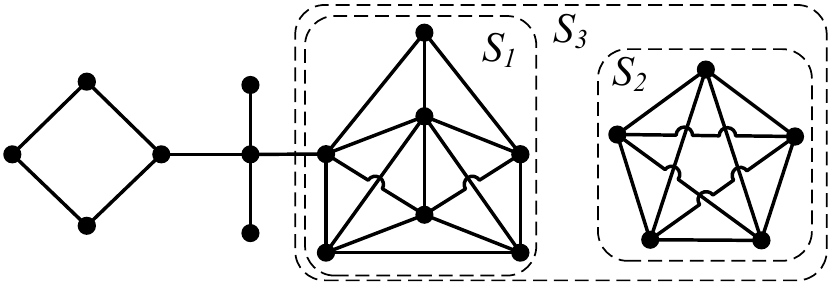}
    \vspace{-0.10in}
	\caption{Illustrating the core-based algorithms.}
    \label{fig:coreAlgo}
\end{figure}

\begin{example}
\label{eg:CoreExact}
{\em
Let $\Psi$ be a single edge and consider the graph in Figure \ref{fig:coreAlgo}, where $k_{\max}$=4.
During core decomposition, we track densities of residual graphs and obtain $\rho'$=25/12$\approx$2.08 (i.e., density of subgraph $S_3$). Thus, we get $\left\lceil {{\rho'}}\right\rceil$=3 and locate the EDS in the 3-core (i.e., subgraph $S_3$). The EDS (i.e., subgraph $S_1$ with density 15/7$\approx$2.14) can be computed by conducting binary search using the flow networks built on the two connected components $S_1$ and $S_2$ of $S_3$, rather than the entire graph, respectively.
}
\qed
\end{example}

%% file: coreAppAlgo.tex
\subsection{The Core-Based Approximation Methods}
\label{sec:coreAppAlgo}

Recall that Theorem~\ref{thm:coreBounds} gives the lower and upper bounds of density of a ($k$, $\Psi$)-core. Moreover, for a specific clique $\Psi$, the larger the value of $k$, the higher the lower bound on the density of the corresponding ($k$, $\Psi$)-core. Using Theorem~\ref{thm:coreBounds}, we can show:

\begin{lemma}
\label{lemma:approx}
Given a graph $G$ and an $h$-clique $\Psi(V_\Psi, E_\Psi)$, the ($k_{\max}$, $\Psi$)-core is a $\frac{1}{|V_\Psi |}$-approximation solution to CDS problem.
\end{lemma}

\begin{proof}
By Theorem~\ref{thm:coreBounds}, %Lemma~\ref{lemma:lowerbound},
we have $\frac{k_{\max}}{{|{V_\Psi}|}} \le \rho ({\mathcal R}_{k_{\max}},\Psi ) \le {k_{\max}}$.
Using the fact that $\rho_{opt} \leq k_{\max}$, we have
\begin{equation}
\small
\frac{{\rho ({{\mathcal R}_{{k_{\max }}}},\Psi )}}{{{\rho _{opt}}}} \ge \frac{{{k_{\max }}/\left| {{V_\Psi}} \right|}}{{{k_{\max }}}} = \frac{1}{{\left| {{V_\Psi}} \right|}}.
\end{equation}
The lemma follows.
\end{proof}

To compute the ($k_{\max}$, $\Psi$)-core, we can use the core decomposition method discussed in Section~\ref{sec:decompose}, which computes all the cores in an incremental manner. We denote this approximation algorithm by {\tt IncApp} (see  Algorithm~\ref{alg:IncApp}). Clearly, it has the same time complexity as the core decomposition algorithm.

\begin{algorithm}
\small
\caption{The algorithm: {\tt IncApp}.}
\label{alg:IncApp}
\KwIn{$G(V, E)$, $\Psi(V_\Psi, E_\Psi)$;}
\KwOut{The ($k_{\max}$, $\Psi$)-core;}
run ($k$, $\Psi$)-core decomposition algorithm (Section~\ref{sec:decompose})\;
%(Algo.~\ref{alg:core} in Appendix~\ref{appendix:codes})\;
\Return the ($k_{\max}$, $\Psi$)-core;
\end{algorithm}

A subtle point is that although the ($k_{\max}$, $\Psi$)-core is dense and provides an approximation solution, the CDS may not be in the ($k_{\max}$, $\Psi$)-core or even share some vertices with it.
For example, in Figure~\ref{fig:coreAlgo}, let $\Psi$ be a single edge. Then, the subgraph $S_2$ is the $k_{\max}$-core ($k_{\max}$=4), but the EDS is the subgraph $S_1$.

To further improve efficiency, we propose another method, called {\tt CoreApp}. Unlike {\tt IncApp} which computes all the cores, it focuses on computing the ($k_{\max}$, $\Psi$)-core directly. It relies on a key observation that the ($k_{\max}$, $\Psi$)-core often tends to be a subgraph of vertices with higher clique-degrees. We thus propose to discover the CDS from a sequence of subgraphs induced by vertices, whose clique-degrees are the largest. Moreover, once we find a core with higher clique-core number, we can prune some subgraphs, whose vertices' clique-degrees are too small. Thus, the CDS can be discovered efficiently. We remark that for $h$-cliques where $h$$\geq$3, computing the clique-degree $deg_G(v, \Psi)$ may be costly. Instead, we replace it by an upper bound $\gamma(v,\Psi)$, which can be computed more efficiently. Specifically, we run the $k$-core decomposition algorithm~\cite{kcore2003}, and for each vertex $v$ in an $x$-core, we set $\gamma(v,\Psi)$=${x \choose h-1}$.

\begin{algorithm}
\small
\caption{The algorithm: {\tt CoreApp}.}
\label{alg:CoreApp}
\KwIn{$G(V, E)$, $\Psi(V_\Psi, E_\Psi)$;}
\KwOut{The ($k_{\max}$, $\Psi$)-core;}
\textbf{for} $\forall v\in V$ \textbf{do} compute $\gamma(v,\Psi)$ of $deg_G(v, \Psi)$\;
sort vertices of $V$ in decreasing order of their $\gamma(v,\Psi)$ values\;
initialize $W$, $k_{\max}\gets$0, $S^*\gets\emptyset$\;
\While{$\mathop {\max }\limits_{v \in V\backslash W} {\gamma(v, \Psi)} \geq k_{\max}$}{
    \textbf{for} $\forall v\in W$ \textbf{do} compute $deg_{G[W]}(v, \Psi)$\;
    $k_l\gets\mathop {\min }\limits_{v \in W} {deg _{G[W]}}(v,\Psi )$,
    $k_u\gets\mathop {\max }\limits_{v \in W} {deg _{G[W]}}(v,\Psi )$\;
    $k\gets \max\{k_l,k_{\max}+1\}$\;
    \While{$k\leq k_u$ and $|W|\textgreater$0}{
        \While{$(\exists v\in W$, $deg_{G[W]}\textless k)$}{
            delete $v$ from $W$ and decrease clique-degrees\;
        }
        \If {$|W|\textgreater 0$} {
            \If {$k \textgreater k_{\max}$} {
                $k_{\max}\gets k$, $S^*\gets G[W]$\;
            }
            $k\gets k$+1\;
        }
    }
    $W\gets$ top-(2$\times|W|$) vertices in $V$\;
}
\Return $S^*$;
\end{algorithm}

Algorithm~\ref{alg:CoreApp} presents {\tt CoreApp}. First, we compute $\gamma(v)$ for each vertex, and sort vertices based on their $\gamma(v)$ values (lines 1-2). Then, we initialize three variables $W$, $k_{\max}$, and $S^*$, where $W$ keeps a set of vertices whose clique-degrees are the largest, and $S^*$ is used to track the ($k_{\max}$, $\Psi$)-core (line 3), which is computed from the vertex-induced subgraph $G[W]$.
Next, we compute the ($k_{\max}$, $\Psi$)-core in $G$, by using a while loop (lines 4-15).
Specifically, we first compute the exact clique-degree for each vertex in $G[W]$, and record the minimum and maximum clique-degrees (lines 6-7).
Then, we perform core decomposition for $G[W]$ with clique-core numbers in $[k_l, k_u]$ (lines 8-15), during which the maximum clique-core number $k_{\max}$ and ($k_{\max}$, $\Psi$)-core are kept (lines 13-14).
After that, we double the size of $W$ for the next iteration (line 15).
The loop can be stopped safely by using the stopping criterion (line 4).
Finally, we get ($k_{\max}$, $\Psi$)-core (line 16).
Note that $k_{\max}$ tracks the maximum clique-core number during the iterations, and for each subgraph $G[W]$, we focus on finding cores with core numbers larger than the previous $k_{\max}$ (line 7).

\noindent\textbf{Correctness.} Essentially, {\tt CoreApp} finds the ($k_{\max}$, $\Psi$)-core from a sequence of subgraphs induced by vertices in $W$ which have the largest clique-degrees. For each small subgraph $G[W]$, it computes the core with the highest core number by running core decomposition steps (lines 7-14). The stopping criterion (line 4) ensures that the ($k_{\max}$, $\Psi$)-core is correctly computed, i.e., since the maximum clique-degree of all the remaining vertices (in the set $V$$\backslash$$W$) is less than $k_{\max}$, their clique-core numbers must be less than $k_{\max}$.

\begin{lemma}
\label{lemma:CoreApp}
The time and space complexities of {\tt CoreApp} are ${\mathcal O}\left( {n \cdot {{d-1} \choose {h-1}}}\right)$ and ${\mathcal O}\left(m\right)$ respectively.
\end{lemma}

\begin{proof}
Let the number of iterations be $t$. Since we adopt the exponential growth strategy, the numbers of vertices involved in these iterations are at most
${\left( {\frac{1}{2}} \right)^{t - 1}} \cdot n,{\left( {\frac{1}{2}} \right)^{t - 2}} \cdot n, \cdots, n$ respectively, which form a geometric sequence. In the $i$-th iteration $i\in$[1, $t$], it takes ${\mathcal O}\left( {{\left( {\frac{1}{2}} \right)^{t - i}}\cdot n \cdot {{d-1} \choose {h-1}}}\right)$ time and ${\mathcal O}(m)$ space, as it performs core decomposition. By summarizing the time cost of all iterations, we obtain
${\mathcal O}\left( {2 \cdot n \cdot {d-1 \choose h-1}}\right)$ =
${\mathcal O}\left( {n \cdot {d-1 \choose h-1}}\right)$. The space cost is ${\mathcal O}(m)$, since the iterations are sequentially executed.% Hence, the lemma holds.
\end{proof}

Although, {\tt CoreApp} has almost the same {\sl worst-case cost} as {\tt PeelApp} and {\tt IncApp}, it performs much faster in practice, because the CDS is often much smaller than $G$ and thus only a few subgraphs are examined in the iterations. As shown by our experiments next, {\tt CoreApp} is up to two orders of magnitude faster than {\tt PeelApp} and {\tt IncApp}. Moreover, the approximation algorithms generate high-quality solutions -- their actual approximation ratios are often much higher than their theoretical approximation ratios.

\noindent\textbf{Remark.}
In \cite{cheng2011efficient}, Cheng et al. present an external-memory core decomposition algorithm, called {\tt EMcore}, which also works in a top-down manner. However, there are four differences between {\tt CoreApp} and {\tt EMcore}:
(1) {\tt CoreApp} can handle any $h$-clique- and pattern-cores, while {\tt EMcore} is developed for processing the classical (edge-based) $k$-cores.
(2) {\tt CoreApp} focuses on computing the ($k_{\max}$, $\Psi$)-core while {\tt EMcore} decomposes all $k$-cores.
(3) The methods of estimating upper bounds of core numbers are different.
(4) In the worst case, for classical $k$-cores, {\tt CoreApp} takes $\mathcal O(n+m)$ time while {\tt EMcore} takes $\mathcal O(k_{\max}(n+m))$ time, since both of them conduct core decomposition for a sequence of subgraphs, but the strategies of considering subgraphs are different.
Our later experiments show that for computing the $k_{\max}$-core, {\tt CoreApp} is faster than {\tt EMcore}.

%% file: advancedDiscuss.tex
\subsection{Discussions}
\label{sec:advancedDiscuss}

Below, we discuss the parallelizability of our algorithms and show that our algorithms can solve a variant of the CDS problem.

\noindent\textbf{Parallelizability.} The existing parallel $k$-core decomposition algorithms \cite{montresor2013distributed,mandal2017distributed,sariyuce2018local} can be easily extended for decomposing ($k$, $\Psi$)-cores, so our approximation solutions, which rely on the ($k_{\max}$, $\Psi$)-core, can be computed in parallel. Moreover, for the exact solution {\tt CoreExact}, the main overhead comes from the step of computing the minimum st-cut. The parallel algorithms of computing the minimum st-cut have been studied extensively \cite{johnson1987parallel,pham2005distributed}, so our exact algorithm can also be easily parallelized.

\noindent\textbf{A variant of CDS problem.}
In \cite{tsourakakis2015k}, Tsourakakis et al. studied a variant of the densest $k$ subgraph problem \cite{bhaskara2010detecting,andersen2009finding}, which aims to find a subgraph that contains a given set $Q$ of $k$ query vertices ($|Q|$=$k$) with the highest density, and its exact solution follows the framework of the exact solution of CDS problem by solving a maximum flow problem. 
To solve this problem with edge-density, we can first decompose $k$-cores and get the minimum core number $x$ of these $k$ vertices. Then, then lower bound of the edge-density of $x$-core is $\frac{x}{2}$ by Theorem~\ref{thm:coreBounds}.
Since $x$-core contains $Q$, we get a lower bound of $\rho_{opt}$ which is $\frac{x}{2}$. As a result, we can locate the densest subgraph in $\frac{x}{2}$-core, so we can build a flow network on $\frac{x}{2}$-core, rather than the entire graph, resulting in higher efficiency.
%In \cite{tsourakakis2013denser}, Tsourakakis et al. studied the top-$k$ densest subgraphs, which iteratively invokes the algorithm of computing the densest subgraph $k$ times in the residual graph after deleting the identified dense subgraphs in the previous iterations. In \cite{qin2015locally}, Qin et al. discovered top-$k$ locally densest subgraphs, each of which is a subgraph with the highest density in its local region in the graph. To compute these densest subgraphs, we can repeatedly execute the steps of computing densest subgraphs and checking whether they are locally densest subgraphs. Our algorithms can be applied to solving these two problems since the key step of computing densest subgraphs in their algorithms can be replaced by our core-based algorithms.

%% file: pds.tex
\section{The PDS Problem and Solutions}
\label{sec:pds}

A {\it pattern} (a.k.a. {\it motif} or {higher-order structure}) is a small graph containing a few vertices (e.g., a diamond in Figure~\ref{fig:intro}(b)).  These patterns can be considered as building blocks of knowledge graphs or biological databases~\cite{nature2003,hu2019discovering,fang2018spatialICDE}.  Compared to graph edges, they can better capture the intricate relationship among vertices, as well as the underlying rich semantics.
For example, in a protein interaction network, proteins are often organized in cohesive patterns of interactions, each of which represents some particular functions~\cite{nature2003}. We now study the discovery of {\it pattern-aware densest subgraphs}, i.e., subgraphs that are ``dense'' in terms of the number of patterns. We term this pattern densest subgraph (PDS) problem and show how our previous CDS solutions can be adapted.

%
%In many emerging applications (e.g., knowledge graphs and biological databases), the graphs often possess rich semantics, and graph edges are no longer sufficient to model the intricate relationship among vertices, which may go beyond just pairwise relationships. To remedy this issue, recently researchers have proposed small {\it patterns}, also called motifs or high-order structures, which are building blocks of large graphs. For example, in a protein interaction network~\cite{nature2003}, proteins are often organized in cohesive patterns of interactions, each of which represents some particular functions. Thus, it is desirable to discover pattern-aware densest subgraphs.
%In this section, we generalize the $h$-clique-density to pattern-density,  formulate the pattern densest subgraph (PDS) problem, and show that our previous solutions can be easily adapted to solve the PDS problem.

\input{pdsProblem}
\input{pdsAlgorithm}

%% file: pdsProblem.tex
\subsection{The PDS Problem}
\label{sec:pdsProblem}

We generalize the $h$-clique to a general pattern, which is a connected simple graph $\Psi(V_\Psi, E_\Psi)$. We formally introduce definitions of pattern instance and pattern-density below.

\begin{definition}[Subgraph isomorphism]
\label{def:iso}
A graph $G(V$, $E)$ is subgraph isomorphic to a pattern $\Psi(V_\Psi$, $E_\Psi)$ if there exists an injection $\phi$:$V_\Psi\rightarrow V$, such that for all $v,v' \in V_{\Psi}$, if $(v,v')\in E_{\Psi}$, then $(\phi(v), \phi(v')) \in E$.
\end{definition}

\begin{definition}[Pattern instance]
\label{def:instance}
Given a graph $G(V$, $E)$ and a pattern $\Psi(V_\Psi$, $E_\Psi)$, a subgraph $S(V_S, E_S)\subseteq G$
is a pattern instance of $\Psi$, if $S$ is isomorphic to $\Psi$.
\end{definition}

\begin{definition}[Pattern-degree]
\label{def:patterndegree}
Given a graph $G(V,E)$ and a pattern $\Psi$, the pattern-degree of a vertex $v$,
or $deg_{G}(v, \Psi)$, is the number of pattern instances of $\Psi$ containing $v$.
\end{definition}

Clearly, $G$ is subgraph isomorphic to $\Psi$ iff it has a subgraph $S(V_S, E_S)$ that is isomorphic to $\Psi$. Note that $S$ may not be a vertex-induced subgraph, although we note that our algorithms can be easily adapted for the vertex-induced case.
Due to symmetry, for a single subgraph $S$ of $G$, there may be multiple mappings witnessing that $\Psi$ is isomorphic to $S$, which are automorphisms, but in this case we do not distinguish between different automorphisms of $S$ and instead count instances based on the edge set.

\begin{definition}[Pattern-density]
\label{def:patterndensity}
Given a graph $G(V,$ $E)$ and a pattern $\Psi(V_\Psi, E_\Psi)$, the pattern-density of $G$ w.r.t. $\Psi$ is
$\rho(G, \Psi)=\frac{\mu(G, \Psi)}{|V|}$,
where $\mu(G, \Psi)$ is the number of pattern instances of $\Psi$ in $G$.
\end{definition}

\begin{problem}[PDS Problem]
\label{prob:PDS}
Given a graph $G(V, E)$ and a pattern $\Psi(V_\Psi,E_\Psi)$, return the subgraph $D$ of $G(V, E)$,
whose pattern-density $\rho(D, \Psi)$ is the highest.
\end{problem}

For example, consider the graph in Figure~\ref{fig:flow}(a) and let $\Psi$ be the diamond pattern (Figure~\ref{fig:intro}(b)). Then, the subgraph of $\{A,D,E,F\}$ is the densest subgraph, which contains three pattern instances (Figure~\ref{fig:flow}(c)) and has the highest pattern-density.

%% file: pdsAlgorithm.tex
\subsection{Algorithms for PDS Problem}
\label{sec:pdsAlgorithm}

\noindent\textbf{Approximation methods.} To compute the approximate PDS's, we can directly adapt algorithm {\tt PeelApp} by replacing the steps of computing clique instances (clique-degrees) by pattern instances (pattern-degrees). The correctness is guaranteed by Lemma~\ref{lemma:peelApp}. Similarly, {\tt IncApp} and {\tt CoreApp} can be adapted.

\begin{lemma}
\label{lemma:peelApp}
Given a graph $G$ and a pattern $\Psi(V_\Psi, E_\Psi)$, the subgraph $S^*$ returned by {\tt PeelApp} is a $\frac{1}{|V_\Psi |}$-approximation solution to the PDS problem w.r.t. pattern-density for pattern $\Psi$.
\end{lemma}

\textsc{Proof sketch.} We can prove the lemma by generalizing Lemma \ref{lemma:approx} and Theorem \ref{thm:coreBounds} for supporting an arbitrary pattern.\qed

\vspace{0.05in}

\noindent\textbf{Exact methods.} The algorithm {\tt Exact} in Section \ref{sec:basicExact} cannot be trivially extended for computing the exact PDS's since it relies on ($h$--1)-cliques. Nevertheless, we can adapt the exact CDS algorithm in \cite{tsourakakis2015k}, which follows the framework of {\tt Exact} but introduces a different flow network construction method, for computing the exact PDS's by replacing the steps of computing clique instances (clique-degrees) by pattern instances (pattern-degrees). We denote this algorithm by {\tt PExact}, and its pseudocodes are presented in the technical report \cite{fullVersion}. Theorem~\ref{thm:exactCorrect} shows its correctness.

\begin{theorem}
\label{thm:exactCorrect}
Given a graph $G$ and a pattern $\Psi(V_\Psi, E_\Psi)$, the algorithm {\tt PExact} correctly finds the PDS of $G$ w.r.t. pattern-density of $\Psi$.
\end{theorem}

\textsc{Proof.} Please refer to the technical report \cite{fullVersion}.
\qed

\vspace{0.05in}

Our core-based techniques can be used for improving {\tt PExact}. Specifically, we adopt the $k$-pattern-core in Section \ref{sec:kcoreDiscuss} and use the three optimization techniques in Section \ref{sec:advancedExact}.
In addition, we propose a new optimization strategy, which relies on the following key observation: for a general pattern $\Psi$, different pattern instances may share the same set of vertices, but {\tt PExact} creates a node for each of them when building the flow network.
For example, consider the graph in Figure~\ref{fig:flow}(a). If the pattern is a diamond, then the three pattern instances in Figure~\ref{fig:flow}(c) share the same set of vertices.

\begin{algorithm}[htb]{}
\small
\caption{{\tt construct+}($G$, $\Psi$, $\alpha$).}
\label{alg:construct+}
\KwIn{$G(V, E)$, $\Psi(V_\Psi,E_\Psi)$, $\alpha$;}
\KwOut{The flow network ${\mathcal F}(V_{\mathcal F}, E_{\mathcal F})$;}
$\Lambda\gets$ all the pattern instances of $\Psi$ in $G$\;
$\Lambda'$=\{$g_1$, $g_2$, $\cdots$, $g_{|\Lambda'|}$\}$\gets$ group the pattern instances in $\Lambda$\;
$V_{\mathcal F}\gets \{s\}\cup V\cup\Lambda'\cup\{t\}$\;
$\forall v\in V$, add an edge $s$$\rightarrow$$v$ with capacity $deg_G(v, \Psi)$\;
$\forall v\in V$, add an edge $v$$\rightarrow$$t$ with capacity $\alpha|V_\Psi|$\;
$\forall v\in V$, if it appears in a group $g\in\Lambda'$, add an edge $v$$\rightarrow$$g$ with capacity $|g|$\;
$\forall g\in\Lambda'$, if it contains a vertex $v$, add an edge $g$$\rightarrow$$v$ with capacity $|g|(|V_\Psi|-1)$\;
\Return ${\mathcal F}(V_{\mathcal F}, E_{\mathcal F})$\;
\setlength{\textfloatsep}{0pt}
\end{algorithm}

Based on the observation above, we propose a new flow network construction method {\tt construct+}, by grouping nodes of pattern instances having same set of vertices. Algorithm \ref{alg:construct+} shows {\tt construct+}.
First, a set $\Lambda'$=$\{g_1, g_2, \cdots, g_{|\Lambda'|}\}$ is collected, where each $g_i$ denotes a group of pattern instances sharing the same set of vertices.
Second, for each vertex $v\in V$, we set the capacities of edges ($s$, $v$) and ($v$, $t$) similarly with that in {\tt PExact}.
Third, for each vertex $v\in V$, if it appears in a group $g\in\Lambda'$, the capacity of edge ($v$, $g$) is set to $|g|$; for each group $g\in\Lambda'$, the capacity of edge ($g$, $v$) is set to $|g|(|V_\Psi|-1)$.
Here, we define the capacties based on the intuition that the densest subgraph $D$ is obtained by computing the minimum st-cut ($\mathcal S$, $\mathcal T$), and vertices of $D$ must be in one partition $\mathcal S$. This implies that nodes of all the pattern instances in $D$ should be in $\mathcal S$, and thus we can accumulate their capacities by using the term $|g|$ when computing the maximum flow from $\mathcal S$ to $\mathcal T$.  Note that if $\Psi$ is a clique, then $|g|$=1.
The correctness is stated by Lemma~\ref{lemma:compressF}. We illustrate {\tt construct+} by Example~\ref{eg:flow}.
We denote the above core-based exact PDS algorithm by {\tt CorePExact}.

\begin{lemma}
\label{lemma:compressF}
Given a graph $G$, a pattern $\Psi(V_\Psi, E_\Psi)$, the flow networks built by {\tt PExact} (lines 5-12) and {\tt construct+} have the same capacity for their minimum st-cut.
\end{lemma}

\textsc{Proof.} Please refer to the technical report \cite{fullVersion}.
\qed

\begin{example}
\label{eg:flow}
{\em
Let $\Psi$ be the diamond pattern. The graph in Figure~\ref{fig:flow}(a) has 4 pattern instances, which are grouped into 2 groups as shown in Figures \ref{fig:flow}(b) and \ref{fig:flow}(c).
Clearly, we can locate the PDS in (1, $\Psi$)-core, in which the vertex set is $\{A, B, \cdots, F\}$ and $\Lambda'$ = $\{g_1, g_2\}$.
To build $\mathcal F(V_{\mathcal F}$, $E_{\mathcal F})$, we first collect the set $V_{\mathcal F}$, then create 10 nodes, and finally add edges. For example, for group $g_2$, we link it to all its vertices with capacities $|g_2|(|V_{\Psi}|$--1$)$=9 and their reversed edges are with capacities 3. Figure~\ref{fig:flow}(d) shows $\mathcal F$.}
\qed
\end{example}

\begin{figure}[]
	\centering
	\includegraphics[width=0.91\linewidth]{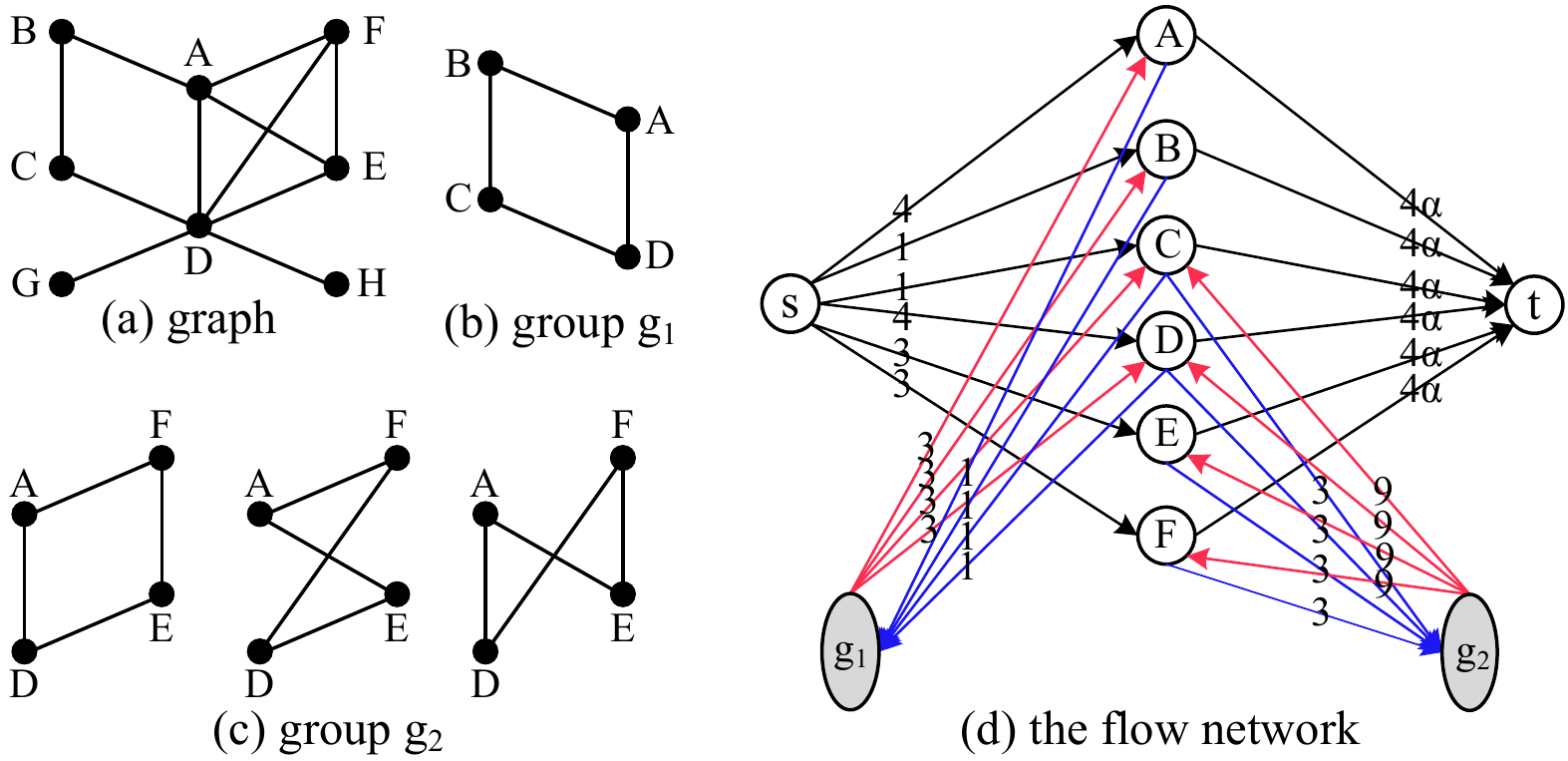}
    \vspace{-0.10in}
	\caption{Illustrating the flow network in {\tt CorePExact}. }
	\label{fig:flow}
\end{figure}

\noindent\textbf{Remark.} {\tt CorePExact} relies on the core decomposition. For some special patterns such as stars and loops, the core decomposition algorithm in Algorithm~\ref{alg:core} can be performed faster by optimizing the steps of computing pattern-degrees and decreasing the vertices' pattern-degrees. For details, please refer to the technical report~\cite{fullVersion}. For general patterns, we use the state-of-the-art pattern enumeration algorithm \cite{qiao2017subgraph} for computing the pattern-degrees.

%% file: experiments.tex
\section{Experiments}
\label{sec:exp}

%We now present our experimental results. We discuss the setup in Section~\ref{sec:setup}. Section~\ref{sec:expKDS} evaluates DSD solutions for edge- and $k$-clique-densities. In Section~\ref{sec:expMDS}, we discuss t for different patterns.

%The experimental setup is given in Section~\ref{sec:setup}. Section~\ref{sec:expKDS} evaluates the efficiency of DSD for edge and cliques. Section~\ref{sec:expMDS} presents the performance of DSD for general patterns.

\input{setup}

\input{expKDS}

\input{expPDS}

%% file: setup.tex
We have performed experiments on ten real graphs~\footnote{The datasets are:
Yeast (\url{https://dip.doe-mbi.ucla.edu/dip/Stat.cgi});
Netscience (\url{http://www-personal.umich.edu/~mejn/netdata/});
DBLP (\url{http://dblp.uni-trier.de/xml/}); and
Enwiki-2017 and UK-2002 (\url{http://law.di.unimi.it/}).
Others are found at~\url{https://snap.stanford.edu/data/}.}
(see Table~\ref{tab:dataset}).
These graphs cover various domains, such as biological networks (e.g., Yeast), collaboration networks (e.g., Ca-HepTh), autonomous system graphs (e.g., As-Caida), bibliographical graphs (e.g., DBLP), web graphs (e.g., UK-2002), citation networks (e.g., Cit-Patents), social networks (e.g., Friendster), etc.

\begin{table}[ht]
  \vspace{-0.1in}
  \centering
  \scriptsize
  \caption {Datasets used in our experiments.}
  \label{tab:dataset}
  \begin{tabular}{c|c|r|r}
     \hline
          {\bf Graph}     & \multicolumn{1}{c|}{\textbf{Name}}
                         & \multicolumn{1}{c|}{\textbf{Vertices}}
                         & \multicolumn{1}{c}{\textbf{Edges}}\\
     \hline\hline
          \multirow{4}*{\tabincell{c}{Real small graphs\\ (all algo.)}}
          &Yeast         &  1,116    &  2,148\\
     \cline{2-4}
          & Netscience    &  1,589    &  2,742\\
     \cline{2-4}
          &As-733        &  1,486    &  3,172\\
     \cline{2-4}
          &Ca-HepTh      &  9,877    &  25,998 \\
     \cline{2-4}
          &As-Caida      &  26,475   &  106,762\\
     \hline
          \multirow{6}*{\tabincell{c}{Real large graphs\\ (approx. algo.)}} & DBLP          &  425,957  &  1,049,866\\
     \cline{2-4}
          & Cit-Patents   &  3,774,768&  16,518,948\\
     \cline{2-4}
          & Friendster    &  20,145,325& 106,570,765\\
     \cline{2-4}
          & Enwiki-2017   & 5,409,498 & 122,008,994\\
     \cline{2-4}
          & UK-2002       &  18,520,486	& 298,113,762\\
     \hline
           \multirow{3}*{\tabincell{c}{Synthetic\\ random graphs}}
          & SSCA      &  100,000  &  3,405,676\\
     \cline{2-4}
          & ER         &  100,000  &  4,837,534\\
     \cline{2-4}
          & R-MAT        &  100,000  &  2,571,986\\
     \hline
  \end{tabular}
\end{table}

Besides, as shown in Table \ref{tab:dataset}, we have used three synthetic random graphs (SSCA, ER, and R-MAT) generated by GTgraph \footnote{\label{foot:test}GTgraph random graph generator: \url{http://www.cse.psu.edu/~kxm85/software/GTgraph/}}.
These three graphs follow three representative distributions: SSCA is made by random-sized cliques, ER follows the random distribution, and R-MAT follows the power-law distribution. Note that for SSCA and R-MAT, we set default parameters of their generators; for ER, we set the probability of an edge between any pair of vertices to 0.0005, which is also the chance that an edge exists in the real graph Cap-HepTh. A more detailed analysis of the characteristics of these datasets is in the technical report \cite{fullVersion}, which we omit here due to lack of space.

We considered two groups of patterns: (1) $h$-cliques (with $h\in[2,6]$); and (2) seven other patterns (Figure~\ref{fig:7motifs}) studied in \cite{nature2003,leskovec2006,lai2016scalable}, each of which is associated with an ID (e.g., 4 = diamond).

\begin{figure}[ht]
	\centering
	\includegraphics[width=0.88\linewidth]{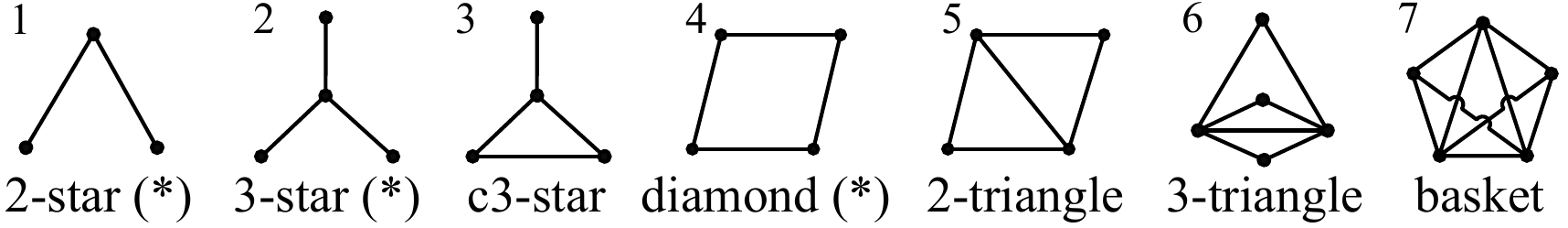}
    \vspace{-0.05in}
	\caption{Patterns used in evaluation of PDS.}
	\label{fig:7motifs}
\end{figure}

For CDS problem, we tested 2 exact algorithms ({\tt Exact} and {\tt CoreExact}) and 5 approximation algorithms ({\tt Nucleus} \cite{sariyuce2018local}, {\tt EMcore} \cite{cheng2011efficient}, {\tt PeelApp}, {\tt IncApp}, and {\tt CoreApp}).
{\tt Nucleus} is applied for decomposing the ($k$, $\Psi$)-core where $\Psi$ is an $h$-clique. For fair comparison, we implement the the faster nucleus decomposition algorithm AND \cite{sariyuce2018local} on a single core.
We also adapt {\tt EMcore} such that it works in main memory and stops when the $k_{\max}$-core is computed.
For PDS problem, we tested both exact algorithms ({\tt PExact}, {\tt CorePExact}) and approximation algorithms ({\tt PeelApp}, {\tt IncApp}, {\tt CoreApp}).
For special patterns marked $\ast$ in Figure~\ref{fig:7motifs}, we have implemented optimizations discussed in the technical report~\cite{fullVersion}, for all algorithms.
Note that {\tt CoreExact}, {\tt IncApp}, {\tt CoreApp}, and {\tt CorePExact} are our core-based approaches.
All these solutions are implemented in Java, and executed on a machine having an Intel(R) Xeon(R) 3.40GHz processor, 16 cores, and 125GB of memory, with Ubuntu installed.

\begin{figure*}[htp]
\hspace*{-.6cm}
\centering
\begin{tabular}{c c c c c}
  &
  \begin{minipage}{3.30cm}
	\includegraphics[width=9.0cm]{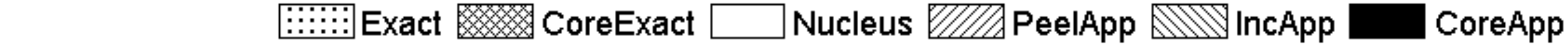}
  \end{minipage}
  &
  &
  &
  \\

  \begin{minipage}{3.30cm}
	\includegraphics[width=3.55cm]{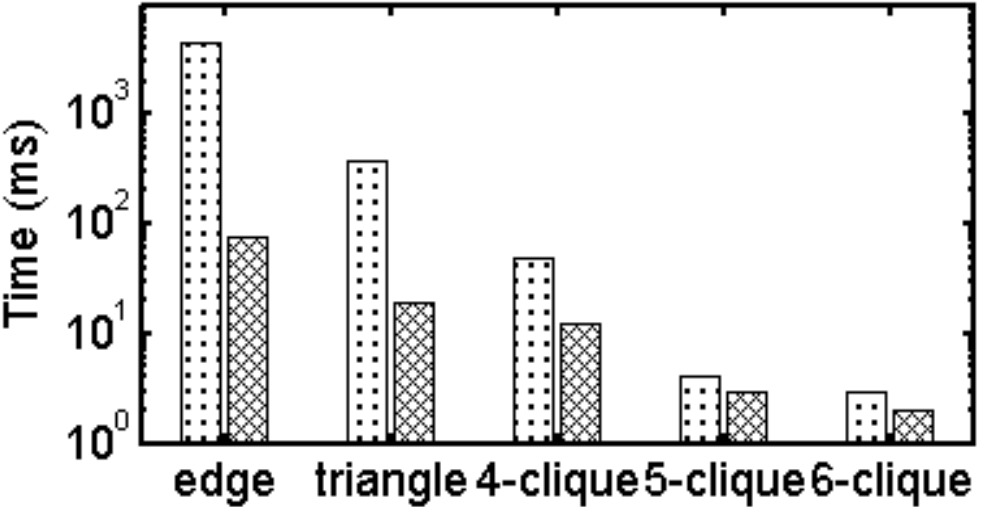}
  \end{minipage}
  &
  \begin{minipage}{3.30cm}
	\includegraphics[width=3.55cm]{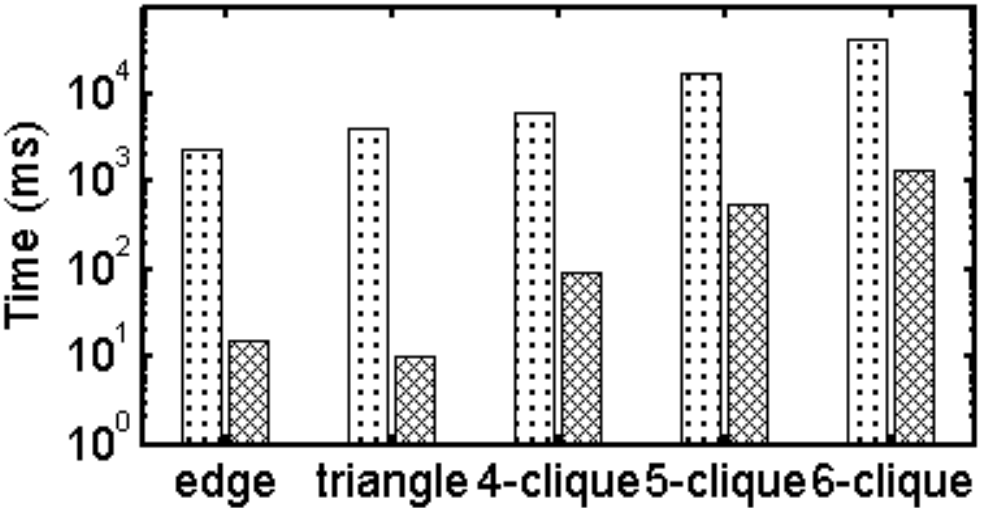}
  \end{minipage}
  &
  \begin{minipage}{3.30cm}
	\includegraphics[width=3.55cm]{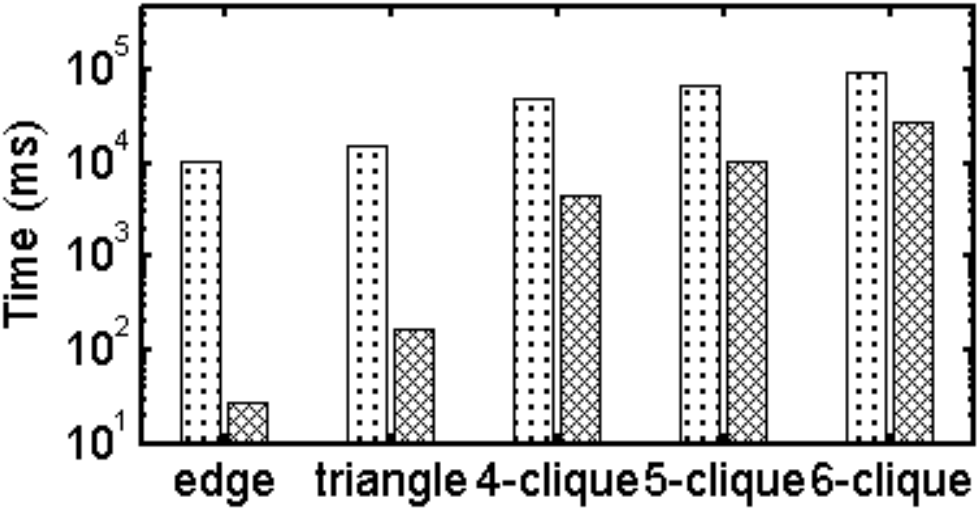}
  \end{minipage}
  &
  \begin{minipage}{3.30cm}
	\includegraphics[width=3.55cm]{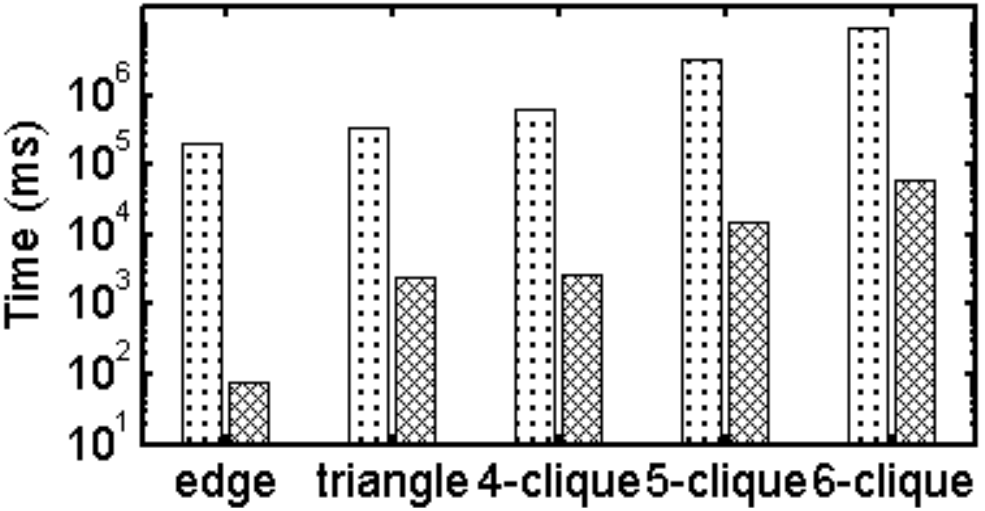}
  \end{minipage}
  &
  \begin{minipage}{3.30cm}
	\includegraphics[width=3.55cm]{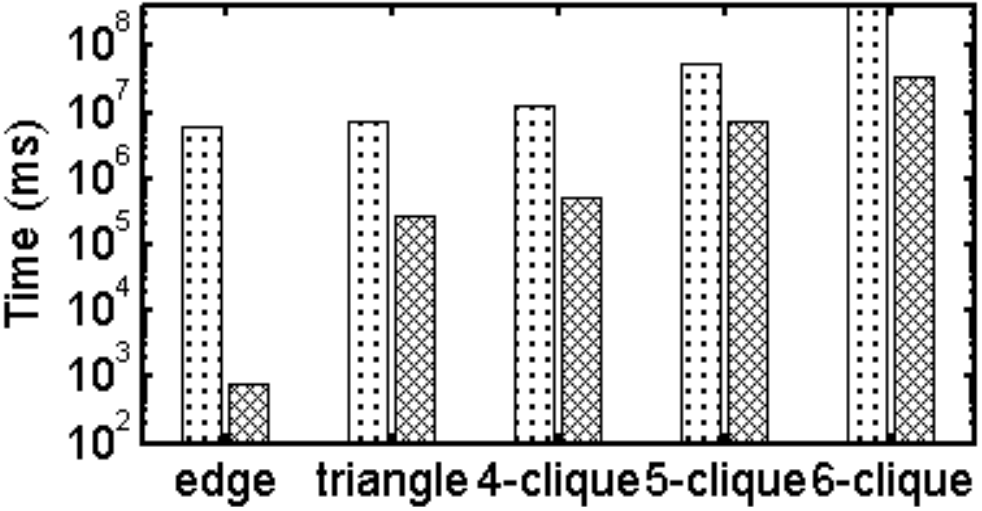}
  \end{minipage}
  \\
  (a) Yeast (exact)
  &
  (b) Netscience (exact)
  &
  (c) As-733 (exact)
  &
  (d) Ca-HepTh (exact)
  &
  (e) As-Caida  (exact)
  \\

  \begin{minipage}{3.30cm}
	\includegraphics[width=3.55cm]{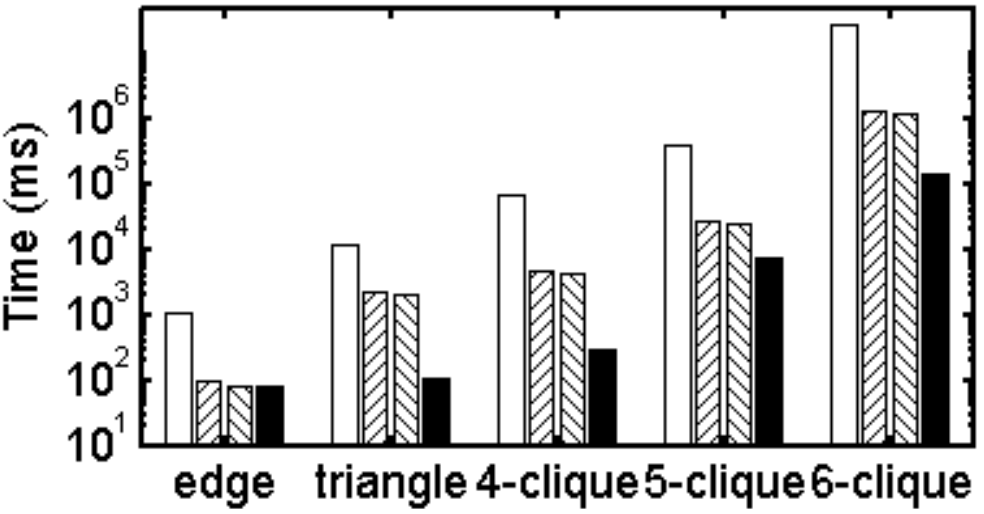}
  \end{minipage}
  &
  \begin{minipage}{3.30cm}
	\includegraphics[width=3.55cm]{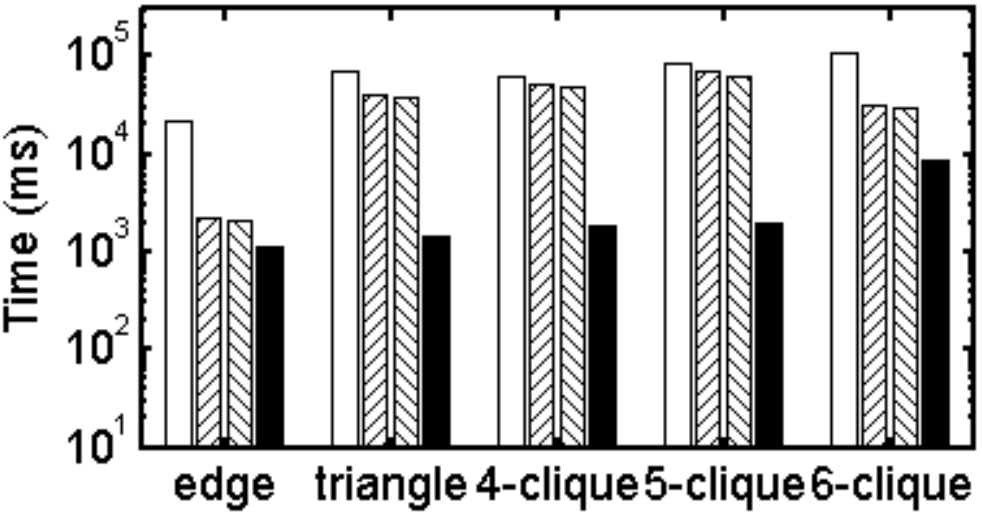}
  \end{minipage}
  &
  \begin{minipage}{3.30cm}
	\includegraphics[width=3.55cm]{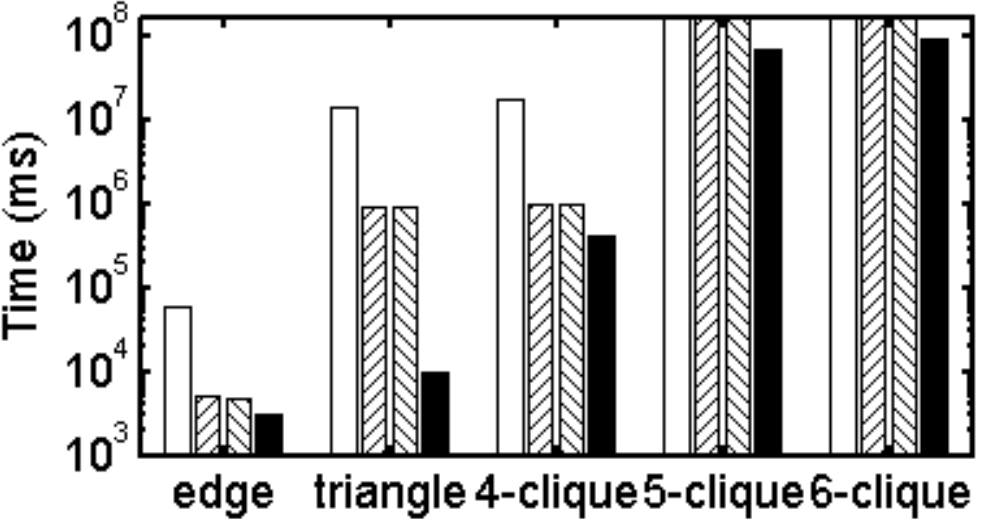}
  \end{minipage}
  &
  \begin{minipage}{3.30cm}
	\includegraphics[width=3.55cm]{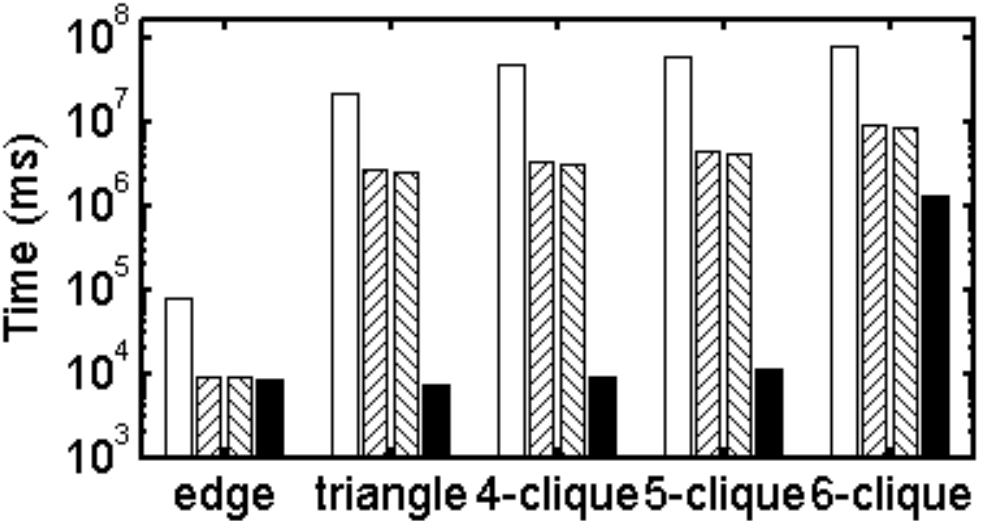}
  \end{minipage}
  &
  \begin{minipage}{3.30cm}
	\includegraphics[width=3.55cm]{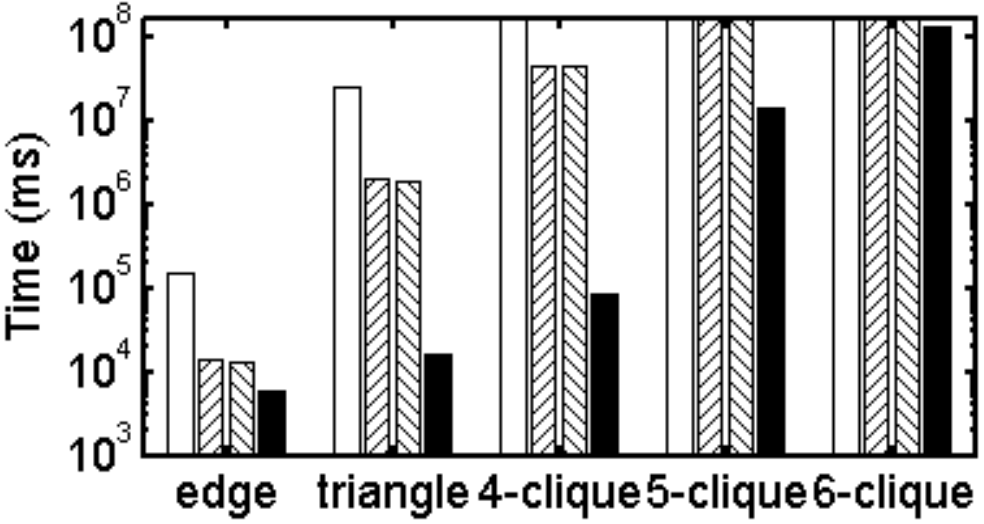}
  \end{minipage}
  \\
  (f) DBLP (app.)
  &
  (g) Cit-Patents (app.)
  &
  (h) Friendster (app.)
  &
  (i) Enwiki-2017 (app.)
  &
  (j) UK-2002 (app.)
\end{tabular}
\vspace{-0.10in}
\caption{Efficiency of exact and approximation CDS algorithms.}
\label{fig:exp-kds}
\end{figure*}

%% file: expKDS.tex
\subsection{DSD for Edge- and $h$-Clique-Densities}
\label{sec:expKDS}

\noindent
{\bf\underline{1 Exact algorithms}.} Figures~\ref{fig:exp-kds}(a)-(e) show the performance of exact algorithms on five small datasets. (As these solutions cannot finish in a reasonable time on larger datasets, we do not report their results here.)  We see that the time costs of all the algorithms increase with the $h$-clique size. Moreover, {\tt CoreExact} is at least 4.5$\times$ and up to four orders of magnitude faster than the existing algorithm {\tt Exact}~\footnote{For exact algorithms, bars touching the upper boundaries mean that the corresponding algorithms cannot finish within 5 days.}.
This is because {\tt CoreExact} employs the $k$-clique-cores, or ($k$, $\Psi$)-cores, which not only effectively locate the CDS in some smaller subgraphs, but also significantly reduce the flow network sizes in the binary search process. In contrast, the flow network of {\tt Exact} is built on the entire graph in each iteration, and the sizes of the flow networks remain unchanged in all the iterations. Hence, {\tt CoreExact} is faster than {\tt Exact}.

We now investigate how the flow network size (number of nodes) changes in the first six iterations of {\tt CoreExact}, on Ca-HepTh and As-Caida (Figure~\ref{fig:flowExp}). In the $x$-axis, ``--1" denotes that the flow network is constructed for the entire graph $G$, instead of a subgraph located by the ($k$, $\Psi$)-cores (Section~\ref{sec:basicExact}); ``0" means that the flow network is built on the subgraph located by the clique-cores.
The ($k$, $\Psi$)-cores are indeed effective for locating the CDS, as it greatly prunes vertices and clique instances.
The flow networks shrink, as the number of iterations increases.  After an iteration of the binary search is completed, a tighter lower bound of $\rho _{\rm{opt}}$ is obtained, which can then be used to locate the CDS in a smaller subgraph with a larger core number, resulting in a smaller flow network. For example, for the triangle on the Ca-HepTh dataset, over 95\% of the nodes in the flow network is pruned after six iterations. As the flow network is smaller, the minimum st-cut can be computed faster, thus yielding a better performance. As the clique size (i.e., $h$) increases, the proportion of cliques instances in the densest subgraph becomes larger, so the degree of pruning gets smaller.

\begin{figure}[]
\hspace*{-.15cm}
\centering
\begin{tabular}{c c}
  \begin{minipage}{3.601cm}
	\includegraphics[width=7.2cm]{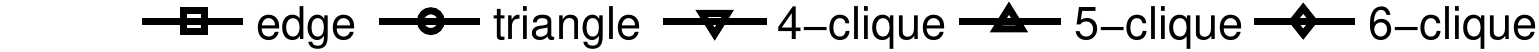}
  \end{minipage}
  &
  \\
  \begin{minipage}{3.601cm}
	\includegraphics[width=3.601cm]{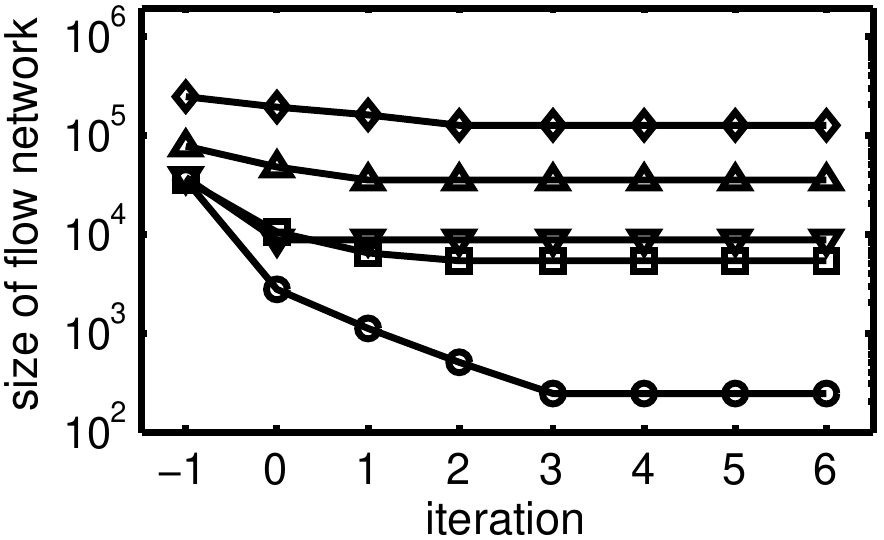}
  \end{minipage}
  &
  \begin{minipage}{3.601cm}
	\includegraphics[width=3.601cm]{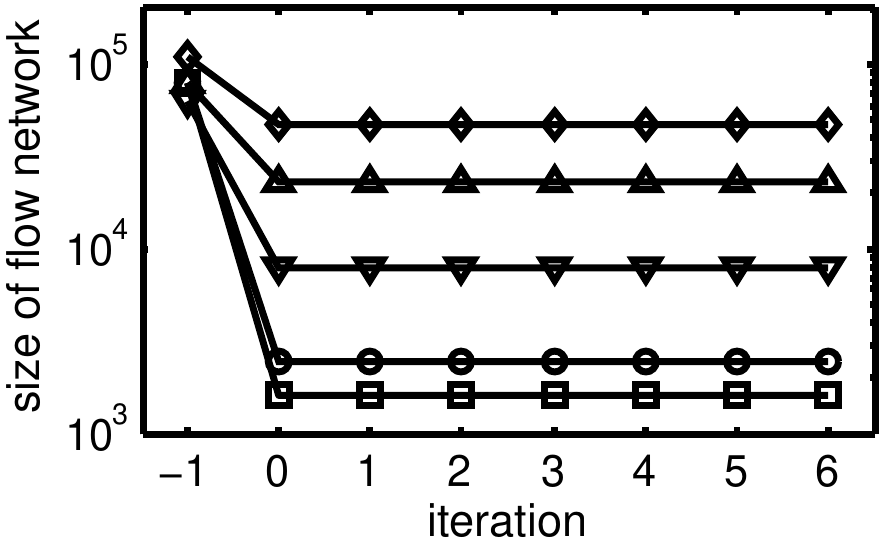}
  \end{minipage}
  \\
  (a) Ca-HepTh
  &
  (b) As-Caida
\end{tabular}
\vspace{-0.05in}
\caption{Flow network sizes in {\tt CoreExact}.}
\label{fig:flowExp}
\end{figure}

\begin{figure}[h]
\hspace*{-.15cm}
\centering
\begin{tabular}{c c}
  \begin{minipage}{3.801cm}
	\includegraphics[width=6.001cm]{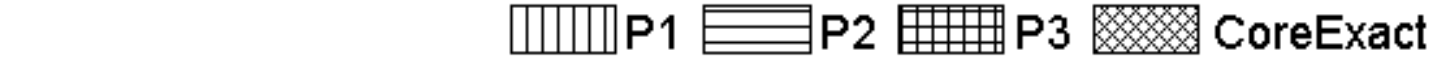}
  \end{minipage}
  &
  \\
  \begin{minipage}{3.601cm}
	\includegraphics[width=3.601cm]{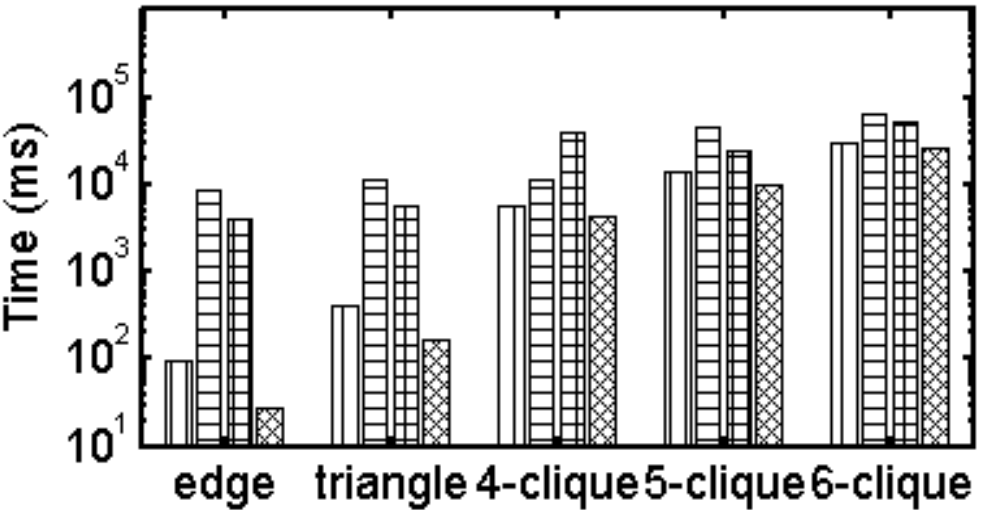}
  \end{minipage}
  &
  \begin{minipage}{3.601cm}
	\includegraphics[width=3.601cm]{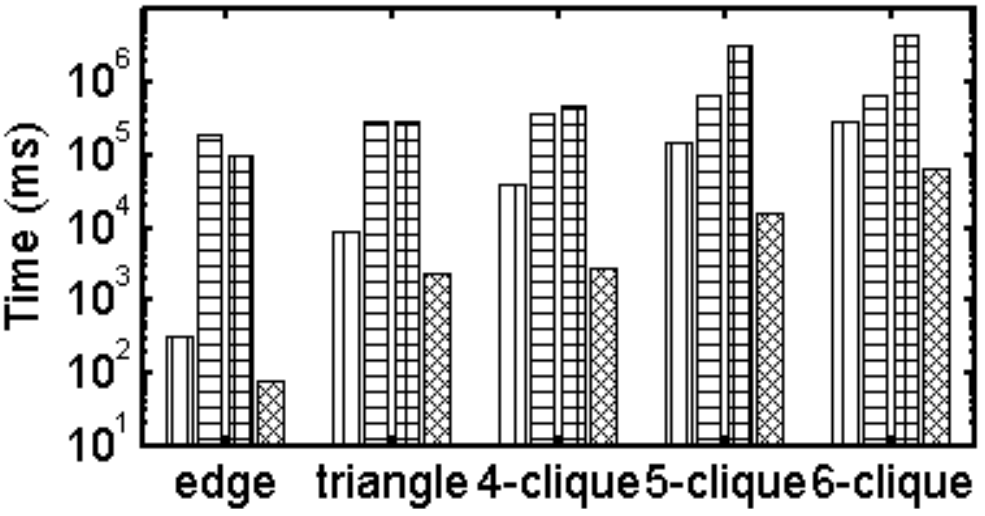}
  \end{minipage}
  \\
  (a) As-733
  &
  (b) Ca-HepTh
\end{tabular}
\vspace{-0.05in}
\caption{The effect of pruning criteria in {\tt CoreExact}.
\label{fig:pruningThree}}
\end{figure}

Next, we evaluate the individual effect of the three pruning criteria in {\tt CoreExact}. We create three variants of {\tt CoreExact}, namely $P1$, $P2$, and $P3$, which only include {\it Pruning1}, {\it Pruning2}, and {\it Pruning3} respectively, while other steps are the same as those of {\tt CoreExact}. Our experimental results (Figure \ref{fig:pruningThree}) confirm that each of the pruning strategies makes a contribution to the efficiency of {\tt CoreApp}. Most of the savings come from {\it Pruning1}; however, while the contribution of other pruning strategies is small on the As-733 and Ca-HepTh, {\it Pruning2} and {\it Pruning3} still make a non-trivial contribution on Ca-HepTh.

Finally, we examine the percentage of time cost of core decomposition in {\tt CoreExact}. As shown in Table \ref{tab:coreTime}, the percentage is small and decreases with the $h$-clique size. Besides, cores are effective for locating the CDS in some small subgraphs. Thus, {\tt CoreExact} achieves high efficiency, while incurring negligible overhead from core decomposition.

\begin{table}[ht]
  \centering
  \scriptsize
  %\small
  \caption {\% of time cost of core decomposition.}\label{tab:coreTime}
  \begin{tabular}{c|c|c|c|c|c}
     \hline
          {\bf Dataset}  & {\bf edge}
                         & {\bf triangle}
                         & {\bf 4-clique}
                         & {\bf 5-clique}
                         & {\bf 6-clique}\\
     \hline\hline
     As-733&    57.14\%&	8.28\%&	0.31\%&	0.09\%&	0.04\%\\
     \hline
     Ca-HepTh&  69.74\%&	6.01\%&	2.32\%&	0.87\%&	0.65\%\\
     \hline
  \end{tabular}
\end{table}

\noindent
{\bf\underline{2 Approximation algorithms.}} We next report the efficiency results of approximation solutions on the five largest datasets. From Figures~\ref{fig:exp-kds}(f)-(j), we observe that core-based approximation algorithms ({\tt IncApp} and {\tt CoreApp}) are consistently faster than {\tt Nucleus}
and {\tt PeelApp}~\footnote{For approximation algorithms, bars touching the upper boundaries mean that the corresponding algorithms cannot finish within 2 days.}. This implies that for decomposing cores, our algorithm (Algorithm \ref{alg:core}) which is almost the same as {\tt IncApp} is faster than the nucleus decomposition algorithm.
The average running time of {\tt IncApp} is only 90\% of that of {\tt PeelApp}. Both algorithms iteratively remove vertices from the graph $G$. Particularly, {\tt PeelApp} computes the density after removing each vertex, and only stops after $G$ has no more vertices.  However, {\tt IncApp} does not compute the density, and stops after the ($k_{\max}$, $\Psi$)-core is discovered. {\tt CoreApp} performs the best, as it finds the ($k_{\max}$, $\Psi$)-core in a top-down manner, and skips the computation of cores with smaller clique-core numbers. In our experiments, {\tt CoreApp} is up to three and two orders of magnitude faster than {\tt Nucleus} and {\tt PeelApp} respectively.

\begin{figure}
        \centering
        \begin{tabular}{c c}
          \begin{minipage}{3.60cm}
        	\includegraphics[width=5.80cm]{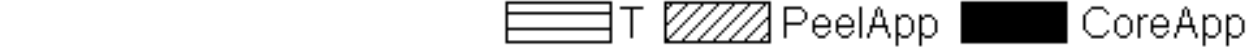}
          \end{minipage}
          &
          \\
          \begin{minipage}{3.60cm}
        	\includegraphics[width=3.60cm]{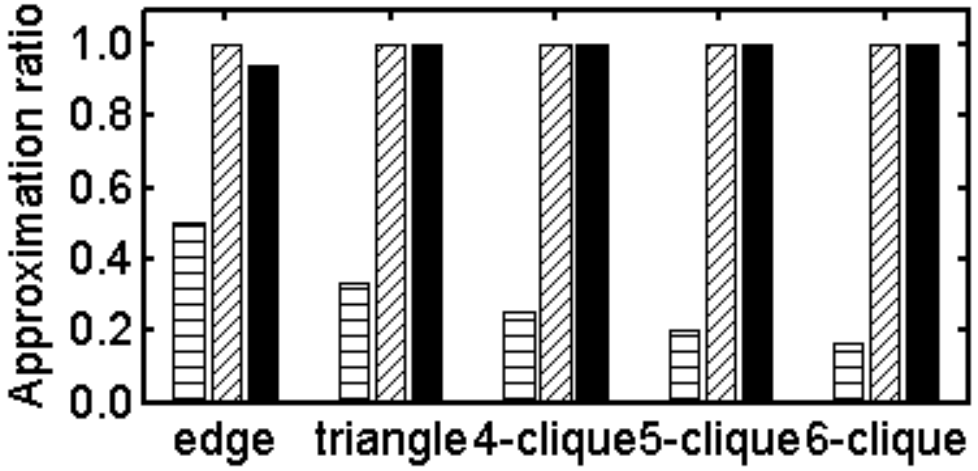}
          \end{minipage}
          &
          \begin{minipage}{3.60cm}
        	\includegraphics[width=3.60cm]{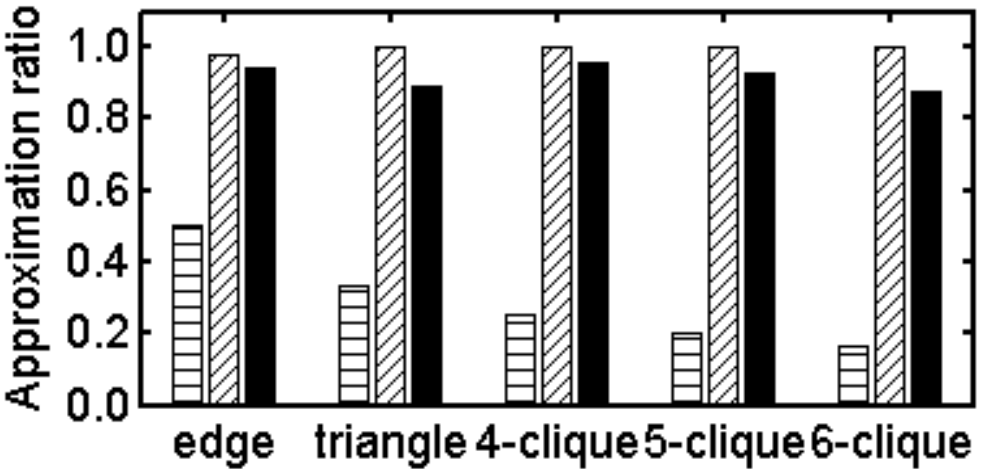}
          \end{minipage}
          \\
          (a) Netscience
          &
          (b) As-Caida
        \end{tabular}
        \vspace{-0.10in}
        \caption{Approximation ratio.}
        \label{fig:approxExp}
\end{figure}

As the clique size ($h$) increases, the speedup of {\tt CoreApp} over {\tt PeelApp} decreases, because the proportion of clique instances in the densest subgraph becomes larger, increasing the time cost of computing ($k_{\max}$, $\Psi$)-core.
Meanwhile, the running time generally grows as the clique size ($h$) increases, except for the Cit-Patents dataset. This is because on Cit-Patents, the numbers of 5-cliques and 6-cliques are less than the number of 4-cliques.
In addition, we compare {\tt CoreApp} with {\tt EMcore} for computing approximate EDS's on five largest datasets. As reported in Table \ref{tab:emcore}, {\tt EMcore} is slower than {\tt CoreApp}, because it differs with {\tt CoreApp} on four aspects as discussed in Section \ref{sec:coreAppAlgo}.

\begin{table}[ht]
  \centering
  \scriptsize
  \vspace{-0.1in}
  \caption {Efficiency of {\tt EMcore} and {\tt CoreApp} (seconds).}
  \label{tab:emcore}
  \begin{tabular}{c|c|c|c|c|c}
     \hline
          Algo.  & DBLP & CitPatents & FriendSter & Enwiki-2017 & UK-2002\\
     \hline\hline
     {\tt EMcore}&   0.091&	1.132&  3.143& 8.543 & 7.543\\
     \hline
     {\tt CoreApp}&  0.077&	1.021&	2.986& 8.139 & 5.825\\
     \hline
  \end{tabular}
\end{table}

We next report the theoretical ratio $T$ (i.e., $\frac{1}{|V_\Psi|}$) and actual approximation ratios $R$ of approximation methods.
Since {\tt Nucleus}, {\tt IncApp} and {\tt CoreApp} return the same ($k_{\max}$, $\Psi$)-core, their $R$ values are the same, so we only show results for {\tt CoreApp}. As shown in Figure~\ref{fig:approxExp}, $R$ is often larger than $T$. Although {\tt CoreApp} is slightly worse than {\tt PeelApp} on 6/10 instances (the average ratio of {\tt CoreApp} is 0.956 times that of {\tt PeelApp}), they have the same theoretical guarantee and their actual ratios are close to 1.0 in most cases, so {\tt CoreApp} produces high-quality results in practice.

In addition, we compare the efficiency of core-based exact and approximation approaches on two datasets. As shown in Figure~\ref{fig:cmpExactApp}, {\tt CoreApp} is much faster than {\tt CoreExact}. The reason is that {\tt CoreExact} relies on not only core decomposition, but also computing the minimum st-cut from flow networks using binary search, whereas {\tt CoreApp} just computes the ($k_{\max}$, $\Psi$)-core directly.

\noindent
{\bf Remark.}  For small-to-moderate-sized graphs (e.g., Ca-HepTh), {\tt CoreExact} is the best choice, as it computes an exact result in a reasonable time.  For larger graphs (e.g., UK-2002), {\tt CoreApp} is a much better option since it achieves high accuracy and efficiency.

\begin{figure}[ht]
        \centering
        \begin{tabular}{c c}
          \begin{minipage}{3.60cm}
        	\includegraphics[width=5.40cm]{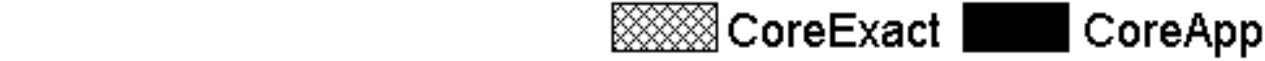}
          \end{minipage}
          &
          \\
          \begin{minipage}{3.60cm}
        	\includegraphics[width=3.60cm]{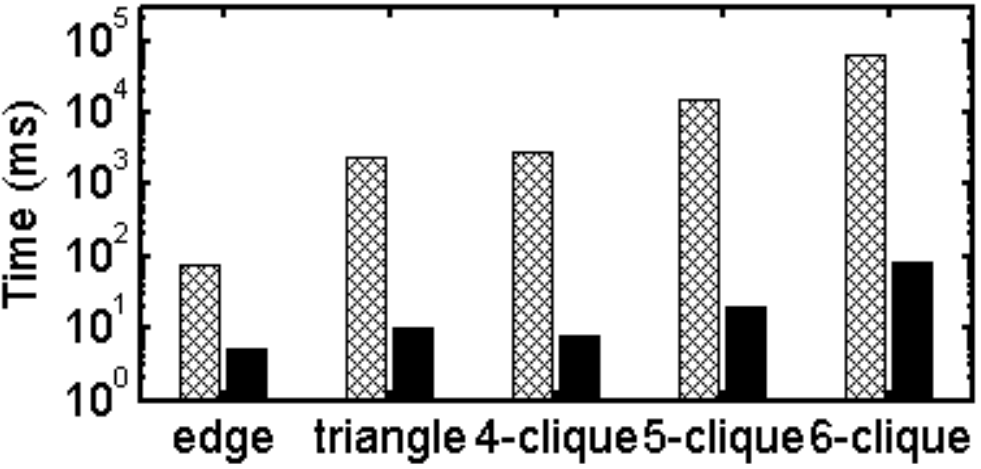}
          \end{minipage}
          &
          \begin{minipage}{3.60cm}
        	\includegraphics[width=3.60cm]{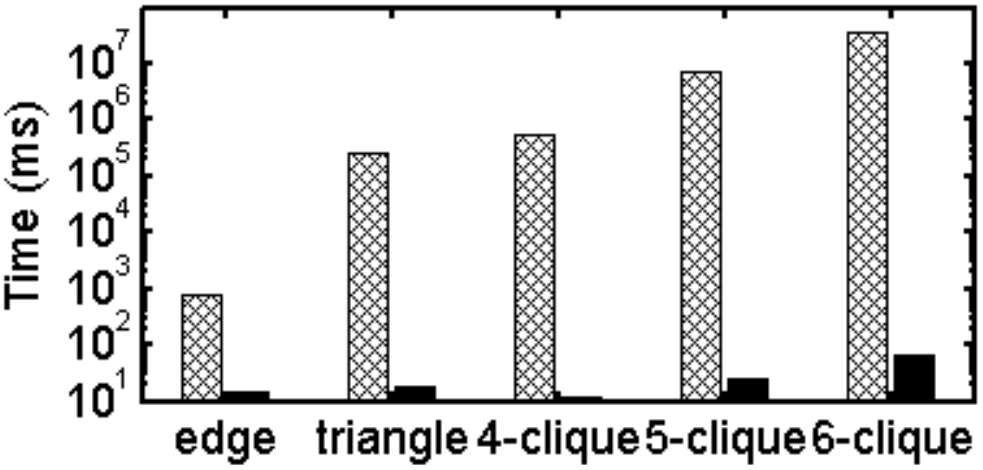}
          \end{minipage}
          \\
         (a) Ca-HepTh
          &
         (b) As-Caida
        \end{tabular}
        \vspace{-0.05in}
        \caption{{\tt CoreExact} and {\tt CoreApp}.}
        \label{fig:cmpExactApp}
\end{figure}

\begin{figure*}
   \hspace*{-.6cm}
   \centering
   \begin{minipage}[t]{0.59\linewidth}
        \centering
        \begin{tabular}{c c c}
          &
          \begin{minipage}{3.30cm}
        	\includegraphics[width=3.40cm]{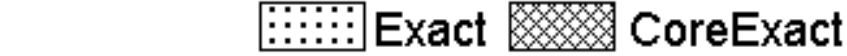}
          \end{minipage}
          &
          \\
          \begin{minipage}{3.30cm}
        	\includegraphics[width=3.45cm]{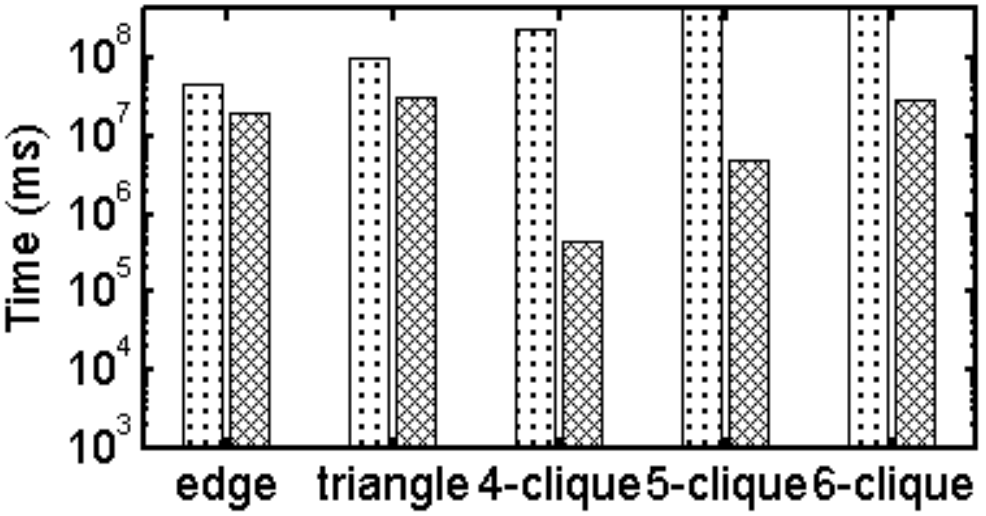}
          \end{minipage}
          &
          \begin{minipage}{3.30cm}
        	\includegraphics[width=3.45cm]{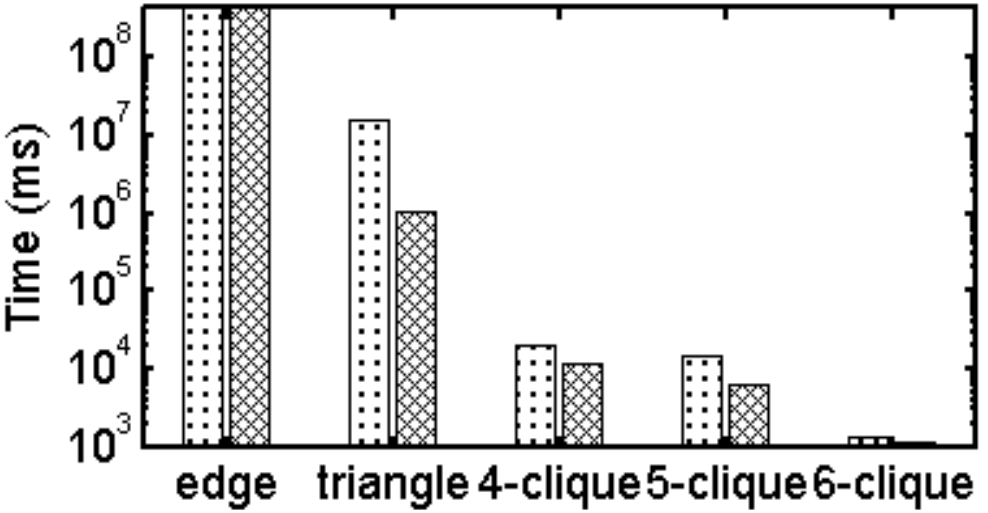}
          \end{minipage}
          &
          \begin{minipage}{3.30cm}
        	\includegraphics[width=3.45cm]{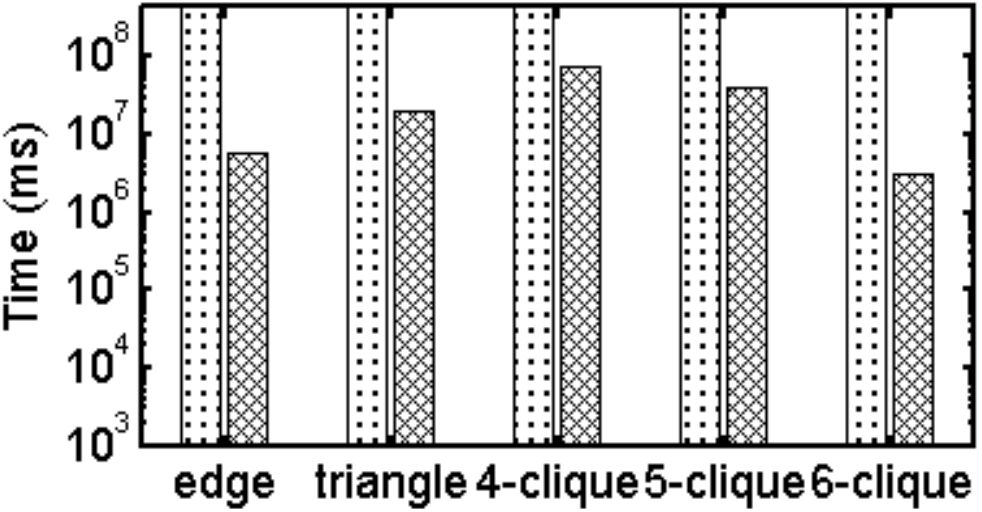}
          \end{minipage}
          \\
          (a) SSCA
          &
          (b) ER
          &
          (c) R-MAT
        \end{tabular}
        \vspace{-0.15in}
        \caption{Efficiency of exact CDS algorithms on random graphs.}
        \label{fig:RandExact}

        \centering
        \begin{tabular}{c c c}
          &
          \begin{minipage}{3.30cm}
	        \includegraphics[width=5.001cm]{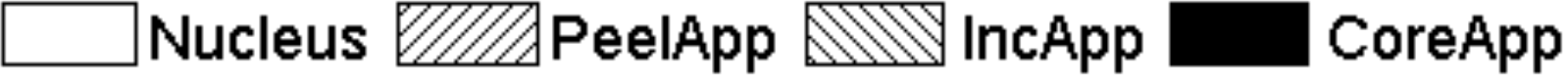}
          \end{minipage}
          &
          \\
          \begin{minipage}{3.30cm}
        	\includegraphics[width=3.45cm]{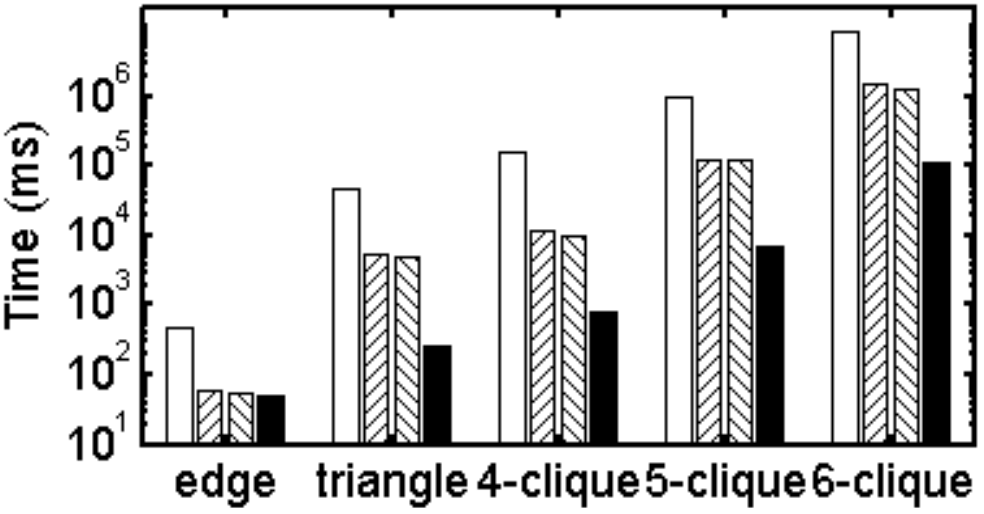}
          \end{minipage}
          &
          \begin{minipage}{3.30cm}
        	\includegraphics[width=3.45cm]{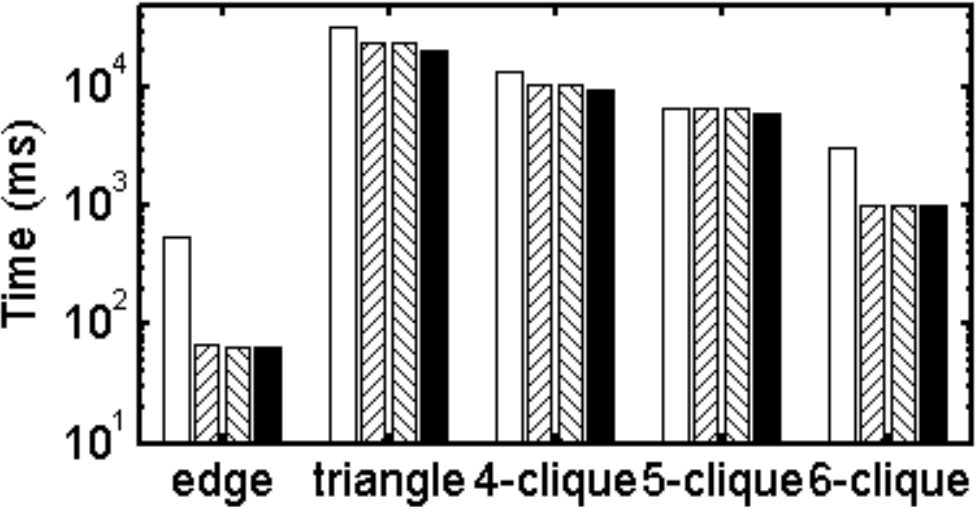}
          \end{minipage}
          &
          \begin{minipage}{3.30cm}
        	\includegraphics[width=3.45cm]{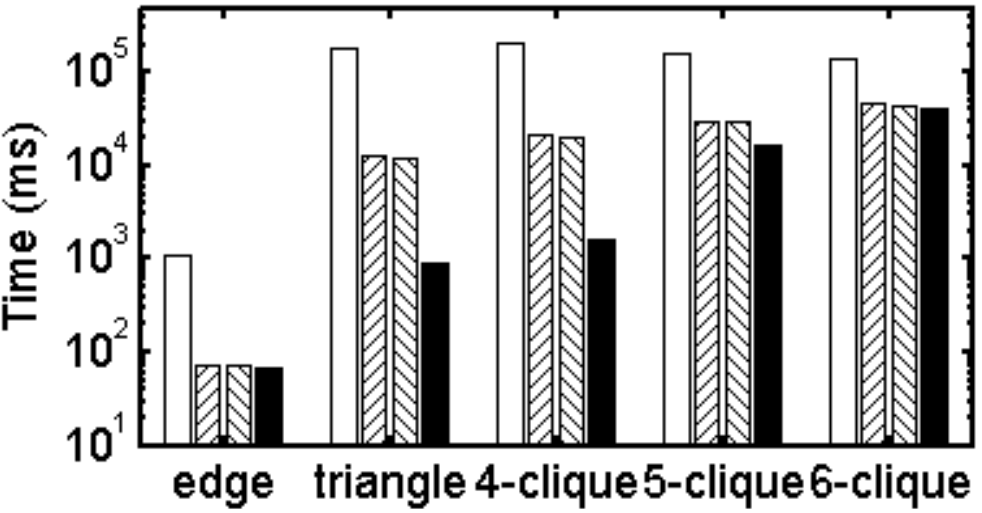}
          \end{minipage}
          \\
          (a) SSCA
          &
          (b) ER
          &
          (c) R-MAT
        \end{tabular}
        \vspace{-0.15in}
        \caption{Efficiency of approximation CDS algorithms on random graphs.}
        \label{fig:RandApp}
   \end{minipage}
   \hspace{3ex}
   \begin{minipage}[t]{0.39\linewidth}
      \centering
      \begin{tabular}{c c}
        \begin{minipage}{3.30cm}
        	\includegraphics[width=5.10cm]{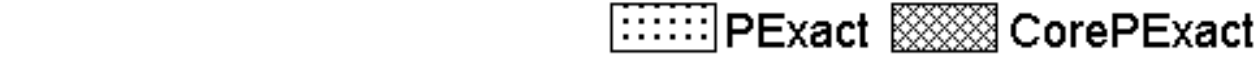}
        \end{minipage}
        &
        \\
        \begin{minipage}{3.30cm}
	       \includegraphics[width=3.45cm]{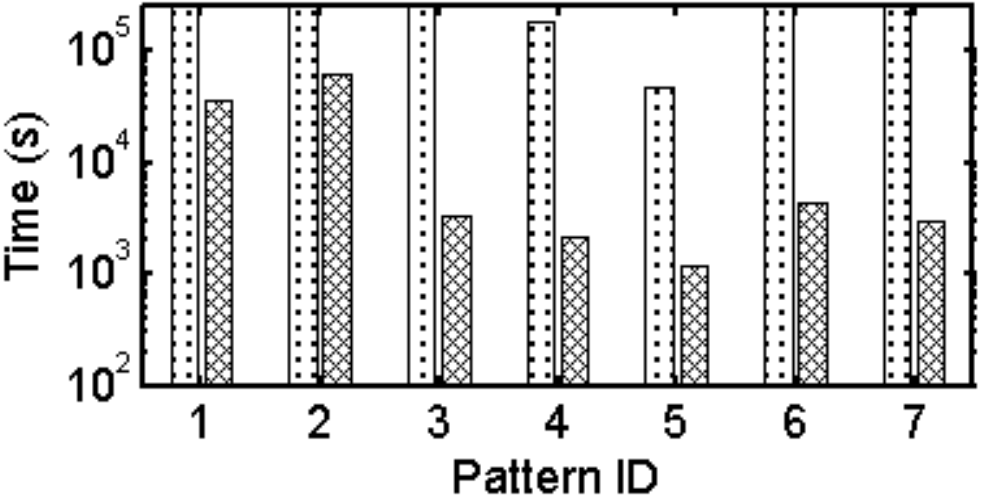}
        \end{minipage}
        &
        \begin{minipage}{3.30cm}
	       \includegraphics[width=3.45cm]{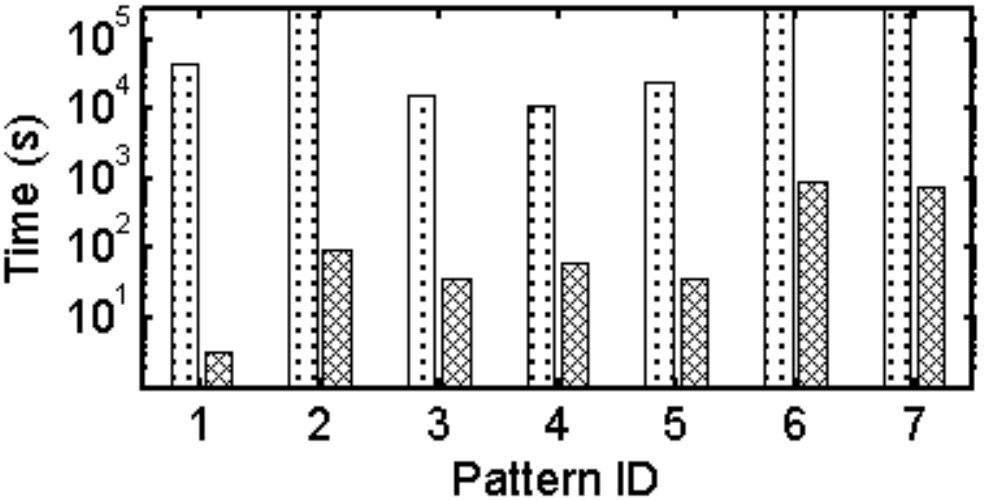}
        \end{minipage}
        \\
        (a) As-733
        &
        (b) Ca-HepTh
      \end{tabular}
      \vspace{-0.15in}
      \caption{Efficiency of exact PDS algorithms.}
      \label{fig:exp-mds-exact}

      \begin{tabular}{c c}
        \begin{minipage}{3.30cm}
	       \includegraphics[width=5.601cm]{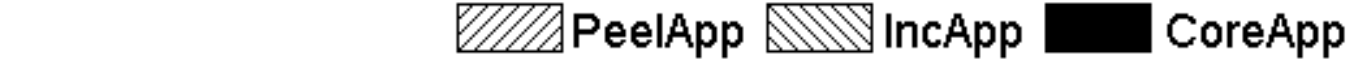}
        \end{minipage}
        &
        \\
        \begin{minipage}{3.30cm}
	       \includegraphics[width=3.45cm]{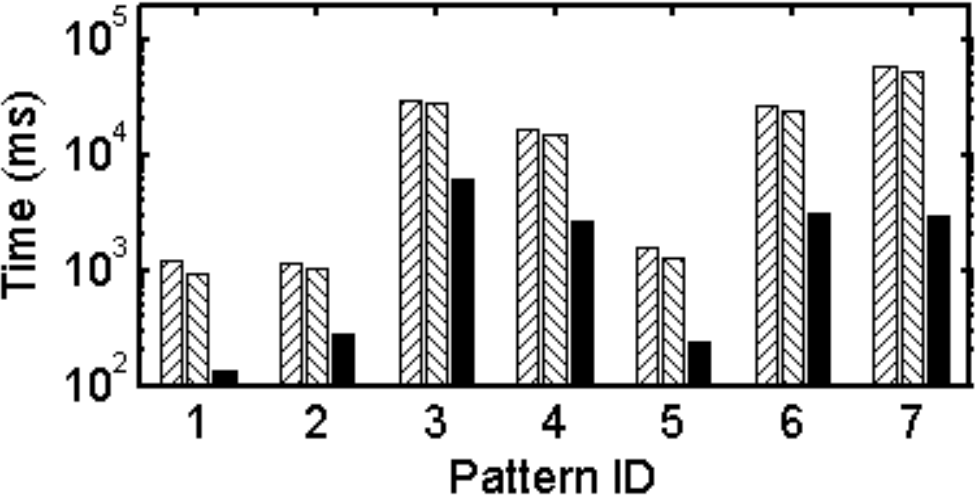}
        \end{minipage}
        &
        \begin{minipage}{3.30cm}
	       \includegraphics[width=3.45cm]{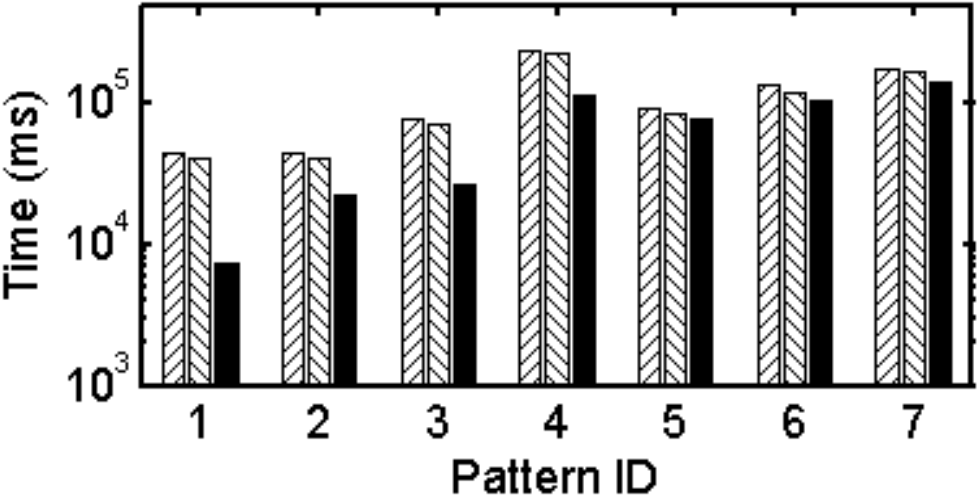}
        \end{minipage}
        \\
        (a) DBLP
        &
        (b) Cit-Patents
      \end{tabular}
      \vspace{-0.15in}
      \caption{Efficiency of approx. PDS algorithms.}
      \label{fig:exp-mds-app}
   \end{minipage}
\end{figure*}

\begin{table*}[ht]
  \hspace*{-4.9cm}
  \scriptsize
  %\small
  \vspace{-0.2in}
  \centering
  \caption {The edge-densities and clique-densities (pattern-densities) of CDS's (PDS's).}
  \label{tab:density}
  \begin{tabular}{c|c|c|c|c|c|c|c|c|c|c|c|c|c}
     \hline
     \multirow{2}{*}{\bf Dataset}&
     \multicolumn{1}{c|}{\bf edge}&
     \multicolumn{2}{c|}{\bf triangle}&
     \multicolumn{2}{c|}{\bf 4-clique}&
     \multicolumn{2}{c|}{\bf 5-clique}&
     \multicolumn{2}{c|}{\bf 6-clique}&
     \multicolumn{2}{c|}{\bf 2-star}&
     \multicolumn{2}{|c}{\bf diamond}\\
     \cline{2-14}
     & $\rho_{opt}$
     & $\rho_{opt}$& $\rho$(EDS,$\Psi$)
     & $\rho_{opt}$& $\rho$(EDS,$\Psi$)
     & $\rho_{opt}$& $\rho$(EDS,$\Psi$)
     & $\rho_{opt}$& $\rho$(EDS,$\Psi$)
     & $\rho_{opt}$& $\rho$(EDS,$\Psi$)
     & $\rho_{opt}$& $\rho$(EDS,$\Psi$)  \\
     \hline\hline
       S-DBLP&6
       &	22&	22
       &	55&	55
       &    99& 99
       &    132& 132
       &    73.5&	66
       &	165&	165\\
     \hline
       Yeast&3.13
       &	2.11&	0.467
       &    0.67&	0.0
       &    0.0& 0.0
       &    0.0& 0.0
       &	111.3&	18.13
       &	20&	19.2\\
     \hline
       Netscience&9.50
       &	57.25&	57.25
       &	242.3&	242.3
       &    775.2&  775.2
       &    1938&   1938
       &	171&	171
       &	726.8&	726.8\\
     \hline
       As-733&8.19
       &	31.43&	31.35
       &	68.67&	67.94
       &    92.78&    90.23
       &    79.37&    75.13
       &	826.3&  153.8
       &	3376&	437.7\\
     \hline
  \end{tabular}
\end{table*}

\noindent{{\bf\underline{3 Random graphs.}}
As depicted in Figures \ref{fig:RandExact} and \ref{fig:RandApp}, for SSCA and R-MAT, the performance of our proposed solution
is generally satisfactory. For example, the running time of {\tt CoreApp} is 20 (resp., 201) times faster than {\tt PeelApp} in SSCA (resp., R-MAT) when $\Psi$ is the triangle. For ER, the degree values of vertices are almost the same, and the $k_{\max}$-core contains 96.8\% of the vertices in the graph. This affects the pruning effectiveness of {\tt CoreApp}, rendering a lower performance gain. All in all, our core-based algorithms favor real-world graphs.

\noindent
{\bf\underline{4 Densities of CDS's.}}
We next show the clique-densities of CDS's for different $h$-cliques ($h\geq3$). Specifically, for each dataset, we first use {\tt CoreExact} to compute its exact CDS's for different cliques, then compute the $h$-clique-densities of its EDS, and finally report the $h$-clique-densities of its EDS and CDS's in Table \ref{tab:density}.
Due to the space limitation, we only show the results on four small datasets (where S-DBLP is a sub-graph of the DBLP dataset used in Section \ref{sec:expMDS}).
We remark that for Yeast dataset, the EDS does not contain any 4, 5, 6-clique, so its $h$-clique-density is 0.0 ($h$$\geq$4).
As we can see, for S-DBLP and Netscience, their CDS's are exactly the same as EDS. In fact, they are the maximal clique in the graph, which confirms the conclusion that CDS's can be used for identifying large near-cliques \cite{tsourakakis2015k}.
For Yeast and As-733, the clique-density values of CDS's are higher than those on the EDS.

%% file: expPDS.tex
\subsection{DSD for Pattern-Densities}
\label{sec:expMDS}

Next, we present the results for general patterns in Figure~\ref{fig:7motifs}. For lack of space, we only report results on a subset of datasets. In addition, we perform case studies on real datasets for these patterns.

\noindent
{\bf\underline{1 Exact algorithms.}} In Figure~\ref{fig:exp-mds-exact}, we present the efficiency results of exact algorithms on two small datasets As-733 and Ca-HepTh. The bars touching the top of the figures mean that the corresponding algorithms cannot find densest subgraphs within 3 days, at which point we time them out.  We can see that {\tt CorePExact} is up to four orders of magnitude faster than {\tt PExact}.  For different patterns, their running times vary, because the number of pattern instances in the underlying graph for each pattern can be very different.  For any two patterns $\Psi_1$ and $\Psi_2$ which are not ``special patterns'' (e.g., star and loop), we observe that if $|V_{\Psi_1}|$=$|V_{\Psi_2}|$ and $\Psi_1\subseteq\Psi_2$, then it takes longer to find the densest subgraph w.r.t. $\Psi_1$ than w.r.t. $\Psi_2$.  This is because the number of pattern instances of $\Psi_1$ is more than that of  $\Psi_2$. For example, {\it c3-star} is a subgraph of {\it 2-triangle} (with 4 vertices) and it takes more time to find the densest subgraph w.r.t. {\it c3-star} than {\it 2-triangle}.

\noindent
{\bf\underline{2 Approximation algorithms.}}  As shown in Figure~\ref{fig:exp-mds-app}, the running time of an approximation algorithm increases with the graph size in general.  This is because computing the cores is more expensive for a larger graph. Again, {\tt CoreApp} performs the fastest, and it is up to two orders of magnitude faster than {\tt PeelApp}.  For special patterns (star and diamond), we use optimized algorithms (details are in~\cite{fullVersion}) for core decomposition.  Hence, they need less time cost than other more complicated patterns (e.g., 2-triangle).

\noindent
{\bf\underline{3 Case studies.}} We use two real graphs, namely {\it S-DBLP} and {\it Yeast}. S-DBLP ($|V|$=478, $|E|$=1,086) is a sub-graph of the DBLP dataset. It is the co-authorship network of authors who published at least two DB/DM papers between 2013 and 2015. We consider two 3-vertex patterns, i.e., {\it triangle} and {\it 2-star} (Figure~\ref{fig:7motifs}).
We use the exact algorithm to compute their PDS's, as depicted in Figure~\ref{fig:dblp}.
In a {\it triangle} pattern, every pair of vertices is connected, so the PDS tends to be a near-clique~\cite{tsourakakis2015k}.
The researchers involved in this PDS possess a close collaboration relationship: any two researchers have published papers together. The PDS for {\it 2-star} is quite different from that of {\it triangle}. Particularly, researchers in the ``central'' part of the PDS formed by {\it 2-star} tend to be group directors or senior researchers (e.g., Profs. Jiawei Han and  Chengxiang Zhai), who are linked to their former students or postdocs. For this PDS, over half of the researchers worked in Prof. Han's lab before. Similarly, for Yeast, different PDS's can capture different semantics~\cite{fullVersion}.

\begin{figure}
    \hspace*{-.2cm}
	\centering
    \includegraphics[width=0.9\linewidth]{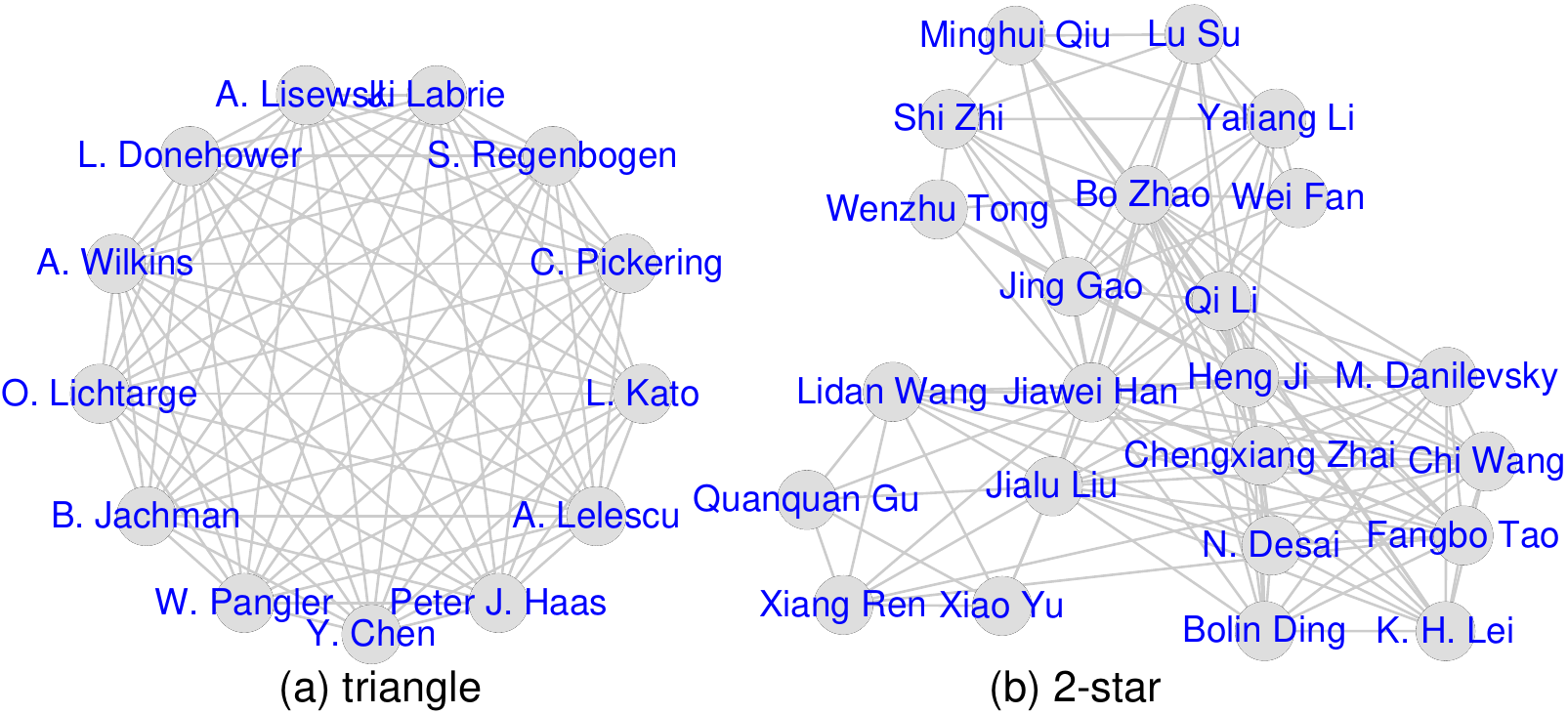}
    \vspace{-0.1in}
    \caption{The densest subgraphs found in DBLP network, based on triangle and 2-star patterns.}
    \label{fig:dblp}
\end{figure}

\noindent
{\bf\underline{4 Densities of PDS's.}}
In this experiment, we analyze the pattern-densities of PDS's for different patterns. Again, for each dataset, we first compute its exact EDS and PDS's for all patterns, and then report the pattern-densities of its EDS and PDS's in Table~\ref{tab:density}. Due to the space limitation, we only show results of 2-star and diamond. As we can observe, for most of the datasets, the pattern-density values of PDS's are higher than those on the EDS.

%% file: conclusion.tex
\section{Conclusions}
\label{sec:conclusion}

The densest subgraph discovery (DSD) problem is fundamental to many graph applications.  In this paper, we develop new algorithms to discover edge- and $h$-clique-based densest subgraphs, which are well studied in the literature.  Our main observation is that densest subgraphs can be derived efficiently from $k$-cores. We extend $k$-core to ($k$, $\Psi$)-core by incorporating an $h$-clique $\Psi$. Based on ($k$, $\Psi$)-cores, we develop core-based exact and approximation solutions to the DSD problem. Moreover, we generalize the edge- and $h$-clique-density to pattern-density and show that our solutions can be easily adapted for finding pattern-density-based densest subgraphs. Extensive experiments show that our exact (resp., approximation) ``core-based solutions'' outperform existing algorithms by up to four orders (resp., two orders) of magnitude.

In the future, we will attempt to derive even tighter bounds for densities of ($k$, $\Psi$)-cores. We will also extend our core-based algorithms for finding densest subgraphs with size constraints. Another interesting research direction is to exploit our core-based techniques to speed up the randomized approximation algorithm in \cite{mitzenmacher2015scalable}.

%In the future, we will attempt to derive even tighter bounds for densities of ($k$, $\Psi$)-cores. We will also include distributed implementation of our proposed algorithms. Another interesting research direction is to extend our core-based algorithms for solving the densest subgraphs with size constraints.

%% file: appendix.tex
\appendix

\section{Dataset Statistics}
\label{app:statistics}

\begin{figure}[ht]
	\centering
	\includegraphics[width=0.96\linewidth]{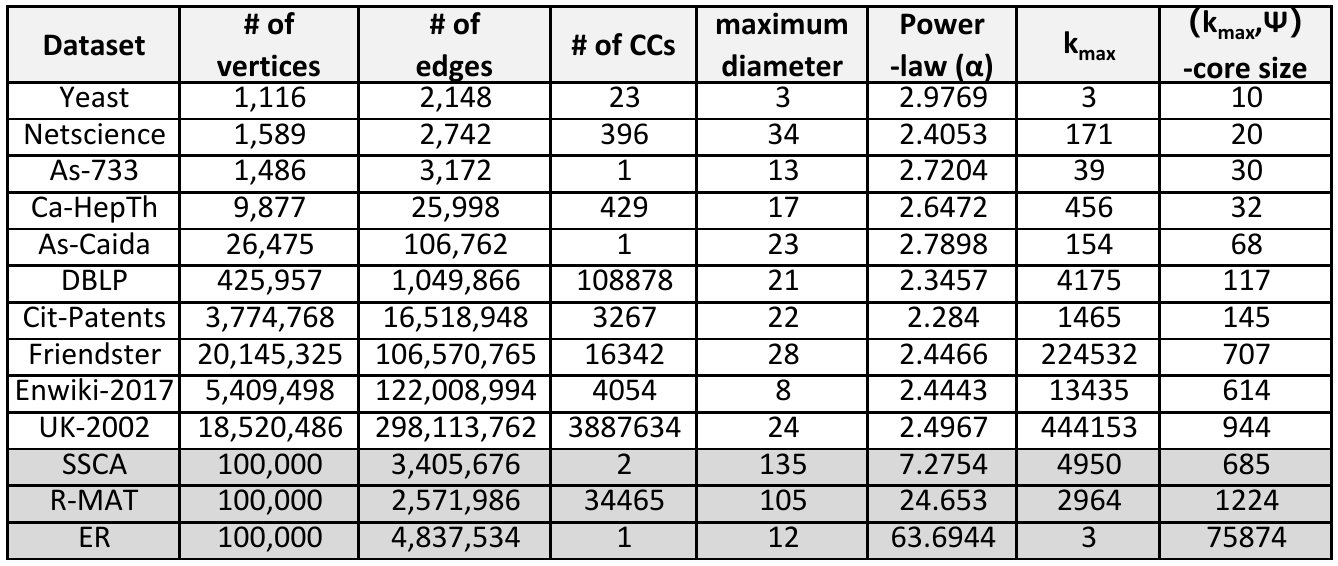}
    \vspace{-0.15in}
	\caption{Characteristics of the networks in our experiments.}
	\label{fig:statistics}
\end{figure}

In this section, we analyze the properties of the graph datasets used, according to their number of vertices and edges, the number of connected components (\# of CCs), diameter, decay factor $\alpha$ of the power law distribution (where $\alpha$ in $f(x)$=$x^{-\alpha}$), $k_{\max}$, and size of ($k_{\max}$, $\Psi$)-core (where $\Psi$ is triangle).  Fig.~\ref{fig:statistics} shows these statistics.  We can see that the networks have a variety of characteristics. For example, the number of CCs ranges from 1 to 3.8M; the diameter varies from 3 to 135; the value of $\alpha$ is between 2.28 and 63.7. To conclude, the graphs we used exhibit a wide range of characteristics.

%In the last two columns of Fig.~\ref{fig:statistics}, we show the performance gain of our methods ({\tt CoreExact} and {\tt CoreApp}) over existing ones ({\tt Exact} and {\tt PeelApp}), where ``/'' means {\tt Exact} fails to obtain CDS's within five days. We use {\it Exact/CoreExact} ({\it PeelApp/ CoreApp}) to denote the ratio of the running time of {\tt Exact} ({\tt PeelApp}) to that of {\tt CoreExact} ({\tt CoreApp}).  We can see that our solutions are generally better than existing ones. Particularly, {\tt CoreExact} is one to three orders of magnitude faster than {\tt Exact}; {\tt CoreApp} is one to two orders of magnitude faster than {\tt PeelApp}. The performance gain is relatively lower for Yeast, Netscience, and ER.  For Yeast and Netscience, they are small graphs (with about $1K$ vertices), so our more powerful pruning strategies do not prune a lot more vertices than existing ones. For ER, which is a synthetic dataset, the degree values of vertices are almost the same, and the size of ($k_{\max}$, $\Psi$)-core and the value of $\alpha$ are both large. These affect the pruning effectiveness of {\tt CoreExact} and {\tt CoreApp}.  To summarize, our approaches are better for graphs whose maximum core numbers (i.e., $k_{\max}$) are large and the ($k_{\max}$, $\Psi$)-cores are small. This is because the densest subgraphs are in cores with large core numbers, allowing vertices with smaller core numbers to be quickly pruned.

\section{Additional Proofs}
\label{sec:appProofs}

\input{proof-sec4}
\input{proof-sec5.2}
\input{proof-sec7.2}

\input{pexact}
\input{kcoreSpec}

\section{Results on Additional Datasets}
\label{sec:addExp}

Table \ref{tab:moreDatasets} shows the additional three real datasets used in our experiments. The efficiency results are presented in Figure \ref{fig:addExp}. We can observe that the results are highly similar to those presented in the main paper, so we skip the detailed description.

\begin{table}[]
  \centering
  \scriptsize
  \vspace{-0.10in}
  \caption {Additional Datasets.}
  \label{tab:moreDatasets}
  \begin{tabular}{c|r|r}
     \hline
         \textbf{Name}
                         & \textbf{Vertices}
                         & \textbf{Edges}\\
     \hline\hline
          Flickr        &  214,698  &  2,096,306\\
     \hline
          Google        &  875,713&  4,322,051\\
     \hline
          Foursquare    &  2,127,093&  8,640,352\\
     \hline
  \end{tabular}
\end{table}

\begin{figure*}[]
\hspace*{-.10cm}
\centering
\begin{tabular}{c c c}
  &
  \begin{minipage}{3.30cm}
	\includegraphics[width=5.1cm]{legend-rand-app}
  \end{minipage}
  &
  \\
  \begin{minipage}{3.601cm}
	\includegraphics[width=3.601cm]{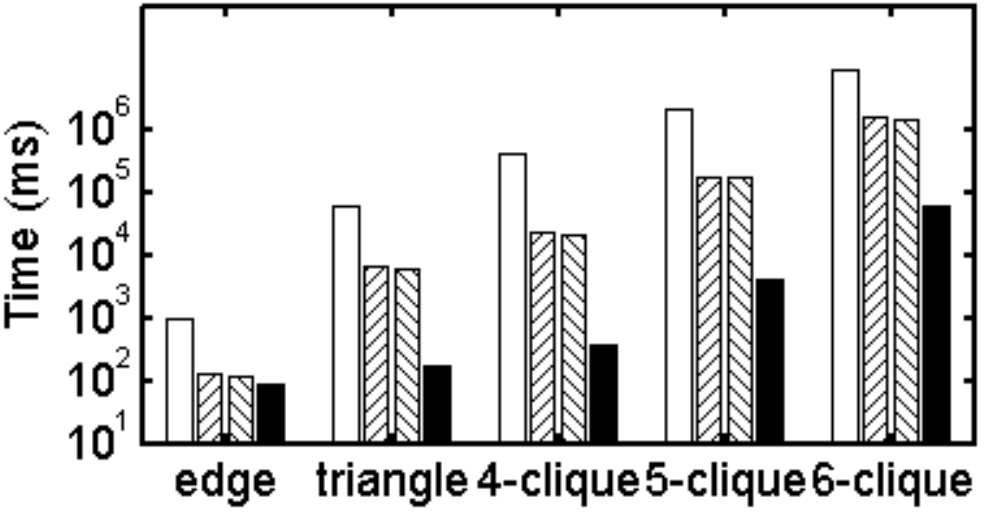}
  \end{minipage}
  &
  \begin{minipage}{3.601cm}
	\includegraphics[width=3.601cm]{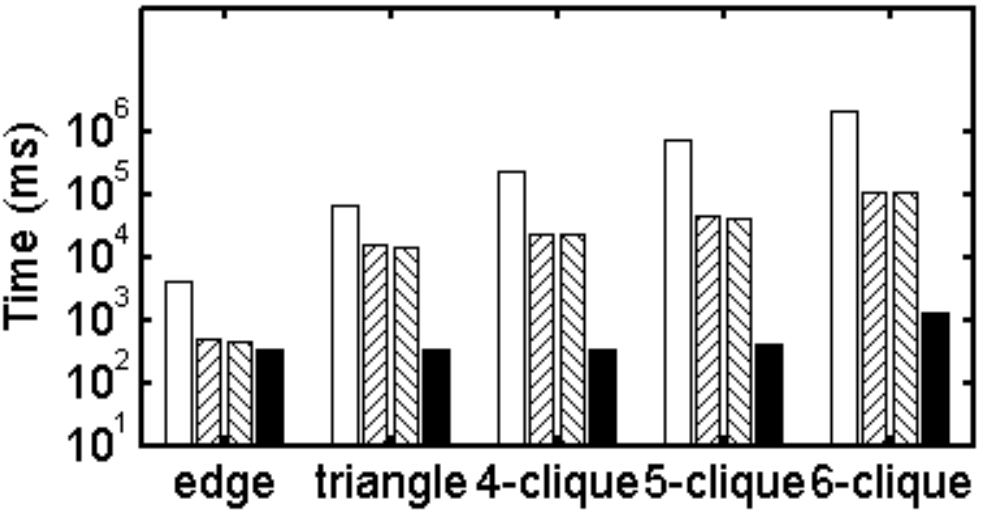}
  \end{minipage}
  &
  \begin{minipage}{3.601cm}
	\includegraphics[width=3.601cm]{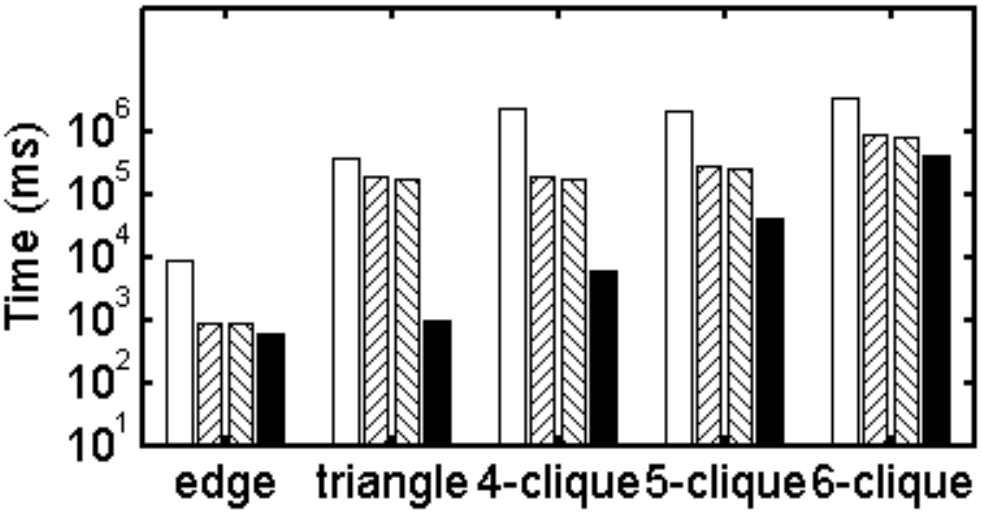}
  \end{minipage}
  \\
  (a) Flickr
  &
  (b) Google
  &
  (c) Foursquare
\end{tabular}
\vspace{-0.10in}
\caption{Efficiency of approximation CDS algorihms.}
\label{fig:addExp}
\end{figure*}

\input{caseStudy}

%% file: proof-sec4.tex
\subsection{Proofs for Section~\ref{sec:baseline}}
\label{sec:proof-sec4}

\noindent\textsc{Lemma}~\ref{lemma:exactTime}.
\emph{Given a graph $G(V$,$E)$ and an $h$-clique $\Psi(V_\Psi$,$E_\Psi)$, {\tt Exact} takes
${\mathcal O}$$\left( {n\cdot{d-1 \choose h-1}}+ (n|\Lambda|+\min {{(n,|\Lambda|)}^3})\log n\right)$ time and $\mathcal{O}$ $(n+|\Lambda|)$ space, where $\Lambda$ is set of ($h-$1)-clique instances in $G$ \cite{tsourakakis2015k}.
}

\begin{proof}
To collect the $h$-clique instances, for each vertex $v$, we compute the number of clique instances it involves. In the worst case, any $h$--1 neighbors of $v$ can form an $h$-clique with $v$, so the maximum number of clique instances it involves is ${d-1 \choose h-1}$. As a result, collecting all the instances of the $h$-clique and building the flow network takes $O(n\cdot{d-1 \choose h-1})$.
The number of binary search queries can be bound by ${\mathcal O}\left(\left\lceil {\log ({n^{|V_\Psi|}} \cdot n \cdot (n - 1))} \right\rceil\right)$=${\mathcal O}(\log n)$.
In each binary search query, we adopt the Gusfield's algorithm~\cite{ahuja1994improved} to compute the minimum st-cut of the flow network and its time cost is ${\mathcal O}(n \cdot |\Lambda | + \min{(n,|\Lambda |)^3})$, where $|\Lambda|$ denotes the number of clique instances.
In addition, to compute the minimum st-cut, the space cost of is linear to the size of flow network,
i.e., $\mathcal{O}(n+|\Lambda|)$.
Therefore, the lemma holds.
\end{proof}

\noindent\textsc{Lemma}~\ref{lemma:peelAppTime}.
\emph{Given a graph $G$ and an $h$-clique $\Psi(V_\Psi,$$E_\Psi)$, then {\tt PeelApp} takes ${\mathcal O}\left( {n \cdot {{d-1} \choose {h-1}}}\right)$ time and ${\mathcal O}\left(m\right)$ space \cite{tsourakakis2015k}.}

\begin{proof}
For each vertex, computing its clique-degree can be completed in ${\mathcal O}\left( {{d-1 \choose h-1}}\right)$ time. After removing a vertex $v$, we need to update the clique-degrees of vertices which share at least one clique instance with $v$, which can be completed in ${\mathcal O}\left( {{d-1 \choose h-1}}\right)$ time.
After computing the clique-degrees, the clique-density of each residual subgraph can be updated in constant time.
Note that computing the clique-degree of $v$ takes $m$ space as we sequentially compute the clique instances.
Hence, the lemma holds.
\end{proof}

%% file: proof-sec5.2.tex
\subsection{Proofs for Section~\ref{sec:coredensity}}
\label{sec:proof-sec5.2}

\noindent\textsc{Lemma}~\ref{lemma:connectivity}.
\emph{Given a graph $G$ and an $h$-clique $\Psi$, the connected components of CDS $D$ have the same clique-density.}

\begin{proof}
We first consider the case that $D$ has two connected components $C_1(V_{C_1}, E_{C_1})$ and $C_2(V_{C_2}, E_{C_2})$, and assume that their clique-densities are not the same.
Then, we can easily conclude
\begin{equation}
{\rho _{opt}} = \frac{{\mu ({C_1},\Psi ) + \mu ({C_2},\Psi )}}{{|{V_{{C_1}}}| + |{V_{{C_2}}}|}}.
\end{equation}

Since $|V_{C_1}|$ and $|V_{C_2}|$ are non-negative, we can derive
\begin{equation}
\small{
\min \left\{ {\frac{{\mu ({C_1},\Psi )}}{{|{V_{{C_1}}}|}},\frac{{\mu ({C_2},\Psi )}}{{|{V_{{C_2}}}|}}} \right\} \le {\rho _{opt}} \le \max \left\{ {\frac{{\mu ({C_1},\Psi )}}{{|{V_{{C_1}}}|}},\frac{{\mu ({C_2},\Psi )}}{{|{V_{{C_2}}}|}}} \right\}.}
\end{equation}

Notice that, $\rho ({C_1},\Psi ) = \frac{{\mu(C_1, \Psi)}}{{|{V_{{C_1}}}|}}$ and $\rho ({C_2},\Psi ) = \frac{{\mu(C_2, \Psi)}}{{|{V_{{C_2}}}|}}$. Therefore, if $\rho(C_1, \Psi)\neq \rho(C_2, \Psi)$, we can obtain a denser subgraph by removing the one with smaller density. This, however, contradicts the fact that $D$ is the CDS. Hence, $C_1$ and $C_2$ must be of the same density. The proof for the case where $D$ consists of more than two connected components is a straightforward extension, and is omitted here.
\end{proof}

%% file: proof-sec7.2.tex
\subsection{Proofs for Section~\ref{sec:pdsAlgorithm}}
\label{sec:exactCorrect}

\noindent\textsc{Lemma}~\ref{lemma:peelApp}.
\emph{Given a graph $G$ and a pattern $\Psi(V_\Psi, E_\Psi)$, the subgraph $S^*$ returned by {\tt PeelApp} is a $\frac{1}{|V_\Psi |}$-approximation solution to the PDS problem.}

\begin{proof}
Consider the PDS $D(V_D, E_D)$. After removing any vertex $v$ from $D$, the density will decrease, so we get
\begin{equation}
\small
\rho(D, \Psi) \geq \frac{{\mu(D, \Psi)-deg_D(v, \Psi)}}{{|V_D|-1}} \Leftrightarrow deg_D(v, \Psi) \geq \rho(D, \Psi).
\end{equation}

Let us consider the iteration before {\tt PeelApp} removes the first vertex $v$ which is in $D$.
We denote the subgraph in this iteration by $O(V_O, E_O)$ ($D\subseteq O$).
Since {\tt PeelApp} works in a greedy manner, we can conclude that,
for each vertex $u\in O$, $deg_O(u,\Psi)\geq deg_O(v,\Psi) \geq deg_D(v,\Psi) \geq \rho(D, \Psi)$.
As a result, we have $\mu(O, \Psi)=\frac{1}{|V_\Psi|} \sum\limits_{u \in O} {deg_G} (u, \Psi) \geq \frac{1}{|V_\Psi|}\cdot|V_O|\cdot\rho(D, \Psi)$.
Then, the density of $O$ is
\begin{equation}
\rho(O, \Psi)=\frac{\mu(O, \Psi)}{|V_O|}\geq \frac{1}{|V_\Psi|}\cdot\rho(D, \Psi)=\frac{\rho_{opt}}{|V_\Psi|}.
\end{equation}
Since $O$ is one of the $n$ residual subgraphs and $S^*$ is the densest one, we have $\rho(S^*, \Psi)\geq \rho(D, \Psi) \geq \frac{\rho_{opt}}{|V_\Psi|}$.
\end{proof}

\textsc{Theorem}
\ref{thm:exactCorrect}
\emph{Given a graph $G$ and a pattern $\Psi(V_\Psi, E_\Psi)$, the algorithm {\tt PExact} correctly finds the PDS of $G$ w.r.t. pattern-density of $\Psi$.}

We first prove the correctness of the criterion for deciding when to stop the binary search. Let ${\rho (G[{\Gamma _1}],\Psi )}$ and ${\rho (G[{\Gamma _2}],\Psi )}$ be the pattern-density of any two arbitrary subgraphs $G[{\Gamma _1}]$ and $G[{\Gamma _2}]$ of $G$, where $\Gamma _1$ and $\Gamma _1$ represent their vertex sets, respectively.
We show that if their density values are different, their difference cannot be arbitrary small.
\begin{lemma}
\label{lemma:binary}
\emph{For any two sets of vertices $\Gamma_1$ and $\Gamma_2\subseteq V$, if $\rho (G[{\Gamma _1}],\Psi ) \ne \rho (G[{\Gamma _2}],\Psi )$, then their difference is at least $\frac{1}{{n(n - 1)}}$.
}
\end{lemma}
\begin{proof}
For any two non-empty sets of vertices $\Gamma_1$, $\Gamma_2\subseteq V$, their pattern-density difference is
\begin{equation}
\small
\begin{array}{l}
\Delta  = \left| {\rho (G[{\Gamma _1}],\Psi ) - \rho (G[{\Gamma _2}],\Psi )} \right| = \left| {\frac{{\mu (G[{\Gamma _1}],\Psi )}}{{|{\Gamma _1}|}} - \frac{{\mu (G[{\Gamma _2}],\Psi )}}{{|{\Gamma _2}|}}} \right|\\
 = \left| {\frac{{\mu (G[{\Gamma _1}],\Psi )|{\Gamma _2}| - \mu (G[{\Gamma _2}],\Psi )|{\Gamma _1}|}}{{|{\Gamma _1}||{\Gamma _2}|}}} \right|
\end{array}
\end{equation}
, which is larger than 0, since $\rho (G[{\Gamma _1}],\Psi ) \ne \rho (G[{\Gamma _2}],\Psi )$.

If $\left| {{\Gamma _1}} \right|$=$\left| {{\Gamma _2}} \right|$, then
$\Delta  \ge \frac{1}{n}$; otherwise, since $0 \textless |{\Gamma _1}||{\Gamma _2}| \le n(n - 1)$, we have
$\Delta  \ge \frac{1}{{n(n - 1)}}$.
Hence, the lemma holds.
\end{proof}

The lemma above implies that, during the binary search process, if the distance between the lower and upper bounds is at most $\frac{1}{{n(n - 1)}}$, then we can stop the binary search.

Next, we discuss the capacity of the minimum st-cut. For ease of presentation, we introduce a notation, $\Psi [R, i]$, to indicate the number of pattern instances that involve exactly $i$ ($1\leq i\leq |V_\Psi|$) vertices from a given set $R$ of vertices.

\begin{lemma}
\label{lemma:mincut}
\emph{Given a flow network ${\mathcal F}(V_{\mathcal F}, E_{\mathcal F})$ constructed by {\tt construct+},
let (${\mathcal S}$, ${\mathcal T}$) denote its minimum st-cut,
${\mathcal A_1}$=${\mathcal S}\bigcap {\mathcal A}$,
${\mathcal B_1}$=${\mathcal S}\bigcap {\mathcal B}$,
${\mathcal A_2}$=${\mathcal T}\bigcap {\mathcal A}$, and
${\mathcal B_2}$=${\mathcal T}\bigcap {\mathcal B}$.
Then, the capacity of the minimum st-cut is
\begin{equation}
\label{eq:capacity}
\Phi  + \alpha \cdot|{V_\Psi }|\cdot|{\mathcal{A}_1}| + \Pi ,
\end{equation}
where $\Phi$=$\sum\limits_{v \notin {\mathcal A_1}} {de{g_G}} (v,\Psi )$, and $\Pi=\sum\limits_{i = 1}^{|{V_\Psi }| - 1} i \cdot\Psi [{\mathcal A_1},i]$.
}
\end{lemma}
\begin{proof}
We consider two different cases as follows.

\noindent\textbf{\underline{$i$) ${\mathcal A_1}$=$\emptyset$:}} In this case, ${\mathcal B_1}$ must be empty. This is because, if ${\mathcal B_1}\ne\emptyset$, the capacity of the st-cut must be larger than the situation when both ${\mathcal A_1}$ and ${\mathcal B_1}$ are empty sets, so we have ${\mathcal B_1}$=$\emptyset$.
In other words, ${\mathcal S}$=\{$s$\},
${\mathcal T}$=${\mathcal A}\cup {\mathcal B}\cup$ \{$t$\}.
It is easy to conclude that the lemma holds, as the capacity of the st-cut is
$\sum\limits_{v \in {{\mathcal A}_2}} {de{g_G}(v,\Psi )}$=$\sum\limits_{v \notin {{\mathcal A}_1}} {de{g_G}(v,\Psi )}$=$\Phi$.

\noindent\textbf{\underline{$ii$) ${\mathcal A_1}\ne\emptyset$:}} In this case, we need to consider four possible parts of the overall st-cut capacity:

(1) the capacity of the edges from the source node $s$ to nodes in $\mathcal A_2$, which is $\sum\limits_{v \in {{\mathcal A}_2}} {de{g_G}(v,\Psi )}$=$\sum\limits_{v \notin {{\mathcal A}_1}} {de{g_G}(v,\Psi )}$=$\Phi$.

(2) the capacity of the edges from nodes in $\mathcal A_1$ to the sink node $t$, which is $\alpha  \cdot |{V_\Psi }| \cdot |{{\mathcal A}_1}|$.

(3) the capacity of edges from nodes in $\mathcal B_1$ to nodes in $\mathcal A_2$. Consider a specific node (i.e., a pattern instance) in $\mathcal B_1$ and let $X$ ($|X|$=$|V_\Psi|$) be the set of its vertices. If $X\subseteq\mathcal A_1$, then this node and its vertices are in the same partition $\mathcal S$ and thus we skip it directly. If $|X\cap\mathcal A_1|$=$i$ ($1\leq i\leq |V_\Psi|-1$),then the capacity of all the patterns, whose vertex sets are $X$,
is $\mu (G[X],\Psi )(|{V_\Psi }| - 1)(|{V_\Psi }| - i)$, where $\mu (G[X],\Psi )$ equals to the size of this group.

(4) the capacity of edges from nodes in $\mathcal A_1$ to nodes in $\mathcal B_2$. Consider a specific node (i.e., a pattern instance) in $\mathcal B_2$ and let $X$ ($|X|$=$|V_\Psi|$) be the set of its vertices. Again, if $X\subseteq\mathcal A_2$, then this node and its vertices are in the same partition $\mathcal T$ and thus we skip it directly. If $|X\cap\mathcal A_2|$=$|V_\Psi|-i$ ($1\leq i\leq |V_\Psi|-1$),then $|X\cap\mathcal A_1|$=$i$ and the capacity is $i \cdot \mu (G[X],\Psi )$, where $\mu (G[X],\Psi )$ equals to the size of this group.

Notice that, for each node of $\mathcal B$, if its vertices come from both $\mathcal A_1$ and $\mathcal A_2$, it will be considered by both parts (3) and (4). However, since $1\leq i\leq |V_\Psi|-1$, we always have
$\mu (G[X],\Psi )(|{V_\Psi }| - 1)(|{V_\Psi }| - i) \ge i \cdot \mu (G[X],\Psi )$.
As a result, for any node with vertex set $X$ in part (3), if it shares at least one vertex with $\mathcal A_2$, we can move the node of $X$ into $\mathcal B_2$ to reduce the capacity, so we do not need to consider the capacity from part (3).

By summing up the capacity from parts (1), (2), and (3), we get the overall capacity as shown in Eq~(\ref{eq:capacity}).
Hence, the lemma holds for these two cases.
\end{proof}

We further present another lemma, which theoretically shows the correctness of the binary search in {\tt CoreExact}.

\begin{lemma}
\label{lemma:guess}
\emph{Consider a flow network ${\mathcal F}$ built on $G$,
\begin{enumerate}
  \item if there is a subgraph with vertex set $Y$ s.t. $\rho(G[Y], \Psi)\textgreater\alpha$, then any minimum st-cut ($\mathcal S$, $\mathcal T$) of ${\mathcal F}$ has $\mathcal S\backslash$\{$s$\}$\ne\emptyset$;
  \item if there does not exist a subgraph with vertex set $Y$ such that $\rho(G[Y], \Psi)\textgreater\alpha$, then the minimum st-cut satisfies that $\mathcal S$=\{$s$\} and $\mathcal T$=$\mathcal A\cup\mathcal B\cup$\{$t$\}.
\end{enumerate}}
\end{lemma}

\begin{proof}
We sequentially prove these two cases.

\noindent\textbf{\underline{First case:}}
We prove it by contradiction. Suppose that there is a subgraph with vertex set $Y$ such that $\rho(G[Y], \Psi)\textgreater\alpha$ and the minimum st-cut is achieved by $(\{s\},\mathcal A\cup \mathcal B\cup \{t\})$. In this case the capacity of the minimum st-cut is
\begin{equation}
\label{eq:firstTerm}
\sum_{v\in \mathcal A} deg_G(v, \Psi)=|V_\Psi|\mu(G,\Psi).
\end{equation}

We now consider another different st-cut ($\mathcal S$, $\mathcal T$), where $\mathcal S$=\{$s$\}$\cup\mathcal A_1\cup\mathcal B_1$, $\mathcal A_1$=$Y$, $\mathcal B_1$ is the set of groups of pattern instances formed by vertices in $\mathcal A_1$, and $\mathcal T$ contains the rest nodes in the network $\mathcal F$. Then, the capacity of this st-cut is exactly Eq~(\ref{eq:capacity}).

Since we assume that the minimum st-cut is achieved by $(\{s\},\mathcal A\cup \mathcal B\cup \{t\})$, the relationship between Eqs~(\ref{eq:capacity}) and (\ref{eq:firstTerm}) can be represented as
\begin{equation}
\label{eq:compare}
\footnotesize
\sum\limits_{v \in {\mathcal A}} d e{g_G}(v,\Psi ) \le \sum\limits_{v \notin {{\mathcal A}_1}} d e{g_G}(v,\Psi ) + \alpha  \cdot |{V_\Psi }| \cdot |{{\mathcal A}_1}| + \sum\limits_{i = 1}^{|{V_\Psi }| - 1} i  \cdot \Psi [{{\mathcal A}_1},i].
\end{equation}

The left part of Eq~(\ref{eq:compare}) can be rewritten as
\begin{equation}
\label{eq:formula1}
\sum\limits_{v \in {\mathcal A}} d e{g_G}(v,\Psi ) = \sum\limits_{v \in {{\mathcal A}_1}} d e{g_G}(v,\Psi ) + \sum\limits_{v \notin {{\mathcal A}_1}} d e{g_G}(v,\Psi ).
\end{equation}

Notice that the first term of the right part of Eq~(\ref{eq:formula1}) is
\begin{equation}
\label{eq:formula2}
\sum\limits_{v \in {{\mathcal A}_1}} d e{g_G}(v,\Psi ) = \sum\limits_{i = 1}^{|{V_\Psi }| - 1} {i\cdot\Psi [{{\mathcal A}_1},i]}  + |{V_\Psi }|\cdot\Psi [{{\mathcal A}_1},|{V_\Psi }|].
\end{equation}

By considering the two equations above, we can simplify the inequation of Eq~(\ref{eq:compare}) as
\begin{equation}
\label{eq:capacityCompare}
|{V_\Psi }| \cdot \Psi [{{\mathcal A}_1},|{V_\Psi }|] \le \alpha  \cdot |{V_\Psi }| \cdot |{{\mathcal A}_1}|.
\end{equation}

Since $|V_\Psi|\geq 1$, $|{V_\Psi }| \cdot \Psi [{{\mathcal A}_1}, |{V_\Psi }|]$=$\mu(G[Y], \Psi)$, and $|\mathcal A_1|$=$|Y|$, we can conclude that $\rho (G[Y],\Psi ) = \frac{{\mu (G[Y],\Psi )}}{{|Y|}} \le \alpha$,
which contradicts the assumption. Hence the first case holds.

\noindent\textbf{\underline{Second case:}}
We also prove by contradiction. Suppose that there does not exist a subgraph with vertex set $Y$ such that $\rho(G[Y], \Psi)\textgreater\alpha$, and the st-cut (\{$s$\}, $\mathcal A\cup\mathcal B\cup$ \{$t$\}) is not the minimum st-cut. Notice that the capacity of this st-cut is $\sum\limits_{v \in {\mathcal A}} d e{g_G}(v,\Psi )$.

By Lemma~\ref{lemma:mincut}, for any minimum st-cut (${\mathcal S}$, $\mathcal T$), its capacity is
$\sum\limits_{v \notin {{\mathcal A}_1}} {deg_G} (v,\Psi ) + \alpha \cdot|{V_\Psi }|\cdot|{{\mathcal A}_1}| + \sum\limits_{i = 1}^{|{V_\Psi }| - 1} i \cdot\Psi [{{\mathcal A}_1},i]$.

As a result, we have
\begin{equation}
\label{eq:compare2}
\footnotesize
\sum\limits_{v \in {\mathcal A}} d e{g_G}(v,\Psi ) \ge \sum\limits_{v \notin {{\mathcal A}_1}} d e{g_G}(v,\Psi ) + \alpha  \cdot |{V_\Psi }| \cdot |{{\mathcal A}_1}| + \sum\limits_{i = 1}^{|{V_\Psi }| - 1} {i \cdot \Psi [{{\mathcal A}_1},i]}.
\end{equation}

We notice that Eq~(\ref{eq:compare2}) is very similar to Eq~(\ref{eq:compare}), so we can further analyze it using the similar analysis above and finally derive that $\mu (G[Y],\Psi ) > \alpha |Y|$. This contradicts the assumption. Hence, the second case holds.
\end{proof}

Therefore, Theorem~\ref{thm:exactCorrect} is proved.

\noindent\textsc{Lemma}~\ref{lemma:compressF}.
\emph{Given a graph $G$, a pattern $\Psi(V_\Psi, E_\Psi)$, the flow networks built by {\tt PExact} (lines 5-12) and {\tt construct+} have the same capacity for their minimum st-cut.}
\begin{proof}
To prove the lemma, let us revisit the proof of Lemma~\ref{lemma:mincut}, which considers two cases.
Let the flow networks built by {\tt PExact} (lines 5-12) and {\tt construct+} be ${\mathcal F}$ and ${\mathcal F}+$ respectively.
We show that, for each case, the minimum st-cuts of ${\mathcal F}$ and ${\mathcal F}+$ are the same.
The case ``${\mathcal A_1}$=$\emptyset$" obviously holds as all the pattern groups are in the same partition $\mathcal T$.

For the other case ``${\mathcal A_1}\ne\emptyset$", the capacity consists of four parts.
The first two parts holds as they do not consider the patterns.
For part (3), in the proof of Lemma~\ref{lemma:mincut}, when considering a specific pattern instance with vertex set $X$, we have considered all the patterns which share the set of vertices. Let $g$ be the group of pattern instances with vertex set $X$. Then, we can observe that $\mu(G[X], \Psi)$=$|g|$. Thus, after grouping these pattern instances and increasing the capacity of the edges by $|g|$ times, the capacity of part (3) does not change.
Similarly, we can prove that the capacity in part (4) remains unchanged.
Hence, the lemma holds.
\end{proof}

%% file: pexact.tex
\section{Pseudocodes of PExact}
\label{app:pexact}

Algorithm \ref{alg:basicPExact} presents the detailed pseudocodes of {\tt PExact}.

\begin{algorithm}{}
\small
\caption{The algorithm: {\tt PExact}.}
\label{alg:basicPExact}
\KwIn{$G(V,E)$, $\Psi(V_\Psi,E_\Psi)$;}
\KwOut{The PDS $D(V_D, E_D)$;}
initialize $l \gets 0$, $u\gets\mathop {\max }\limits_{v \in V} deg_G(v, \Psi)$\;
initialize $\Lambda\gets$all the pattern instances of $\Psi$ in $G$, $D\gets\emptyset$\;
\While{$ u-l\geq \frac{1}{n(n-1)} $}{
    $\alpha\gets \frac{l+u}{2}$\;
    $V_{\mathcal F}\gets \{s\}\cup V\cup\Lambda\cup\{t\}$\tcp*{build a flow network}
    \For{$each$ $vertex$ $v\in V$}{
        add an edge $s$$\rightarrow$$v$ with capacity $deg_G(v, \Psi)$\;
        add an edge $v$$\rightarrow$$t$ with capacity $\alpha|V_\Psi|$\;

    }
    \For{$each$ $pattern$ $\psi\in\Lambda$}{
        \For{$each$ $vertex$ $v\in\psi$}{
            add an edge $v$$\rightarrow$$\psi$ with capacity 1\;
            add an edge $\psi$$\rightarrow$$v$ with capacity ($|V_\Psi|$--1)\;
        }
    }
    find minimum st-cut ($\mathcal S$, $\mathcal T$) from the flow network $\mathcal F(V_{\mathcal F}, E_{\mathcal F})$\;
    \textbf{if} $\mathcal S$=$\{s\}$ \textbf{then} $u \gets \alpha$\;
    \textbf{else} \text{ } \text{ } \text{ } \text{ } $l \gets \alpha$,  $D\gets$ the subgraph induced by $\mathcal S\backslash \{s\}$\;
}
\Return $D$\;
\end{algorithm}

%% file: kcoreSpec.tex
\section{Optimizing Core Decomposition for Special Patterns}
\label{app:kcoreSpec}

The core decomposition algorithm in Section~\ref{sec:decompose} can be extended for a general pattern $\Psi$. For particular patterns, such as stars, and loops, the decomposition process above can be performed faster, because the two key steps -- \emph{computing the pattern-degrees} and \emph{decreasing the vertices' pattern-degrees}, can be done more efficiently.
In the following, we illustrate this for two kinds of commonly encountered patterns, namely \emph{star}, and \emph{loop}.

\begin{figure}[h]
\hspace*{-.1cm}
\centering
\begin{tabular}{c c c}
  \begin{minipage}{2.2cm}
	\includegraphics[width=2.2cm]{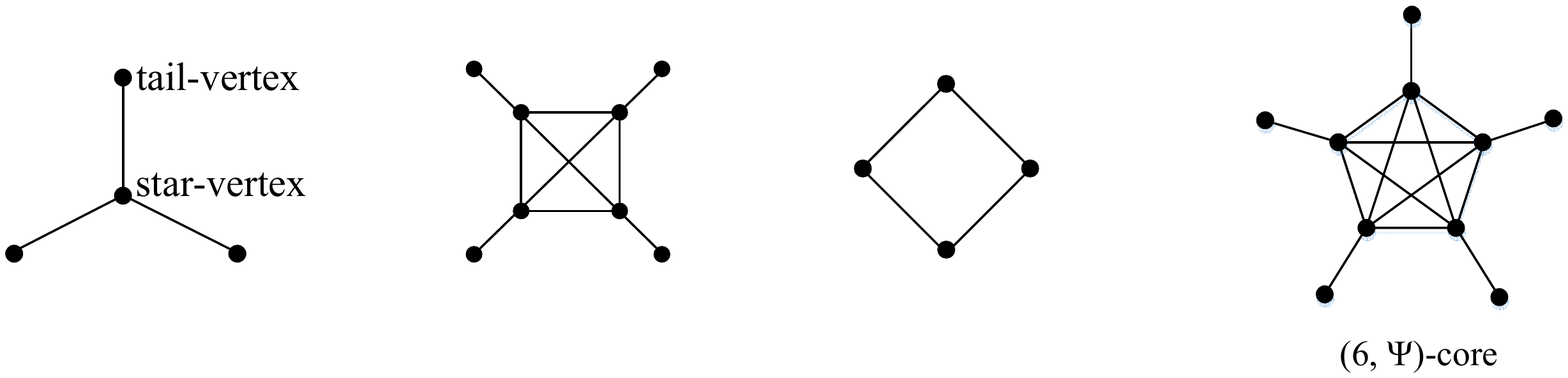}
  \end{minipage}
  &
  \begin{minipage}{1.40cm}
	\includegraphics[width=1.40cm]{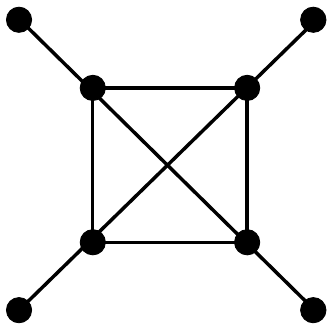}
  \end{minipage}
  &
  \begin{minipage}{1.40cm}
	\includegraphics[width=1.40cm]{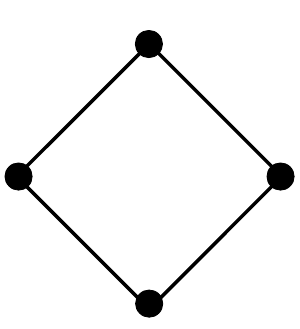}
  \end{minipage}
  \\
  (a) 3-star
  &
  (b) (3, $\Psi$)-core ($\Psi$=3-star)
  &
  (c) diamond
\end{tabular}
\caption{Illustrating the special patterns.}
\label{fig:starGraph}
\end{figure}

\noindent \textbf{1 Star Patterns.}
The star pattern has been widely investigated~\cite{nature2003,leskovec2006}. Specifically, a star pattern consists of a star-vertex (center) with a few tail-vertices (spokes). We call a star pattern with $x$ ($x\textgreater$1) tail-vertices an $x$-star pattern. In Figure~\ref{fig:starGraph}(a), we depict the 3-star pattern $\Psi$. The whole graph in Figure~\ref{fig:starGraph}(b) is an example (3, $\Psi$)-core: each vertex participates in at least three 3-star pattern instances.

Given an $x$-star pattern, we now discuss how to efficiently perform these two key steps for a vertex $v$ during the decomposition process. Assume that $v$ has $y$ neighbors, and let $v$'s $i$-th (1$\leq i\leq y$) neighbor be $u$.

\textbf{\underline{1.1 Computing $v$'s pattern-degree.}}
There are two disjoint cases: (1) $v$ is the star-vertex and (2) $v$ is a tail-vertex.
For case (1), the number of pattern instances is ${y \choose x}$,
since any $x$ neighbors of $v$ can form a $x$-star with $v$.
For case (2), the number of pattern instances is
$\sum\limits_{1 \le i \le y \wedge x \le z_i} {{z_i-1} \choose x-1}$,
where $z_i$ is $u$'s degree, since $v$ and any other $x$--1 neighbors of $u$ can form a $x$-star with $u$.
Therefore, we have
\begin{equation}
\small
deg_G(v, \Psi)={y \choose x} + \sum\limits_{1 \le i \le y \wedge x \le z_i} {{{z_i-1} \choose x-1}}.
\end{equation}
Here, we adopt the convention that ${y \choose x} =0$, whenever $y\textless x$.

\textbf{\underline{1.2 Decreasing pattern-degrees.}}
After removing $v$, we also have two disjoint cases:
(1) decreasing the pattern-degrees of $v$'s 1-hop neighbors,
and (2) decreasing the pattern-degrees of $v$'s 2-hop neighbors.
Let us denote $v$'s $i$-th (1$\leq i\leq y$) neighbor by $u$.
For case (1), we have to compute the number of pattern instances that include both $v$ and $u$.
When $v$ is the star-vertex, the number of $x$-stars that have a tail-vertex $u$ is ${{y-1} \choose {x-1}}$,
since $u$ and any other $x$--1 neighbors of $v$ can form a $x$-star with $v$;
when $v$ is the tail-vertex, the number of $x$-stars that have a star-vertex $u$ is ${{z_i-1} \choose {x-1}}$,
since $v$ and any other $x$--1 neighbors of $u$ can form a $x$-star with $u$.
Therefore, we decrease $u$'s pattern-degree by
${{y-1} \choose {x-1}}+{{z_i-1} \choose {x-1}}$.

In case (2), $u$ must be a star-vertex and $v$ is a tail-vertex. For each neighbor $w$ of $u$ ($w\neq v$), we need to decrease $w$'s pattern-degree by ${{z_i-2} \choose {x-2}}$, because $v$, $w$, and any other $x$--2 neighbors of $u$ can form an $x$-star with $u$.

From the discussions above, we can easily conclude that, the first and second steps can be completed in ${\mathcal O}(d)$ and ${\mathcal O}(d^2)$ time respectively, because we only need to enumerate their 1-hop and 2-hop neighbors. Note that the values of different ${y \choose x}$ ($x$$\leq$$y$$\leq$$d$) can be precomputed in advance. Therefore, for any $x$-star pattern, the time cost of performing $k$-pattern-core decomposition can be reduced from ${\mathcal O}\left( {n \cdot d^x} \right)$ to ${\mathcal O}(n\cdot{d^2})$.

\noindent \textbf{2 Loop Patterns.}
For ease of exposition, we take the diamond pattern (see Figure~\ref{fig:starGraph}(c)) as an example to illustrate the details of these two steps. The diamond pattern has also been used in many real applications~\cite{nature2003,leskovec2006}.

\textbf{\underline{2.1 Computing $v$'s pattern-degree.}}
We first enumerate all the paths from $v$ to its 2-hop neighbors, and then organize these paths into $h$ groups, each of which share the same 2-hop neighbor. Then, for each group, any pair of paths can form an instance of the diamond pattern.
Let the size of the $i$-th group be $y_i$ ($1\leq i\leq h$). We thus have $de{g_G}(v,\Psi) = \sum\limits_{1 \le i \le h \wedge {y_i} \ge 2} {y_i \choose 2}$.

\textbf{\underline{2.2 Decreasing pattern-degrees.}}
We handle the groups one by one. Consider the specific group $h_i$, which only has one 2-hop neighbor of $v$.
If $y_i$=1, we skip this group as there is not any pattern instance.
If $y_i\geq 2$, we decrease the pattern-degree of the 2-hop neighbor by ${y_i \choose 2}$.
For each 1-hop neighbor in $g_i$, we decrease its pattern-degree by $y_i$--1, because each path participates in $y_i$--1 pattern instances.

Clearly, the time complexity of performing the two steps above is ${\mathcal O}(d^2)$. Hence,
the time cost of performing core decomposition is reduced from ${\mathcal O}\left( {n \cdot d^3} \right)$ to ${\mathcal O}(n\cdot{d^2})$.

%% file: caseStudy.tex
\begin{figure*}[ht]
\centering
\begin{tabular}{c c c c c}
  \begin{minipage}{3.9cm}
	\includegraphics[width=3.9cm]{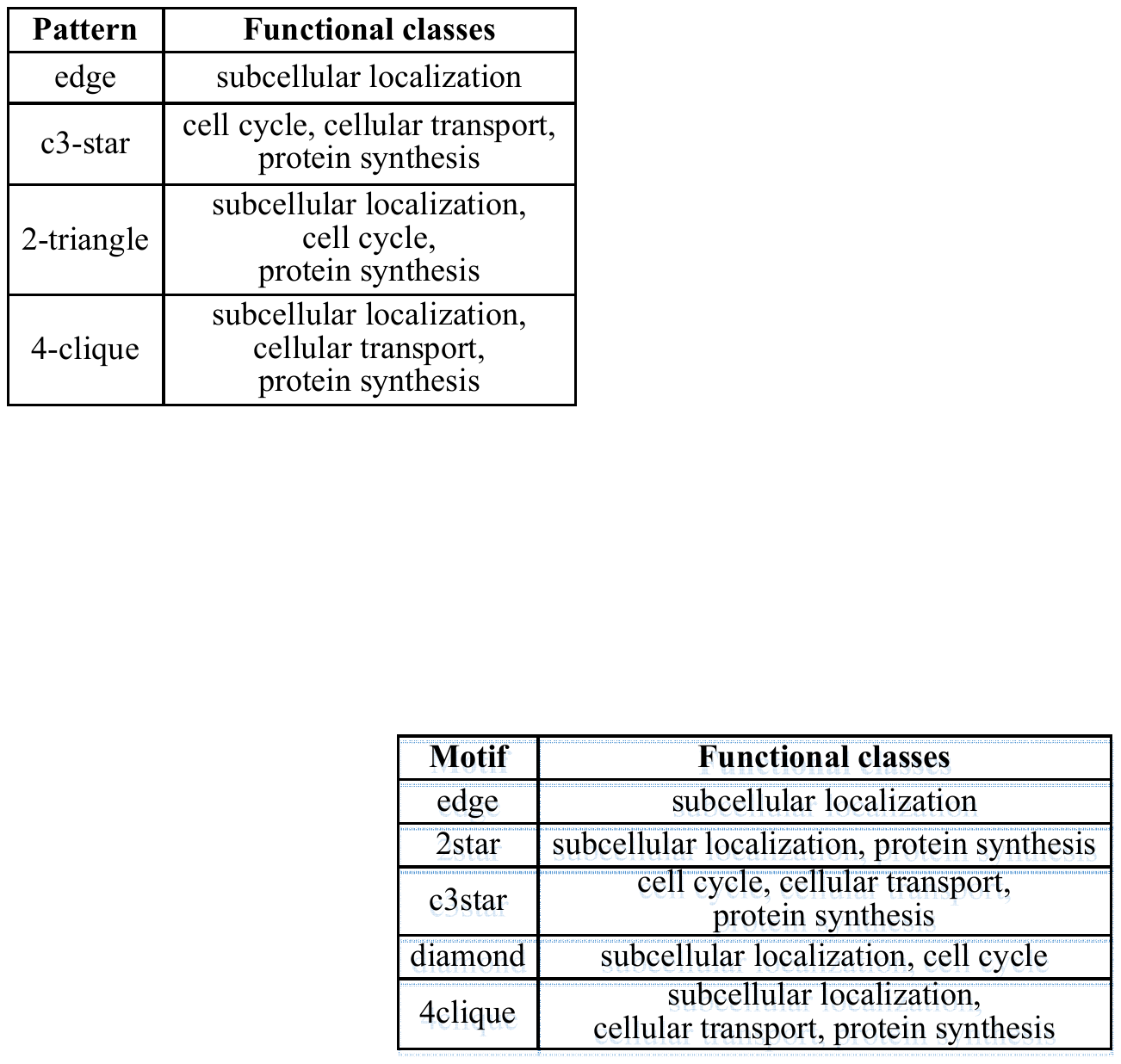}
  \end{minipage}
  &
  \begin{minipage}{3.0cm}
	\includegraphics[width=3.0cm]{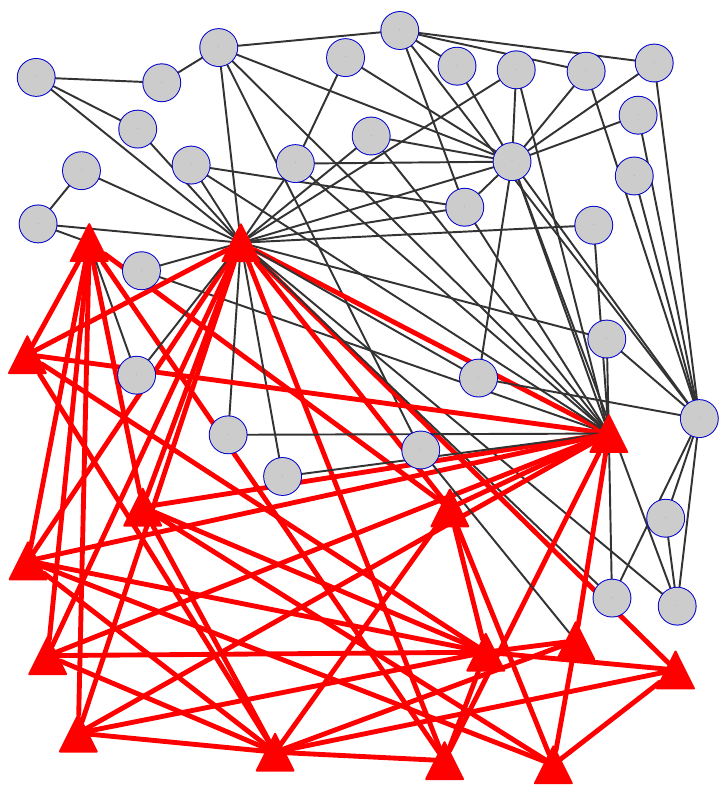}
  \end{minipage}
  &
  \begin{minipage}{3.0cm}
	\includegraphics[width=3.0cm]{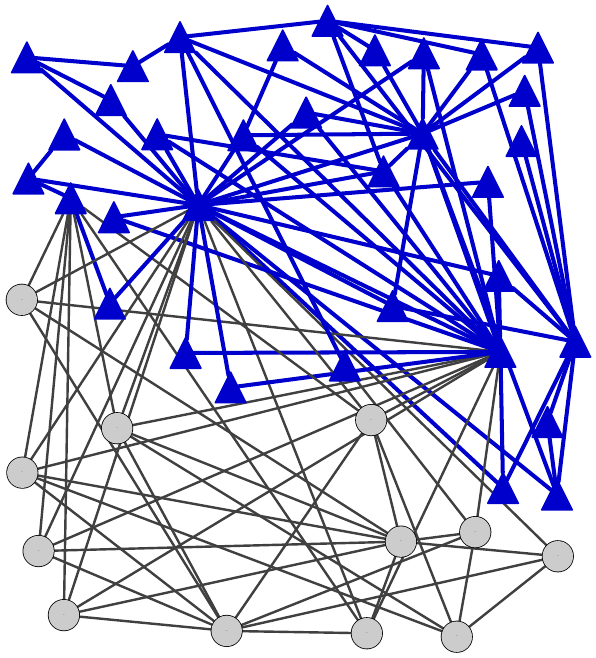}
  \end{minipage}
  &
  \begin{minipage}{3.0cm}
	\includegraphics[width=3.0cm]{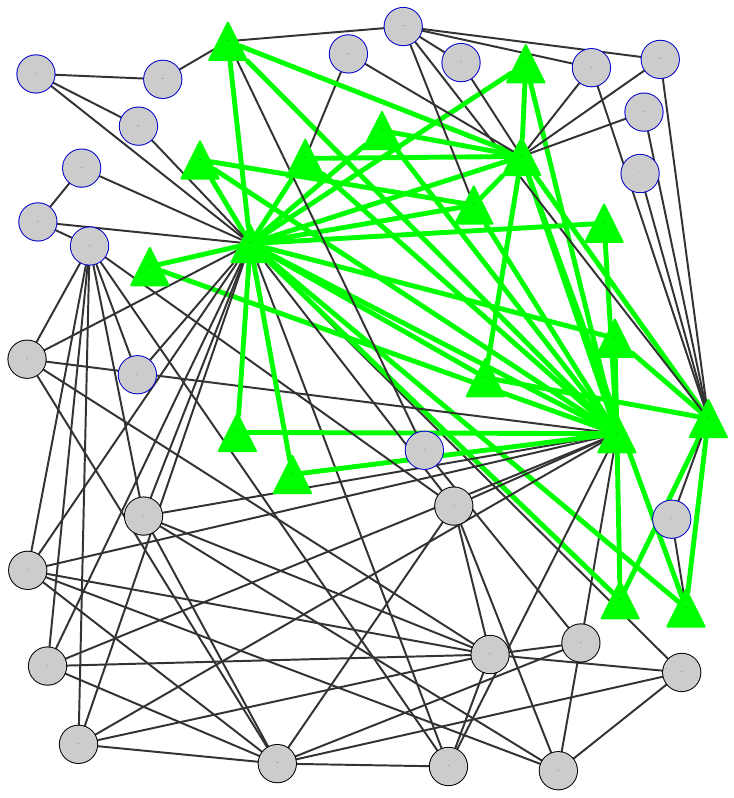}
  \end{minipage}
  &
  \begin{minipage}{3.0cm}
	\includegraphics[width=3.0cm]{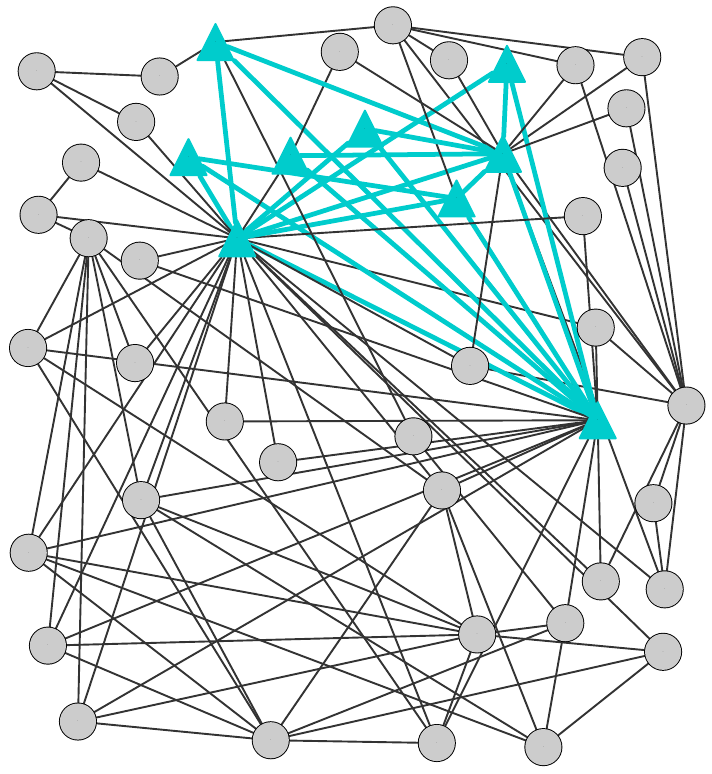}
  \end{minipage}
  \\

  (a) Functional classes~\cite{nature2003}
  &
  (b) edge (red)
  &
  (c) c3-star (blue)
  &
  (d) 2-triangle (green)
  &
  (e) 4-clique (cyan)
\end{tabular}
\small\caption{PDS's (vertices in colored triangles) in the yeast PPI network.}
\label{fig:casestudy}
\end{figure*}

\section{More case studies}
\label{app:caseStudy}

\textbf{Yeast.} As \cite{nature2003} pointed out, the conservation of biological cellular functions across species during the evolution process is often carried out by integrated activities of various patterns. A protein can belong to more than one functional class, and a pattern corresponds to multiple classes. Figure~\ref{fig:casestudy}(a) shows four patterns and their corresponding functional classes for yeast proteins~\cite{nature2003}.
We studied a yeast PPI network~\footnote{DIP: \url{http://dip.doe-mbi.ucla.edu/dip/Stat.cgi}} ($|V|$=1,116, $|E|$=2,148), where  vertices represent proteins and edges represent interactions.
For each pattern, we compute the PDS, as highlighted in Figure~\ref{fig:casestudy}, which also shows the subnetwork induced by the vertices in the union of all the PDS's, for convenience.
Notice that the PDS's corresponding to different patterns have distinct shapes. Each PDS %so formed
could represent a subnetwork with a specific function.  Tasks such as analyzing the conservation and evolution of cellular components~\cite{nature2003} can then be performed on these subnetworks.

%% file: main.bbl
\begin{thebibliography}{10}

\bibitem{clique}
https://en.wikipedia.org/wiki/Clique\_(graph\_theory).

\bibitem{ahuja1994improved}
R.~K. Ahuja, J.~B. Orlin, C.~Stein, and R.~E. Tarjan.
\newblock Improved algorithms for bipartite network flow.
\newblock {\em SIAM Journal on Computing}, 23(5):906--933, 1994.

\bibitem{andersen2009finding}
R.~Andersen and K.~Chellapilla.
\newblock Finding dense subgraphs with size bounds.
\newblock In {\em International Workshop on Algorithms and Models for the
  Web-Graph}, pages 25--37, 2009.

\bibitem{asahiro2002complexity}
Y.~Asahiro, R.~Hassin, and K.~Iwama.
\newblock Complexity of finding dense subgraphs.
\newblock {\em Discrete Applied Mathematics}, 121(1):15--26, 2002.

\bibitem{asahiro2000greedily}
Y.~Asahiro, K.~Iwama, H.~Tamaki, and T.~Tokuyama.
\newblock Greedily finding a dense subgraph.
\newblock {\em Journal of Algorithms}, 34(2):203--221, 2000.

\bibitem{bahmani2012densest}
B.~Bahmani, R.~Kumar, and S.~Vassilvitskii.
\newblock Densest subgraph in streaming and mapreduce.
\newblock {\em PVLDB}, 5(5):454--465, 2012.

\bibitem{kcore2003}
V.~Batagelj and M.~Zaversnik.
\newblock An o(m) algorithm for cores decomposition of networks.
\newblock {\em arXiv preprint cs/0310049}, 2003.

\bibitem{bhaskara2010detecting}
A.~Bhaskara, M.~Charikar, E.~Chlamtac, U.~Feige, and A.~Vijayaraghavan.
\newblock Detecting high log-densities: an o (n $1/4$) approximation for
  densest k-subgraph.
\newblock In {\em STOC}, pages 201--210, 2010.

\bibitem{thesis2013}
E.~T. Charalampos.
\newblock {\em Mathematical and Algorithmic Analysis of Network and Biological
  Data}.
\newblock PhD thesis, Carnegie Mellon University, 2013. 2.1.

\bibitem{charikar2000greedy}
M.~Charikar.
\newblock Greedy approximation algorithms for finding dense components in a
  graph.
\newblock In {\em APPROX}, pages 84--95. Springer, 2000.

\bibitem{chen2012dense}
J.~Chen and Y.~Saad.
\newblock Dense subgraph extraction with application to community detection.
\newblock {\em TKDE}, 24(7):1216--1230, 2012.

\bibitem{chen2018exploring}
Y.~Chen, Y.~Fang, R.~Cheng, Y.~Li, X.~Chen, and J.~Zhang.
\newblock Exploring communities in large profiled graphs.
\newblock {\em IEEE Transactions on Knowledge and Data Engineering},
  31(8):1624--1629, 2019.

\bibitem{cheng2011efficient}
J.~Cheng, Y.~Ke, S.~Chu, and M.~T. {\"O}zsu.
\newblock Efficient core decomposition in massive networks.
\newblock In {\em ICDE}, pages 51--62, 2011.

\bibitem{cohen2003reachability}
E.~Cohen, E.~Halperin, H.~Kaplan, and U.~Zwick.
\newblock Reachability and distance queries via 2-hop labels.
\newblock {\em SIAM Journal on Computing}, 32(5):1338--1355, 2003.

\bibitem{cohen2008trusses}
J.~Cohen.
\newblock Trusses: Cohesive subgraphs for social network analysis.
\newblock {\em National Security Agency Technical Report}, 16, 2008.

\bibitem{cui2013online}
W.~Cui, Y.~Xiao, H.~Wang, Y.~Lu, and W.~Wang.
\newblock Online search of overlapping communities.
\newblock In {\em SIGMOD}, pages 277--288. ACM, 2013.

\bibitem{Mauro1988}
M.~Danisch, O.~Balalau, and M.~Sozio.
\newblock Listing k-cliques in sparse real-world graphs.
\newblock In {\em WWW}, pages 589--598, 2018.

\bibitem{danisch2017large}
M.~Danisch, T.-H.~H. Chan, and M.~Sozio.
\newblock Large scale density-friendly graph decomposition via convex
  programming.
\newblock In {\em Proceedings of the 26th International Conference on World
  Wide Web}, pages 233--242, 2017.

\bibitem{evolving2015}
A.~Epasto, S.~Lattanzi, and M.~Sozio.
\newblock Efficient densest subgraph computation in evolving graphs.
\newblock In {\em WWW}, pages 300--310, 2015.

\bibitem{fang2017VLDBJ}
Y.~Fang, R.~Cheng, Y.~Chen, S.~Luo, and J.~Hu.
\newblock Effective and efficient attributed community search.
\newblock {\em The VLDB Journal}, 26(6):803--828, 2017.

\bibitem{fang2018spatialICDE}
Y.~Fang, R.~Cheng, G.~Cong, N.~Mamoulis, and Y.~Li.
\newblock On spatial pattern matching.
\newblock In {\em IEEE International Conference on Data Engineering (ICDE)},
  pages 293--304. IEEE, 2018.

\bibitem{fang2017effective}
Y.~Fang, R.~Cheng, X.~Li, S.~Luo, and J.~Hu.
\newblock Effective community search over large spatial graphs.
\newblock {\em PVLDB}, 10(6):709--720, 2017.

\bibitem{Fang2016}
Y.~Fang, R.~Cheng, S.~Luo, and J.~Hu.
\newblock Effective community search for large attributed graphs.
\newblock {\em PVLDB}, 9(12):1233--1244, 2016.

\bibitem{fang2019survey}
Y.~Fang, X.~Huang, L.~Qin, Y.~Zhang, W.~Zhang, R.~Cheng, and X.~Lin.
\newblock A survey of community search over big graphs.
\newblock {\em The VLDB Journal}, 2019.

\bibitem{fang2018spatial}
Y.~Fang, Z.~Wang, R.~Cheng, X.~Li, S.~Luo, J.~Hu, and X.~Chen.
\newblock On spatial-aware community search.
\newblock {\em IEEE Transactions on Knowledge and Data Engineering},
  31(4):783--798, 2018.

\bibitem{fang2018effective}
Y.~Fang, Z.~Wang, R.~Cheng, H.~Wang, and J.~Hu.
\newblock Effective and efficient community search over large directed graphs.
\newblock {\em IEEE Transactions on Knowledge and Data Engineering (TKDE)},
  2018.

\bibitem{fullVersion}
Y.~Fang, K.~Yu, R.~Cheng, L.~V.~S. Lakshmanan, and X.~Lin.
\newblock Efficient algorithms for densest subgraph discovery (technical
  report). {arXiv, 2019}.
\newblock \url{https://arxiv.org/abs/1906.00341}.

\bibitem{fratkin2006motifcut}
E.~Fratkin, B.~T. Naughton, D.~L. Brutlag, and S.~Batzoglou.
\newblock Motifcut: regulatory motifs finding with maximum density subgraphs.
\newblock {\em Bioinformatics}, 22(14):e150--e157, 2006.

\bibitem{gallo1989fast}
G.~Gallo, M.~D. Grigoriadis, and R.~E. Tarjan.
\newblock A fast parametric maximum flow algorithm and applications.
\newblock {\em SIAM Journal on Computing}, 18(1):30--55, 1989.

\bibitem{gionis2013piggybacking}
A.~Gionis, F.~Junqueira, V.~Leroy, M.~Serafini, and I.~Weber.
\newblock Piggybacking on social networks.
\newblock {\em PVLDB}, 6(6):409--420, 2013.

\bibitem{tutorial}
A.~Gionis and C.~E. Tsourakakis.
\newblock Dense subgraph discovery: Kdd 2015 tutorial.
\newblock In {\em SIGKDD}, pages 2313--2314, NY, USA, 2015. ACM.

\bibitem{goldberg1984finding}
A.~V. Goldberg.
\newblock {\em Finding a maximum density subgraph}.
\newblock UC Berkeley, 1984.

\bibitem{flow}
G.~T. Heineman, G.~Pollice, and S.~Selkow.
\newblock Chapter 8: Network flow algorithms.
\newblock {\em Algorithms in a Nutshell}, pages 226--250, 2008.

\bibitem{hu2019discovering}
J.~Hu, R.~Cheng, K.~C.-C. Chang, A.~Sankar, Y.~Fang, and B.~Y. Lam.
\newblock Discovering maximal motif cliques in large heterogeneous information
  networks.
\newblock In {\em IEEE International Conference on Data Engineering (ICDE)},
  pages 746--757. IEEE, 2019.

\bibitem{hu2016querying}
J.~Hu, X.~Wu, R.~Cheng, S.~Luo, and Y.~Fang.
\newblock Querying minimal steiner maximum-connected subgraphs in large graphs.
\newblock In {\em International on Conference on Information and Knowledge
  Management (CIKM)}, pages 1241--1250. ACM, 2016.

\bibitem{hu2017minimal}
J.~Hu, X.~Wu, R.~Cheng, S.~Luo, and Y.~Fang.
\newblock On minimal steiner maximum-connected subgraph queries.
\newblock {\em IEEE Transactions on Knowledge and Data Engineering},
  29(11):2455--2469, 2017.

\bibitem{Huang2014}
X.~Huang, H.~Cheng, L.~Qin, W.~Tian, and J.~X. Yu.
\newblock Querying k-truss community in large and dynamic graphs.
\newblock In {\em SIGMOD}, pages 1311--1322. ACM, 2014.

\bibitem{huang2017attribute}
X.~Huang and L.~V. Lakshmanan.
\newblock Attribute-driven community search.
\newblock {\em PVLDB}, 10(9):949--960, 2017.

\bibitem{huang2016truss}
X.~Huang, W.~Lu, and L.~V. Lakshmanan.
\newblock Truss decomposition of probabilistic graphs: Semantics and
  algorithms.
\newblock In {\em Proceedings of the 2016 International Conference on
  Management of Data}, pages 77--90. ACM, 2016.

\bibitem{jha2015path}
M.~Jha, C.~Seshadhri, and A.~Pinar.
\newblock Path sampling: A fast and provable method for estimating 4-vertex
  subgraph counts.
\newblock In {\em WWW}, pages 495--505, 2015.

\bibitem{jin20093}
R.~Jin, Y.~Xiang, N.~Ruan, and D.~Fuhry.
\newblock 3-hop: a high-compression indexing scheme for reachability query.
\newblock In {\em SIGMOD}, pages 813--826. ACM, 2009.

\bibitem{johnson1987parallel}
D.~B. Johnson.
\newblock Parallel algorithms for minimum cuts and maximum flows in planar
  networks.
\newblock {\em Journal of the ACM (JACM)}, 34(4):950--967, 1987.

\bibitem{directed1999}
R.~Kannan and V.~Vinay.
\newblock {\em Analyzing the structure of large graphs}.
\newblock Rheinische Friedrich-Wilhelms-Universit{\"a}t Bonn, 1999.

\bibitem{khuller2009finding}
S.~Khuller and B.~Saha.
\newblock On finding dense subgraphs.
\newblock {\em Automata, Languages and Programming}, pages 597--608, 2009.

\bibitem{lai2016scalable}
L.~Lai, L.~Qin, X.~Lin, Y.~Zhang, L.~Chang, and S.~Yang.
\newblock Scalable distributed subgraph enumeration.
\newblock {\em PVLDB}, 10(3):217--228, 2016.

\bibitem{leskovec2006}
J.~Leskovec, A.~Singh, and J.~Kleinberg.
\newblock Patterns of influence in a recommendation network.
\newblock In {\em PAKDD}, pages 380--389. Springer, 2006.

\bibitem{lu2016h}
L.~L{\"u}, T.~Zhou, Q.-M. Zhang, and H.~E. Stanley.
\newblock The h-index of a network node and its relation to degree and
  coreness.
\newblock {\em Nature communications}, 7:10168, 2016.

\bibitem{mandal2017distributed}
A.~Mandal and M.~Al~Hasan.
\newblock A distributed k-core decomposition algorithm on spark.
\newblock In {\em International Conference on Big Data}, pages 976--981. IEEE,
  2017.

\bibitem{mitzenmacher2015scalable}
M.~Mitzenmacher, J.~Pachocki, R.~Peng, C.~Tsourakakis, and S.~C. Xu.
\newblock Scalable large near-clique detection in large-scale networks via
  sampling.
\newblock In {\em International Conference on Knowledge Discovery and Data
  Mining (SIGKDD)}, pages 815--824. ACM, 2015.

\bibitem{montresor2013distributed}
A.~Montresor, F.~De~Pellegrini, and D.~Miorandi.
\newblock Distributed k-core decomposition.
\newblock {\em IEEE Transactions on parallel and distributed systems},
  24(2):288--300, 2013.

\bibitem{peng2018efficient}
Y.~Peng, Y.~Zhang, W.~Zhang, X.~Lin, and L.~Qin.
\newblock Efficient probabilistic k-core computation on uncertain graphs.
\newblock In {\em International Conference on Data Engineering (ICDE)}, pages
  1192--1203. IEEE, 2018.

\bibitem{pham2005distributed}
T.~L. Pham, I.~Lavallee, M.~Bui, and S.~H. Do.
\newblock A distributed algorithm for the maximum flow problem.
\newblock In {\em International Symposium on Parallel and Distributed Computing
  (ISPDC)}, pages 131--138. IEEE, 2005.

\bibitem{qiao2017subgraph}
M.~Qiao, H.~Zhang, and H.~Cheng.
\newblock Subgraph matching: on compression and computation.
\newblock {\em PVLDB}, 11(2):176--188, 2017.

\bibitem{qin2015locally}
L.~Qin, R.-H. Li, L.~Chang, and C.~Zhang.
\newblock Locally densest subgraph discovery.
\newblock In {\em KDD}, pages 965--974. ACM, 2015.

\bibitem{saha2010dense}
B.~Saha, A.~Hoch, S.~Khuller, L.~Raschid, and X.-N. Zhang.
\newblock Dense subgraphs with restrictions and applications to gene annotation
  graphs.
\newblock In {\em RECOMB}, volume 6044, pages 456--472, 2010.

\bibitem{Sahu:2017}
S.~Sahu, A.~Mhedhbi, S.~Salihoglu, J.~Lin, and M.~T. \"{O}zsu.
\newblock The ubiquity of large graphs and surprising challenges of graph
  processing.
\newblock {\em PVLDB}, 11(4):420--431, 2017.

\bibitem{samusevich2016local}
R.~Samusevich, M.~Danisch, and M.~Sozio.
\newblock Local triangle-densest subgraphs.
\newblock In {\em International Conference on Advances in Social Networks
  Analysis and Mining (ASONAM)}, pages 33--40. IEEE, 2016.

\bibitem{sariyuce2016fast}
A.~E. Sariy{\"u}ce and A.~Pinar.
\newblock Fast hierarchy construction for dense subgraphs.
\newblock {\em PVLDB}, 10(3):97--108, 2016.

\bibitem{sariyuce2018local}
A.~E. Sariy{\"u}ce, C.~Seshadhri, and A.~Pinar.
\newblock Local algorithms for hierarchical dense subgraph discovery.
\newblock {\em PVLDB}, 12(1):43--56, 2018.

\bibitem{sariyuce2015finding}
A.~E. Sariy{\"u}ce, C.~Seshadhri, A.~Pinar, and U.~V. Catalyurek.
\newblock Finding the hierarchy of dense subgraphs using nucleus
  decompositions.
\newblock In {\em WWW}, pages 927--937, 2015.

\bibitem{sariyuce2017nucleus}
A.~E. Sariy{\"u}ce, C.~Seshadhri, A.~Pinar, and {\"U}.~V. {\c{C}}ataly{\"u}rek.
\newblock Nucleus decompositions for identifying hierarchy of dense subgraphs.
\newblock {\em ACM Transactions on the Web (TWEB)}, 11(3):16, 2017.

\bibitem{md1983}
S.~B. Seidman.
\newblock Network structure and minimum degree.
\newblock {\em Social networks}, 1983.

\bibitem{seidman1978graph}
S.~B. Seidman and B.~L. Foster.
\newblock A graph-theoretic generalization of the clique concept.
\newblock {\em Journal of Mathematical sociology}, 6(1):139--154, 1978.

\bibitem{tatti2015density}
N.~Tatti and A.~Gionis.
\newblock Density-friendly graph decomposition.
\newblock In {\em Proceedings of the 24th International Conference on World
  Wide Web}, pages 1089--1099, 2015.

\bibitem{tsourakakis2015k}
C.~Tsourakakis.
\newblock The k-clique densest subgraph problem.
\newblock In {\em WWW}, pages 1122--1132, 2015.

\bibitem{tsourakakis2013denser}
C.~Tsourakakis, F.~Bonchi, A.~Gionis, F.~Gullo, and M.~Tsiarli.
\newblock Denser than the densest subgraph: extracting optimal quasi-cliques
  with quality guarantees.
\newblock In {\em KDD}, pages 104--112. ACM, 2013.

\bibitem{wang2018efficient}
K.~Wang, X.~Cao, X.~Lin, W.~Zhang, and L.~Qin.
\newblock Efficient computing of radius-bounded k-cores.
\newblock In {\em IEEE International Conference on Data Engineering (ICDE)},
  pages 233--244. IEEE, 2018.

\bibitem{wang2018vertex}
K.~Wang, X.~Lin, L.~Qin, W.~Zhang, and Y.~Zhang.
\newblock Vertex priority based butterfly counting for large-scale bipartite
  networks.
\newblock {\em PVLDB}, 12(10):1139--1152, 2019.

\bibitem{wu2015robust}
Y.~Wu, R.~Jin, J.~Li, and X.~Zhang.
\newblock Robust local community detection: on free rider effect and its
  elimination.
\newblock {\em PVLDB}, 8(7):798--809, 2015.

\bibitem{nature2003}
S.~Wuchty, Z.~N. Oltvai, and A.-L. Barab{\'a}si.
\newblock Evolutionary conservation of motif constituents within the yeast
  protein interaction network.
\newblock {\em Nature Genetics}, 35:176--179, 2003.

\bibitem{YZhang12}
Y.~Zhang and S.~Parthasarathy.
\newblock Extracting analyzing and visualizing triangle k-core motifs within
  networks.
\newblock In {\em ICDE}, pages 1049--1060, 2012.

\bibitem{zhao2012large}
F.~Zhao and A.~K. Tung.
\newblock Large scale cohesive subgraphs discovery for social network visual
  analysis.
\newblock {\em PVLDB}, 6(2):85--96, 2012.

\end{thebibliography}
